\pdfoutput=1
\RequirePackage{ifpdf}
\ifpdf % We are running pdfTeX in pdf mode
\documentclass[pdftex]{sigma}
\else
\documentclass{sigma}
\fi

\numberwithin{equation}{section}

\usepackage{braket}

\usepackage{physics}
\usepackage{tikz}
\usetikzlibrary{patterns}
\usepackage{mathtools}

\newcommand{\Hvier}{$H^{4}$}
\newcommand{\Hsechs}{$H^{6}$}
\newcommand{\Z}{\mathbb{Z}}

\newcommand{\N}{\mathbb{N}}
\newcommand{\Cn}{\mathbb{C}}

\DeclareMathOperator{\Id}{Id}

\newcommand{\tHeq}[2][]{${}_{t}H^{\varepsilon #1}_{#2}$}

\newcommand{\Fp}[1]{F^{(+)}_{#1}}
\newcommand{\Fm}[1]{F^{(-)}_{#1}}

\newcommand{\Fppp}{\Fp{n}\Fp{m}}
\newcommand{\Fpmm}{\Fp{n}\Fm{m}}
\newcommand{\Fmpm}{\Fm{n}\Fp{m}}
\newcommand{\Fmmp}{\Fm{n}\Fm{m}}

\renewcommand{\epsilon}{\varepsilon}
\renewcommand{\imath}{\mathrm{i}}

\makeatletter
\renewcommand{\pdv}[2]{\begingroup
 \@tempswafalse\toks@={}\count@=\z@
 \@for\next:=#2\do
 {\expandafter\check@var\next\@nil
 \advance\count@\der@exp
 \if@tempswa
 \toks@=\expandafter{\the\toks@\,}%
 \else
 \@tempswatrue
 \fi
 \toks@=\expandafter{\the\expandafter\toks@\expandafter\partial\der@var}}%
 \frac{\partial\ifnum\count@=\@ne\else^{\number\count@}\fi#1}{\the\toks@}%
 \endgroup}
\def\check@var{\@ifstar{\mult@var}{\one@var}}
\def\mult@var#1#2\@nil{\def\der@var{#2^{#1}}\def\der@exp{#1}}
\def\one@var#1\@nil{\def\der@var{#1}\chardef\der@exp\@ne}
\makeatother

\newtheorem{Theorem}{Theorem}[section]
\newtheorem{Proposition}[Theorem]{Proposition}
 { \theoremstyle{definition}
\newtheorem{Remark}[Theorem]{Remark} }

\begin{document}

\allowdisplaybreaks

\newcommand{\arXivNumber}{1704.05805}

\renewcommand{\thefootnote}{}

\renewcommand{\PaperNumber}{008}

\FirstPageHeading

\ShortArticleName{Darboux Integrability of Trapezoidal $H^{4}$ and $H^{6}$ Families of Lattice Equations II}

\ArticleName{Darboux Integrability of Trapezoidal $\boldsymbol{H^{4}}$ and $\boldsymbol{H^{6}}$\\ Families of Lattice Equations II: General Solutions\footnote{This paper is a~contribution to the Special Issue on Symmetries and Integrability of Dif\/ference Equations. The full collection is available at \href{http://www.emis.de/journals/SIGMA/SIDE12.html}{http://www.emis.de/journals/SIGMA/SIDE12.html}}}

\Author{Giorgio GUBBIOTTI~$^{\dag\ddag}$, Christian SCIMITERNA~$^{\ddag}$ and Ravil I.~YAMILOV~$^{\S}$}

\AuthorNameForHeading{G.~Gubbiotti, C.~Scimiterna and R.I.~Yamilov}

\Address{$^{\dag}$~School of Mathematics and Statistics, F07, The University of Sydney,\\
\hphantom{$^{\dag}$}~New South Wales 2006, Australia}
\EmailD{\href{mailto:giorgio.gubbiotti@sydney.edu.au}{giorgio.gubbiotti@sydney.edu.au}}

\Address{$^{\ddag}$~Dipartimento di Matematica e Fisica, Universit\`a degli Studi Roma Tre and Sezione INFN\\
\hphantom{$^{\ddag}$}~di Roma Tre, Via della Vasca Navale 84, 00146 Roma, Italy}
\EmailD{\href{mailto:gubbiotti@mat.uniroma3.it}{gubbiotti@mat.uniroma3.it}, \href{mailto:scimiterna@fis.uniroma3.it}{scimiterna@fis.uniroma3.it}}

\Address{$^{\S}$~Institute of Mathematics, Ufa Scientific Center, Russian Academy of Sciences,\\
\hphantom{$^{\S}$}~112 Chernyshevsky Str., Ufa 450008, Russia}
\EmailD{\href{mailto:RvlYamilov@matem.anrb.ru}{RvlYamilov@matem.anrb.ru}}

\ArticleDates{Received April 26, 2017, in f\/inal form January 16, 2018; Published online February 02, 2018}

\Abstract{In this paper we construct the general solutions of two families of quad-equations, namely the trapezoidal $H^{4}$ equations and the $H^{6}$ equations. These solutions are obtained exploiting the properties of the f\/irst integrals in the Darboux sense, which were derived in [Gubbiotti~G., Yamilov~R.I., \textit{J.~Phys.~A: Math. Theor.} \textbf{50} (2017), 345205, 26~pages]. These f\/irst integrals are used to reduce the problem to the solution of some linear or linearizable non-autonomous ordinary dif\/ference equations which can be formally solved.}


\Classification{37K10; 37L60; 39A14}

\renewcommand{\thefootnote}{\arabic{footnote}}
\setcounter{footnote}{0}

\section{Introduction}

Since its introduction the integrability criterion denoted consistency around the cube (CAC) has been a source of many results in the classif\/ication of \emph{quad-equations}. We def\/ine a quad-equation to be a relation of the form:
\begin{gather}
 Q ( x,x_{1},x_{2},x_{12} )=0, \label{eq:quadgraphgen}
\end{gather}
where $Q\in \Cn [ x,x_{1},x_{2},x_{12} ]$ is an irreducible multi-af\/f\/ine polynomial. This equation is def\/ined on the four points displayed in Fig.~\ref{fig:geomquad} which form a \emph{square quad graph}.
\begin{figure}[htbp] \centering
 \begin{tikzpicture}
 \node (x1) at (0,0) [circle,fill,label=-135:{$x$}] {};
 \node (x4) at (0,2.5) [circle,fill,label=135:{$x_{1}$}] {};
 \node (x2) at (2.5,0) [circle,fill,label=-45:{$x_{2}$}] {};
 \node (x3) at (2.5,2.5) [circle,fill,label=45:{$x_{12}$}] {};
 \draw [thick] (x2) to (x1);
 \draw [thick] (x4) to (x3);
 \draw [thick] (x3) to (x2);
 \draw [thick] (x1) to (x4);
 \end{tikzpicture}
 \caption{Quad-equation on a square.}\label{fig:geomquad}
\end{figure}
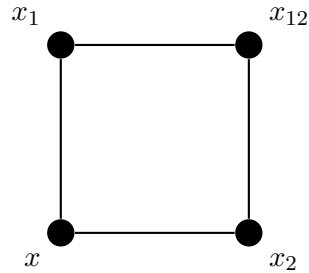

Roughly speaking the CAC approach consist in adding third direction, def\/ined by the label~3 to a quad-equation~\eqref{eq:quadgraphgen}
and extend it to a system of six equations living on the faces of a~cube,
usually labeled $A$, $\bar{A}$, $B$, $\bar{B}$, $C$ and $\bar{C}$, see Fig.~\ref{fig:cube2}.
We say that the system of six equations given by~$A$, $\bar{A}$, $B$, $\bar{B}$, $C$ and $\bar{C}$
possess the Consistency Around the Cube if three ways of computing~$x_{123}$ from $\bar{A}$, $\bar{B}$, and $\bar{C}$ coincide up to
the values of $x_{12}$, $x_{23}$ and $x_{13}$ obtained from~$A$,~$B$ and~$C$ respectively.

The CAC criterion has proved to be important in studying the integrability
properties of quad-equations since from the CAC it is possible to f\/ind
B\"acklund transformations
\cite{BobenkoSuris2002,Bridgman2013,DoliwaSantini1997,Nijhoff2002,Nijhoff2001}
and, as a consequence, Lax pairs.
It is well known~\cite{Yamilov2006} that Lax
pairs and B\"acklund transforms are associated with both linearizable
and integrable equations.
We point out that to be \emph{bona fide}
a Lax pair has to give rise to a genuine
spectral problem \cite{CalogeroDeGasperisIST_I},
otherwise the Lax pair is called \emph{fake Lax pair}
\cite{HayButler2013,HayButler2015, CalogeroNucci1991,Hay2009,Hay2011}.
A fake Lax pair is useless in proving (or disproving)
the integrability, since it can be equally found
for integrable and non-integrable equations.
In the linearizable case Lax pairs must be then fake ones, even though
proving this statement is usually a nontrivial task \cite{GSL_Gallipoli15}.
For a complete, pedagogical explanation of the CAC method
we refer to \cite{Bobenko2008book,HietarintaBook,HietarintaJoshiNijhoff2016}.

\begin{figure}[hbp]
 \centering
 \begin{tikzpicture}[auto,scale=1.0]
 \node (x) at (0,0) [circle,fill,label=-135:$x$] {};
 \node (x1) at (4,0) [circle,fill,label=-45:$x_{1}$] {};
 \node (x2) at (1.5,1.5) [circle,fill,label=-45:$x_{2}$] {};
 \node (x3) at (0,4) [circle,fill,label=135:$x_{3}$] {};
 \node (x12) at (5.5,1.5) [circle,fill,label=-45:$x_{12}$] {};
 \node (x13) at (4,4) [circle,fill,label=-45:$x_{13}$] {};
 \node (x23) at (1.5,5.5) [circle,fill,label=135:$x_{23}$] {};
 \node (x123) at (5.5,5.5) [circle,fill,label=45:$x_{123}$] {};
 \node (A) at (2.75,0.75) {$A$};
 \node (Aq) at (2.75,4.75) {$\bar A$};
 \node (B) at (0.75,2.75) {$B$};
 \node (Bq) at (4.75,2.75) {$\bar B$};
 \node (C) at (2,2) {$C$};
 \node (Cq) at (3.5,3.5) {$\bar C$};
 \draw (x) -- (x1) -- (x12) -- (x123) -- (x23) -- (x3) -- (x);
 \draw (x3) -- (x13) -- (x1);
 \draw (x13) -- (x123);
 \draw [dashed] (x) -- (x2) -- (x12);
 \draw [dashed] (x2) -- (x23);
 \draw [dotted,thick] (A) to (Aq);
 \draw [dotted,thick] (B) to (Bq);
 \draw [dotted,thick] (C) to (Cq);
 \end{tikzpicture}
 \caption{Equations on a cube.} \label{fig:cube2}
\end{figure}
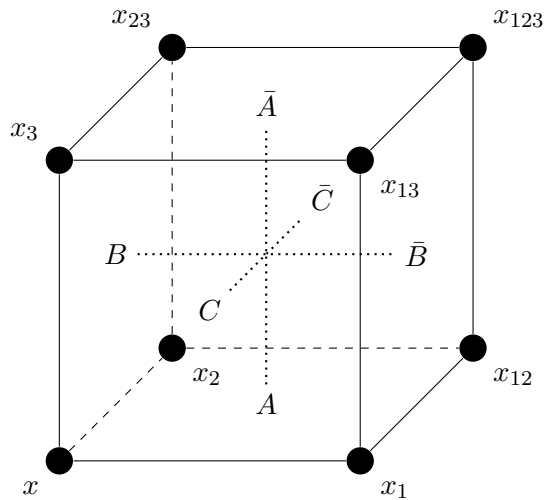

Being algorithmically applicable the CAC criterion proved to be a well suited method to f\/ind and classify integrable quad-equations. The f\/irst attempt to classify, with some additional assumptions, all the quad-equations possessing CAC was carried out in~\cite{ABS2003}. The result was the existence of three classes of discrete autonomous equations with this property: the $H$ equations, the $Q$ equations and the~$A$. The $A$ equations can be transformed in particular cases of the $Q$ equations through non-autonomous M\"obius transformation, therefore they are usually removed from the general classif\/ication. Releasing one of the technical hypothesis of \cite{ABS2003}, i.e., that face of the cube (Fig.~\ref{fig:cube2}) carries the same equation, the same authors in \cite{ABS2009} presented some new equations without classif\/ication purposes. A complete classif\/ication in this extended setting was then accomplished by R.~Boll in a series of papers culminating in his Ph.D.~Thesis \cite{Boll2011,Boll2012a,Boll2012b}. In these papers the classif\/ication of all the consistent sextuples of quad-equations. The only technical assumption used in \cite{Boll2011,Boll2012a,Boll2012b} is the tetrahedron property, i.e., the requirement that $x_{123}$ is independent from~$x$. The obtained equations may fall into three disjoint families depending on their bi-quadratics
\begin{gather*}
 h_{ij}=\pdv{Q}{y_{k}}\pdv{Q}{y_{l}}-Q\pdv{Q}{y_{k},y_{l}}, \qquad Q=Q ( y_{1},y_{2},y_{3},y_{4} ), %\label{eq:biquadr}
\end{gather*}
where we use a special notation for variables of $Q$, and the pair $\{k,l\}$ is the complement of the pair $\{i,j\}$ in $\{ 1,2,3,4\}$. A~bi-quadratic is called \emph{degenerate} if it contains linear factors of the form $y_{i}-c$, where $c$ is a constant, otherwise a bi-quadratic is called \emph{non-degenerate}. The three families are classif\/ied depending on how many bi-quadratics are degenerate:
\begin{itemize}\itemsep=0pt
 \item $Q$-type equations: all the bi-quadratics are non-degenerate,
 \item $H^{4}$-type equations: four bi-quadratics are degenerate,
 \item $H^{6}$-type equations: all of the six bi-quadratics are degenerate.
\end{itemize}
Let us notice that the $Q$ family is the same as the one introduced in~\cite{ABS2003}. The $H^{4}$ equations are divided into two subclasses: \emph{rhombic} and \emph{trapezoidal}, depending on their discrete symmetries. We remark that all classif\/ication results hold locally in the sense that they relate to a single quadrilateral cell or a single cube. The extension on the whole lattice $\Z^{2}$ is obtained through ref\/lection considering an elementary cell of size $2\times2$. This implies that the $H^{4}$ and $H^{6}$ equations as lattice equations are non-autonomous equations with two-periodic coef\/f\/icients. For more details on the construction of equations on the lattice from the single cell equations, we refer to \cite{Boll2011,Boll2012a,Boll2012b, Xenitidis2009} and to the Appendix in~\cite{GSL_general}.

A detailed study of all the lattice equations derived from the \emph{rhombic} $H^{4}$ family, including the construction of their three-leg forms, Lax pairs, B\"acklund transformations and inf\/inite hierarchies of generalized symmetries, has been presented in \cite{Xenitidis2009}. So there was plenty of results about the $Q$ and the rhombic $H^{4}$ equations. On the contrary, besides the CAC property little was known about the integrability features of the \emph{trapezoidal} $H^{4}$ equations and of the $H^{6}$ equations. Therefore these equations where thoroughly studied in a series of papers \cite{GSL_general,GSL_Gallipoli15,GSL_Pavel,GSL_symmetries,GSL_QV} with some unexpected results. First in \cite{GSL_general} was presented their explicit non-autonomous form. Indeed it was shown that on the $\Z^{2}$ lattice with independent variables $(n,m)$ and dependent variable $u_{n,m}$ the trapezoidal $H^{4}$ equations had the following expression
\begin{subequations} \label{eq:trapezoidalH4}
 \begin{gather}
 {}_{t}H_{1}^{\varepsilon} \colon \ (u_{n,m}-u_{n+1,m} ) (u_{n,m+1}-u_{n+1,m+1} ) \nonumber\\
\hphantom{{}_{t}H_{1}^{\varepsilon} \colon \ }{} -\alpha_{2}\epsilon^2\big({F}_{m}^{(+)}u_{n,m+1}u_{n+1,m+1}
 +{F}_{m}^{(-)}u_{n,m}u_{n+1,m}\big) -\alpha_{2}=0, \label{eq:tH1e} \\
 {}_{t}H_{2}^{\varepsilon}\colon \
(u_{n,m}-u_{n+1,m})(u_{n,m+1}-u_{n+1,m+1}) +\alpha_{2} (u_{n,m}+u_{n+1,m}+u_{n,m+1}+u_{n+1,m+1} )\nonumber \\
\hphantom{ {}_{t}H_{2}^{\varepsilon}\colon \ }{} +\frac{\epsilon\alpha_{2}}{2} \big(2{F}_{m}^{(+)}u_{n,m+1}
 +2\alpha_{3}+\alpha_{2}\big)\big(2{F}_{m}^{(+)}u_{n+1,m+1}+2\alpha_{3}+\alpha_{2}\big)\nonumber \\
\hphantom{ {}_{t}H_{2}^{\varepsilon}\colon \ }{}
+\frac{\epsilon\alpha_{2}}{2} \big(2{F}_{m}^{(-)}u_{n,m}+2\alpha_{3}
 +\alpha_{2}\big)\big(2{F}_{m}^{(-)}u_{n+1,m}+2\alpha_{3}+\alpha_{2}\big)\nonumber \\
 \hphantom{ {}_{t}H_{2}^{\varepsilon}\colon \ }{}
 + (\alpha_{3}+\alpha_{2} )^2-\alpha_{3}^2-2\epsilon\alpha_{2}\alpha_{3} (\alpha_{3}+\alpha_{2} )=0, \label{eq:tH2e} \\
{} _{t}H_{3}^{\varepsilon}\colon \
\alpha_{2}(u_{n,m}u_{n+1,m+1}+u_{n+1,m}u_{n,m+1}) \nonumber \\
\hphantom{{} _{t}H_{3}^{\varepsilon}\colon \ }{} - (u_{n,m}u_{n,m+1}+u_{n+1,m}u_{n+1,m+1} ) -\alpha_{3}\big(\alpha_{2}^{2}-1\big)\delta^2\nonumber \\
\hphantom{{} _{t}H_{3}^{\varepsilon}\colon \ }{} -\frac{\epsilon^2(\alpha_{2}^{2}-1)}{\alpha_{3}\alpha_{2}}
 \big({{F}_{m}^{(+)}u_{n,m+1}u_{n+1,m+1} +{F}_{m}^{(-)}u_{n,m}u_{n+1,m}}\big)=0, \label{eq:tH3e}
 \end{gather}
\end{subequations}
and the \Hsechs\ equations had the following expression
\begin{subequations} \label{eq:h6}
 \begin{gather}
{}_{1}D_{2} \colon \ \big( F_{n+m}^{(-)}-\delta_{1} F_{n}^{(+)} F_{m}^{(-)}+\delta_{2} F_{n}^{(+)} F_{m}^{(+)}\big)u_{n,m}\nonumber\\
\hphantom{{}_{1}D_{2} \colon \ }{} +\big( F_{n+m}^{(+)}-\delta_{1} F_{n}^{(-)} F_{m}^{(-)}+\delta_{2} F_{n}^{(-)} F_{m}^{(+)}\big)u_{n+1,m}\nonumber \\
\hphantom{{}_{1}D_{2} \colon \ }{}+\big( F_{n+m}^{(+)}-\delta_{1} F_{n}^{(+)} F_{m}^{(+)}+\delta_{2} F_{n}^{(+)} F_{m}^{(-)}\big)u_{n,m+1}\nonumber \\
\hphantom{{}_{1}D_{2} \colon \ }{}+\big( F_{n+m}^{(-)}-\delta_{1} F_{n}^{(-)} F_{m}^{(+)}+\delta_{2} F_{n}^{(-)} F_{m}^{(-)}\big)u_{n+1,m+1}\nonumber \\
\hphantom{{}_{1}D_{2} \colon \ }{}+\delta_{1}\big( F_{m}^{(-)}u_{n,m}u_{n+1,m}+ F_{m}^{(+)}u_{n,m+1}u_{n+1,m+1}\big)\nonumber \\
 \hphantom{{}_{1}D_{2} \colon \ }{}+ F_{n+m}^{(+)}u_{n,m}u_{n+1,m+1} + F_{n+m}^{(-)}u_{n+1,m}u_{n,m+1}=0, \label{eq:1D2} \\
 {}_{2}D_{2} \colon \
\big(F_{m}^{(-)}-\delta_{1}F_{n}^{(+)}F_{m}^{(-)}+\delta_{2}F_{n}^{(+)}F_{m}^{(+)}-\delta_{1} \lambda F_{n}^{(-)}F_{m}^{(+)}\big)u_{n,m}\nonumber \\
\hphantom{{}_{2}D_{2} \colon \ }{}
+\big(F_{m}^{(-)}-\delta_{1}F_{n}^{(-)}F_{m}^{(-)}+\delta_{2}F_{n}^{(-)}F_{m}^{(+)}-\delta_{1} \lambda F_{n}^{(+)}F_{m}^{(+)}\big)u_{n+1,m}\nonumber \\
\hphantom{{}_{2}D_{2} \colon \ }{}
+\big(F_{m}^{(+)}-\delta_{1}F_{n}^{(+)}F_{m}^{(+)}+\delta_{2}F_{n}^{(+)}F_{m}^{(-)}-\delta_{1} \lambda F_{n}^{(-)}F_{m}^{(-)}\big)u_{n,m+1}\nonumber \\
\hphantom{{}_{2}D_{2} \colon \ }{}
+\big(F_{m}^{(+)}-\delta_{1}F_{n}^{(-)}F_{m}^{(+)}+\delta_{2}F_{n}^{(-)}F_{m}^{(-)}-\delta_{1} \lambda F_{n}^{(+)}F_{m}^{(-)}\big)u_{n+1,m+1}\nonumber \\
\hphantom{{}_{2}D_{2} \colon \ }{}
+\delta_{1}\big(F_{n+m}^{(-)}u_{n,m}u_{n+1,m+1}+F_{n+m}^{(+)}u_{n+1,m}u_{n,m+1}\big)\nonumber \\
\hphantom{{}_{2}D_{2} \colon \ }{}
+F_{m}^{(+)}u_{n,m}u_{n+1,m}+F_{m}^{(-)}u_{n,m+1}u_{n+1,m+1} -\delta_{1}\delta_{2}\lambda=0, \label{eq:2D2} \\
{}_{3}D_{2} \colon \
\big(F_{m}^{(-)}-\delta_{1}F_{n}^{(-)}F_{m}^{(-)}+\delta_{2}F_{n}^{(+)}F_{m}^{(+)}-\delta_{1} \lambda F_{n}^{(-)}F_{m}^{(+)}\big)u_{n,m}\nonumber \\
\hphantom{{}_{3}D_{2} \colon \ }{}
+\big(F_{m}^{(-)}-\delta_{1}F_{n}^{(+)}F_{m}^{(-)}+\delta_{2}F_{n}^{(-)}F_{m}^{(+)}-\delta_{1} \lambda F_{n}^{(+)}F_{m}^{(+)}\big)u_{n+1,m}\nonumber \\
\hphantom{{}_{3}D_{2} \colon \ }{}
+\big(F_{m}^{(+)}-\delta_{1}F_{n}^{(-)}F_{m}^{(+)}+\delta_{2}F_{n}^{(+)}F_{m}^{(-)}-\delta_{1} \lambda F_{n}^{(-)}F_{m}^{(-)}\big)u_{n,m+1}\nonumber \\
\hphantom{{}_{3}D_{2} \colon \ }{}
+\big(F_{m}^{(+)}-\delta_{1}F_{n}^{(+)}F_{m}^{(+)}+\delta_{2}F_{n}^{(-)}F_{m}^{(-)}-\delta_{1} \lambda F_{n}^{(+)}F_{m}^{(-)}\big)u_{n+1,m+1}\nonumber \\
\hphantom{{}_{3}D_{2} \colon \ }{}
+\delta_{1}\big(F_{n}^{(-)}u_{n,m}u_{n,m+1}+F_{n}^{(+)}u_{n+1,m}u_{n+1,m+1}\big)\nonumber \\
\hphantom{{}_{3}D_{2} \colon \ }{}
+F_{m}^{(-)}u_{n,m+1}u_{n+1,m+1} +F_{m}^{(+)}u_{n,m}u_{n+1,m}-\delta_{1}\delta_{2}\lambda=0, \label{eq:3D2} \\
 D_{3} \colon \
F_{n}^{(+)}F_{m}^{(+)}u_{n,m}+F_{n}^{(-)}F_{m}^{(+)}u_{n+1,m} +F_{n}^{(+)}F_{m}^{(-)}u_{n,m+1} \nonumber \\
\hphantom{D_{3} \colon \ }{} +F_{n}^{(-)}F_{m}^{(-)}u_{n+1,m+1} +F_{m}^{(-)}u_{n,m}u_{n+1,m} +F_{n}^{(-)}u_{n,m}u_{n,m+1}+F_{n+m}^{(-)}u_{n,m}u_{n+1,m+1}\nonumber \\
\hphantom{D_{3} \colon \ }{}+F_{n+m}^{(+)}u_{n+1,m}u_{n,m+1}+F_{n}^{(+)}u_{n+1,m}u_{n+1,m+1}+F_{m}^{(+)}u_{n,m+1}u_{n+1,m+1}=0, \label{eq:D3} \\
 {}_{1}D_{4} \colon \
\delta_{1}\big(F_{n}^{(-)}u_{n,m}u_{n,m+1}+F_{n}^{(+)}u_{n+1,m}u_{n+1,m+1}\big)\nonumber\\
\hphantom{{}_{1}D_{4} \colon \ }{} +\delta_{2}\big(F_{m}^{(-)}u_{n,m}u_{n+1,m}+F_{m}^{(+)}u_{n,m+1}u_{n+1,m+1}\big)\nonumber\\
\hphantom{{}_{1}D_{4} \colon \ }{} +u_{n,m}u_{n+1,m+1}+u_{n+1,m}u_{n,m+1}+\delta_{3}=0, \label{eq:1D4} \\
 {} _{2}D_{4} \colon \
\delta_{1}\big(F_{n}^{(-)}u_{n,m}u_{n,m+1} +F_{n}^{(+)}u_{n+1,m}u_{n+1,m+1}\big) \nonumber \\
\hphantom{{} _{2}D_{4} \colon \ }{} +\delta_{2}\big(F_{n+m}^{(-)}u_{n,m}u_{n+1,m+1} +F_{n+m}^{(+)}u_{n+1,m}u_{n,m+1}\big)\nonumber \\
\hphantom{{} _{2}D_{4} \colon \ }{} +u_{n,m}u_{n+1,m}+u_{n,m+1}u_{n+1,m+1}+\delta_{3}=0, \label{eq:2D4}
 \end{gather}
\end{subequations}
where the coef\/f\/icients $F_{k}^{(\pm)}$ are given by
\begin{gather}
 F_{k}^{(\pm)} = \frac{1\pm (-1)^{k}}{2}. \label{eq:fk}
\end{gather}
Then in \cite{GSL_general} the algebraic entropy \cite{BellonViallet1999,HietarintaViallet2007,Viallet2006,Viallet2009} of the trapezoidal \Hvier and the \Hsechs~equations was computed. The result of this computation showed that the \emph{rate of growth} of all the trapezoidal $H^{4}$ \eqref{eq:trapezoidalH4}
and of all $H^{6}$ equations \eqref{eq:h6} is \emph{linear}. This fact according to the \emph{algebraic entropy conjecture} \cite{FalquiViallet1993,HietarintaViallet2007} implies linearizability. To support this result two explicit examples of linearization were given.

\begin{Remark}
 In \cite{HietarintaViallet2012} it was shown that sometimes it is possible
 to construct dif\/ferent consistent embedding in the $\Z^{2}$
 and in $\Z^{3}$ lattices.
 However, in the same paper it was shown that these dif\/ferent
 embedding need not to be integrable.
 In this paper we will consider equations \eqref{eq:trapezoidalH4}
 and \eqref{eq:h6} which are given by the embedding procedure
 of \cite{ABS2009,Boll2011,Boll2012a,Boll2012b}.
 As we underlined above this procedure gives equations which,
 in the sense of the algebraic entropy, are only integrable
 or linearizable \cite{GSL_general,RobertsTran2017}.
 Clearly, it may exist a dif\/ferent embedding in the $\Z^{2}$ for
 which the results presented in this paper do not hold.
 For an example where two dif\/ferent embedding give both rise
 to linearizable equations, but with dif\/ferent properties, see~\cite{GubScimSIDE12}.
\end{Remark}

In \cite{GSL_Gallipoli15} the \tHeq{1} equation \eqref{eq:tH1e} was studied and it was found that it possessed three-point generalized symmetries depending on arbitrary functions. This property was then linked in \cite{GSL_Pavel} to the fact that the \tHeq{1} is \emph{Darboux integrable}~\cite{AdlerStartsev1999}. We say that a quad-equation on the $\Z^{2}$ lattice, possibly non-autonomous:
\begin{gather}
 Q_{n,m}\left( u_{n,m},u_{n+1,m},u_{n,m+1},u_{n+1,m+1} \right)=0,
 \label{eq:quadequana}
\end{gather}
is \emph{Darboux integrable} if there exist two independent \emph{first integrals}, one containing only shifts in the f\/irst direction and the other containing only shifts in the second direction. This means that there exist two functions
%\begin{subequations} \label{eq:darbfint}
 \begin{gather*}
 W_1=W_{1,n,m}(u_{n+l_1,m},u_{n+l_1+1,m},\ldots,u_{n+k_1,m}),\\ %\label{eq:darbfint1} \\
 W_2=W_{2,n,m}(u_{n,m+l_2},u_{n,m+l_2+1},\ldots,u_{n,m+k_2}), %\label{eq:darbfint2}
 \end{gather*}
%\end{subequations}
where $l_1<k_1$ and $l_2<k_2$ are integers, such that the relations
\begin{subequations} \label{eq:darbdef}
 \begin{gather}
 (T_n-\Id)W_2=0, \label{eq:darb1} \\
 (T_m-\Id)W_1=0 \label{eq:darb2}
 \end{gather}
\end{subequations}
hold true identically on the solutions of~\eqref{eq:quadequana}. By $T_n$, $T_m$ we denote the shift operators in the f\/irst and second
directions, i.e., $T_n h_{n,m}=h_{n+1,m}$, $T_m h_{n,m}=h_{n,m+1}$, and by $\Id$ we denote the identity operator $\Id h_{n,m}=h_{n,m}$. The number $k_{i}-l_{i}$, where $i=1,2$, is called the \emph{order} of the f\/irst integral $W_{i}$.

In addition to this result concerning the \tHeq{1} equation in \cite{GSL_Pavel} it was proved that other quad-equations consistent around the cube, which were known to be linearizable \cite{Hietarinta2004,Hietarinta2005}, were in fact Darboux integrable. These facts provide some evidence of an intimate connection between linearizable equations possessing CAC and Darboux integrability. Following these ideas in~\cite{GSY_DarbouxI} it was shown that \emph{all} the trapezoidal $H^{4}$ equations and \emph{all} the $H^{6}$ equations are Darboux integrable. This result was proved by explicitly constructing the f\/irst integrals with a new algorithm based on those proposed in \cite{GarifullinYamilov2012,GarifullinYamilov2015, Habibullin2005}. This new algorithm relies on the fact that in the case of non-autonomous quad-equations \eqref{eq:quadequana} with two-periodic coef\/f\/icients we can, in general, represent the f\/irst integrals in the form
\begin{gather*}
 W_{i} = \Fppp W_{i}^{(+,+)}+\Fmpm W_{i}^{(-,+)}+\Fpmm W_{i}^{(+,-)} + \Fmmp W_{i}^{(-,-)}, %\label{eq:widec}
\end{gather*}
where $F_{k}^{(\pm)}$ are given by \eqref{eq:fk} and the $W_{i}^{(\pm,\pm)}$ are functions. The existence of the f\/irst integrals provides a rigorous
proof of the linearizability of the trapezoidal $H^{4}$ equation \eqref{eq:trapezoidalH4} and of the~$H^{6}$ equations~\eqref{eq:h6}. Indeed equation~\eqref{eq:darbdef} implies that the following two transformations
\begin{subequations}\label{eq:fintlin}
 \begin{gather}
 u_{n,m} \to \tilde u_{n,m} = W_{1,n,m}, \label{eq:fintlinA} \\
 u_{n,m} \to \hat u_{n,m} = W_{2,n,m}, \label{eq:fintlinB}
 \end{gather}
\end{subequations}
bring the quad-equation \eqref{eq:quadequana} into two \emph{trivial linear equations}
\begin{subequations} \label{eq:fintlin2}
 \begin{gather}
 \tilde u_{n,m+1} - \tilde u_{n,m} = 0 , \label{eq:fintlin2A} \\
 \hat u_{n+1,m} - \hat u_{n,m} = 0. \label{eq:fintlin2B}
 \end{gather}
\end{subequations}
Therefore any Darboux integrable equation is \emph{linearizable in two different ways}, i.e., using transformation \eqref{eq:fintlinA} bringing to~\eqref{eq:fintlin2A} or using the transformation~\eqref{eq:fintlinB} bringing to~\eqref{eq:fintlin2B}.

In the f\/inal section of \cite{GSY_DarbouxI} it was shown, in the case of the \tHeq{1} equation, how it is possible to f\/ind the general solution using the f\/irst integrals, applying a modif\/ication of the procedure presented in~\cite{GarifullinYamilov2012}. In particular we showed how it is possible to obtain a general solution using f\/irst integrals of order greater than one. We note that equations with f\/irst integrals of f\/irst order one are \emph{trivial}, since possessing a f\/irst integral of order one means that \emph{the equation itself is a first integral}.

In this paper we show that from the knowledge of the f\/irst integrals and from the properties of the equations it is possible to construct, maybe after some complicate algebra, the general solutions of all the remaining trapezoidal $H^4$ equations \eqref{eq:trapezoidalH4} and of the $H^6$ equations \eqref{eq:h6}.
By general solution we mean a representation of the solution of any of the equations in~\eqref{eq:trapezoidalH4} and~\eqref{eq:h6} in terms of the right number of arbitrary functions of one lattice variable $n$ or $m$. Since the trapezoidal $H^{4}$ equations \eqref{eq:trapezoidalH4} and the $H^{6}$ equations~\eqref{eq:h6} are quad-equations, i.e., the discrete analogue of second-order hyperbolic partial dif\/ferential equations, the general solution must contain an arbitrary function in the $n$ direction and another one in the $m$ direction, i.e., a general solution is an expression of the form
\begin{gather}
 u_{n,m} = F_{n,m} ( a_{n},\alpha_{m} ), \label{eq:gensoldef}
\end{gather}
where $a_{n}$ and $\alpha_{m}$ are arbitrary functions of their variable. Initial conditions are then imposed through substitution in equation~\eqref{eq:gensoldef}. Nonlinear equations usually possesses also other kinds of solutions, namely the singular solutions which satisfy only specif\/ic set of initial values. In this work we outlined when the existence of singular solutions is possible. Moreover we remark that general solutions, in the range of validity of their parameters, enclose also periodic solutions. Periodic initial values will ref\/lect into periodic solution which will arise by f\/ixing properly the arbitrary functions. Indeed let us consider as an example the $(N,-M)$ reduction of a quad-equation \eqref{eq:quadequana}, with $N,M\in\N^{*}$ coprime~\cite{PapageorgiouNijhoff1990,QuispelCapel1991}. This implies to make the following requirement
\begin{gather}
 u_{n+N,m-M} = u_{n,m}. \label{eq:periodicsol}
\end{gather}
If we possess the general solution of the quad-equation in the form~\eqref{eq:gensoldef} then the periodicity requirement~\eqref{eq:periodicsol} is equivalent to
\begin{gather}
 F_{n+N,m-M} ( a_{n+N},\alpha_{m-M} )= F_{n,m}( a_{n},\alpha_{m}). \label{eq:gensolper}
\end{gather}
The existence of the associated periodic solution is subject to the ability to invert formula \eqref{eq:gensolper}. When the integers $N$ and $M$ are not coprime a similar reasoning can be done: taking $K = \gcd(N,M)$ we have just to decompose the reduction condition into $K$ superimposed staircases and convert the scalar condition \eqref{eq:periodicsol} to a vector condition for $K$ f\/ields. The associated reduction will be possible if the associated system possesses a solution.

To obtain the desired solution we will need only the $W_{1}$ integrals derived in~\cite{GSY_DarbouxI} and the fact that the relation \eqref{eq:darb2} implies $W_{1}=\xi_{n}$ with $\xi_{n}$ an arbitrary function of $n$. The equation $W_{1}=\xi_{n}$ can be interpreted as an ordinary dif\/ference equation in the $n$ direction depending parametrically on $m$. Then from every $W_{1}$ integral we can derive two dif\/ferent ordinary dif\/ference equations, one corresponding to $m$ even and one corresponding to $m$ odd. In both the resulting equations we can get rid of the two-periodic terms by considering the cases $n$ even and $n$ odd and def\/ining
\begin{subequations} \label{eq:genautsub}
 \begin{alignat}{3}
 &u_{2k,2l} = v_{k,l},\qquad &&u_{2k+1,2l} = w_{k,l},& \label{eq:genautsubl} \\
 & u_{2k,2l+1} = y_{k,l},\qquad && u_{2k+1,2l+1} = z_{k,l}.& \label{eq:genautsublp}
 \end{alignat}
\end{subequations}
This transformation brings both equations to a \emph{system} of \emph{coupled difference equations}. This reduction to a system is the key ingredient in the construction of the general solutions for the trapezoidal \Hvier\ equations \eqref{eq:trapezoidalH4} and for the \Hsechs\ equations \eqref{eq:h6}.

We note that the transformation \eqref{eq:genautsub} can be applied to the trapezoidal \Hvier~equations\footnote{In fact, in the case of the trapezoidal $H^4$ equations~\eqref{eq:trapezoidalH4}, we use a simpler transformation instead of~\eqref{eq:genautsub}, see Section~\ref{sec:trapH4}.} and~\Hsechs equations themselves. This casts these \emph{non-autonomous equations with two-periodic coefficients} into \emph{autonomous systems of four equations}. We recall that in this way some examples of direct linearization (i.e., without the knowledge of the f\/irst integrals) were produced in~\cite{GSL_general}. Finally we note that if we apply the even/odd splitting of the lattice variables given by equation~\eqref{eq:genautsub} to describe a general solution we will need two arbitrary functions in both directions, i.e., we will need a total of four arbitrary functions.

In practice to construct these general solutions, we need to solve Riccati equations and non-autonomous linear equations which, in general, cannot be solved in closed form. Using the fact that these equations contain arbitrary functions we introduce new arbitrary functions so that we can solve these equations. This is usually done reducing to \emph{total difference}, i.e., to ordinary dif\/ference equations which can be \emph{trivially} solved. Let us assume we are given the dif\/ference equation
\begin{gather}
 u_{n+1,m} - u_{n,m} = f_{n}, \label{eq:tdifprot}
\end{gather}
depending parametrically on another discrete index $m$. Then if we can express the function $f_{n}$ as a \emph{discrete derivative}
\begin{gather*}
 f_{n} = g_{n+1}-g_{n}, %\label{eq:discrder}
\end{gather*}
then the solution of equation \eqref{eq:tdifprot} is simply
\begin{gather*}
 u_{n,m} = g_{n} + \gamma_{m}, %\label{eq:tdifsol}
\end{gather*}
where $\gamma_{m}$ is an arbitrary function of the discrete variable $m$. This is the simplest possible example of reduction to total dif\/ference. The general solutions will then be expressed in terms of these new arbitrary functions obtained reducing to total dif\/ferences and in terms of a f\/inite number of \emph{discrete integrations}, i.e., the solutions of the simple ordinary dif\/ference equation
\begin{gather}
 u_{n+1}-u_{n} = f_{n}, \label{eq:discrintdef}
\end{gather}
where $u_{n}$ is the unknown and $f_{n}$ is an assigned function. We note that the discrete integra\-tion~\eqref{eq:discrintdef} is the discrete analogue of the
dif\/ferential equation $u'(x) = f(x)$.

To give a very simple example of the method of solution we consider how it applies to the prototypical Darboux integrable equation: \emph{the discrete wave equation}
\begin{gather}
 u_{n+1,m+1}+u_{n,m}=u_{n+1,m}+u_{n,m+1}. \label{eq:dwave}
\end{gather}
It is easy to check that the discrete wave equation \eqref{eq:dwave} is Darboux integrable with two f\/irst-order f\/irst integrals
\begin{subequations} \label{eq:fintdwave}
 \begin{gather}
 W_{1} = u_{n+1,m}-u_{n,m}, \label{eq:W1dwave} \\
 W_{2} = u_{n,m+1}-u_{n,m}. \label{eq:W2dwave}
 \end{gather}
\end{subequations}
From the f\/irst integrals \eqref{eq:fintdwave} it is possible to construct the well known discrete d'Alembert solution as follows. From the $W_{1}$ f\/irst integral we can write $W_{1}=\xi_{n}$ with $\xi_{n}$ arbitrary function of its argument. Then we have
\begin{gather}
 u_{n+1,m}-u_{n,m} = \xi_{n}. \label{eq:dwavesol0}
\end{gather}
This means that choosing the arbitrary function as $\xi_{n} = a_{n+1}-a_{n}$, with $a_{n}$ arbitrary function of its argument, we transform \eqref{eq:dwavesol0} into the total dif\/ference
\begin{gather*}
 u_{n+1,m}+a_{n+1} =u_{n,m} + a_{n}, %\label{eq:dwavesol1}
\end{gather*}
which readily implies
\begin{gather*}
 u_{n,m}=a_{n} +\alpha_{m}, %\label{eq:dwavesol2}
\end{gather*}
where $\alpha_{m}$ is an arbitrary function of its argument. This is of course the discrete analog of the d'Alembert solution of the wave equation and it is the simplest example of solution through the f\/irst integrals of a Darboux integrable equation.

Now to summarize, in this paper we prove the following result:
\begin{Theorem} \label{thm:sol}
 The trapezoidal $H^4$ equations \eqref{eq:trapezoidalH4} and $H^6$ equations \eqref{eq:h6} are exactly solvable and we can represent the solution in terms of a finite number of discrete integration \eqref{eq:discrintdef}.
\end{Theorem}

The rest of the paper is devoted to the proof of Theorem~\ref{thm:sol}. In Section~\ref{sec:solutions} we present the general solutions of all the $H^{4}$ and $H^{6}$ equations, except the $_{t}H_{1}^{\varepsilon}$ equation~\eqref{eq:tH1e} which was treated in \cite{GSY_DarbouxI}. In particular in Section~\ref{sec:iD2} we treat the $_{1}D_{2}$, $_{2}D_{2}$ and $_{3}D_{2}$ equations. In Section~\ref{sec:otherH6} we treat the $D_{3}$, $_{1}D_{4}$
and $_{2}D_{4}$ equations. In Section~\ref{sec:trapH4} we treat the \tHeq{2}~and the~\tHeq{3} equations. The partition in subsection is dictated by the procedure used to obtain the general solution, as we will explain below. Due to the technical nature of the procedures we will present only one example per type. The interested reader will f\/ind the remaining procedures of solution in Appendix~\ref{app:remsols}. In Section~\ref{sec:concl} we give some conclusions.

\begin{Remark} We remark that the $H$ equations of the ABS classif\/ication \cite{ABS2003} and their rhombic deformations \cite{ABS2009,Boll2011,Xenitidis2009} should not be Darboux integrable. This can be conf\/irmed directly excluding the existence of integrals up to a~certain order as it was done in \cite{GarifullinYamilov2012} for some other equations. Moreover it was proved rigorously in \cite{RobertsTran2017} using the $\gcd$-factorization method that all the equations of the ABS list \cite{ABS2003} possess quadratic growth of the degrees. At heuristic level a~similar result was presented in \cite{GSL_general} for the rhombic $H^{4}$ equations. According to the Algebraic entropy conjecture these result means that the ABS equations and the rhombic $H^{4}$ equations are integrable, but not linearizable. Since we have recalled the fact that Darboux integrability for lattice equations implies linearizability we expect that these equations will not possess f\/irst integrals of any order. So the results obtained in \cite{GSY_DarbouxI} and in this paper about the trapezoidal~$H^4$ and~$H^6$ equations do not imply anything for~$H$ equations and their rhombic deformations.
\end{Remark}

\section[General solutions of the $H^{4}$ and $H^{6}$ equations]{General solutions of the $\boldsymbol{H^{4}}$ and $\boldsymbol{H^{6}}$ equations}\label{sec:solutions}

In this section we present the general solutions of the $H^{4}$ equations \eqref{eq:trapezoidalH4} and of the $H^{6}$ equations~\eqref{eq:h6}. We choose to divide this section in three subsections since we have three main dif\/ferent kinds of procedures leading to three dif\/ferent representations of the solution.

First in Section~\ref{sec:iD2} we present the general solutions of the $_{1}D_{2}$, $_{2}D_{2}$ and $_{3}D_{2}$ equations (\ref{eq:1D2})--(\ref{eq:3D2}). In this case the construction of the general solution is carried out from the sole knowledge of the f\/irst integral and the equation acts only as a compatibility condition for the arbitrary functions obtained by solving the equations def\/ined by the f\/irst integral. The solution therein obtained is completely explicit and no discrete integration is required.

The in Section~\ref{sec:otherH6} we present the general solution of the $D_{3}$, $_{1}D_{4}$ and $_{2}D_{4}$ equations (\ref{eq:D3})--(\ref{eq:2D4}). In this case the construction of the general solution is carried out through a series of manipulations in the equation itself and from the knowledge of the f\/irst integral. The key point will be that the equations def\/ined by the f\/irst integrals can be reduced to a single linear equation. The solution is no longer completely explicit since it is obtained up to two discrete integrations, one in every direction.

Finally in Section~\ref{sec:trapH4} we present the general solutions of the~\tHeq{2}\ equation \eqref{eq:tH2e} and of the~\tHeq{3}\ equation~\eqref{eq:tH3e}. In this case the construction of the general solution is carried out reducing the equation to a partial dif\/ference equation def\/ined on six points and then using the equations def\/ined by the f\/irst integrals. The equations def\/ined by the f\/irst integrals are reduced to discrete Riccati equations. The solution is then given in terms of four discrete integrations.

Summing up these results and the fact that the general solution of the \tHeq{1} equation \eqref{eq:tH1e} was presented in \cite{GSY_DarbouxI}
we prove Theorem~\ref{thm:sol}.

\subsection[The $_{i}D_{2}$ equations $i=1,2,3$]{The $\boldsymbol{{}_{i}D_{2}}$ equations $\boldsymbol{i=1,2,3}$}\label{sec:iD2}

We have that the following propositions hold true:
\begin{Proposition} \label{prop:1D2}
 The $_{1}D_{2}$ equation \eqref{eq:1D2} is exactly solvable. If $\delta_{1}\neq0$ and $\delta\neq0$, where $\delta$ is defined by
 \begin{gather}
 \delta= 1-\delta_{1}( 1+\delta_{2} ) \label{eq:deltaconstdef}
 \end{gather}
 the general solution is given by
 \begin{subequations}\label{eq:1D2solfin}
 \begin{gather}
 v_{k,l}= \alpha_{l} - \frac{\delta_1}{\delta} \frac{b_{k} b_{k-1} }{b_{k}-b_{k-1}} (c_{k}-c_{k-1}), \label{eq:1D2vsolfin} \\
 w_{k,l} = b_{k}(\beta_{l} +c_{k}) + \frac{\delta}{\delta_{1}}\alpha_{l}, \label{eq:1D2wsolfin} \\
 z_{k,l} = 1-\frac{1}{\delta_{1}}- b_{k}\frac{\beta_{l+1}-\beta_{l}}{\alpha_{l+1}-\alpha_{l}}, \label{eq:1D2zklsolfin} \\
 y_{k,l} = \frac{1}{\delta_{1}} \frac{\beta_{l}\alpha_{l+1}-\beta_{l+1}\alpha_{l}}{\beta_{l+1}-\beta_{l}} +\frac{1}{\delta}
 \frac{b_{k} b_{k-1} \left(c_{k}-c_{k-1}\right) }{b_{k}-b_{k-1}}\nonumber \\
 \hphantom{y_{k,l} =}{} +\left[\frac{(c_{k}-c_{k-1})b_{k-1}}{(b_{k}-b_{k-1})\delta_{1}} +\frac{c_{k}}{\delta_{1}}\right]
 \frac{\alpha_{l+1}-\alpha_{l}}{\beta_{l+1}-\beta_{l}}, \label{eq:1D2yklsolfin}
 \end{gather}
 \end{subequations}
 where $b_{k}$, $c_{k}$, $\alpha_{l}$ and $\beta_{l}$ are arbitrary functions of their arguments. If $\delta=0$ its general solution is given by
 \begin{subequations} \label{eq:1D2d1sol}
 \begin{gather}
v_{k, l} = a_{k}+\beta_{l}, \label{eq:1D2vd1solfin} \\
 w_{k, l} = b_{k}\alpha_{l}, \label{eq:1D2wd1solfin} \\
 z_{k, l}= -b_{k}\frac{\alpha_{l+1}-\alpha_{l}}{\beta_{l+1}-\beta_{l}}-\delta_{2}, \label{eq:1D2zd1solfin} \\
 y_{k, l} = -( 1+\delta_{2} ) \left(\frac{\beta_{l}\alpha_{l+1}-\beta_{l+1}\alpha_{l}}{ \alpha_{l+1}-\alpha_{l}}+a_{k}\right), \label{eq:1D2yd1solfin}
 \end{gather}
 \end{subequations}
 where $a_{k}$, $b_{k}$, $\alpha_{l}$ and $\beta_{l}$ are arbitrary functions of their arguments. If $\delta_{1}=0$ then its general solution is given by
 \begin{subequations} \label{eq:1D20d2sol}
 \begin{gather}
 v_{k,l} = \frac{c_{k}-c_{k-1}}{b_{k}-b_{k-1}} +\alpha_{l}, \label{eq:1D20d2vsolfin} \\
 w_{k,l} = c_{k} + \alpha_{l}b_{k} + \beta_{l}, \label{eq:1D20d2wsolfin} \\
 z_{k,l} = -\delta_2-\frac{\beta_{l+1}-\beta_{l}}{\alpha_{l+1}-\alpha_{l}} - b_{k}, \label{eq:1D20d2zsolfin} \\
 y_{k,l} = \frac{\beta_{l+1} \alpha_{l} -\alpha_{l+1} \beta_{l}}{\alpha_{l+1}-\alpha_{l}} +\frac{\beta_{l+1}-\beta_{l}}{\alpha_{l+1}-\alpha_{l}}
 \frac{c_{k}-c_{k-1}}{b_{k}-b_{k-1}} +b_{k}\frac{c_{k}-c_{k-1}}{b_{k}-b_{k-1}}-c_{k},
 \label{eq:1D20d2ysolfin}
 \end{gather}
 \end{subequations}
 where $b_{k}$, $c_{k}$, $\alpha_{l}$ and $\beta_{l}$ are arbitrary functions of their arguments.
\end{Proposition}

\begin{proof} From \cite{GSY_DarbouxI} we know that the $_{1}D_{2}$ equation \eqref{eq:1D2} is Darboux integrable, and that the form of the f\/irst integral depends on the value of the parameter $\delta_{1}$. We will begin with the general case when $\delta_{1}\neq0$ and $\delta\neq0$ and then consider
 the particular cases.

{\bf Case $\boldsymbol{\delta\neq0}$ and $\boldsymbol{\delta_{1}\neq0}$.} In this case the $W_{1}$ f\/irst integral of the $_{1}D_{2}$ equation \eqref{eq:1D2} is given by~\cite{GSY_DarbouxI}
 \begin{gather}
 W_{1} ={ \Fppp}{ \alpha} \frac {[ ( 1+{ \delta_{2}}) u_{{n,m}}+u_{{n+1,m}}] {
 \delta_{1}}-u_{{n,m}}}{{ [ ( 1+{ \delta_{2}} )u_{{n,m}}+u_{{n-1,m}} ] { \delta_{1}}-u_{{n,m}}}}\nonumber \\
 \hphantom{W_{1} =}{} +{ \Fpmm}{ \alpha} \frac{1+ ( u_{{n+1,m}}-1 ) { \delta_{1}}}{ 1+ ( u_{{n-1,m}}-1 ) { \delta_{1}}}
 +{ \Fmpm}{ \beta} ( u_{{n+1,m}} -u_{{n-1,m}})\nonumber \\
 \hphantom{W_{1} =}{}-{ \Fmmp}{ \beta}{\frac { ( u_{{n+1,m}} -u_{{n-1,m}})
[ 1- ( 1-u_{{n,m}} ) { \delta_{1}} ] }{{ \delta_{2}}+u_{{n,m}}}}. \label{eq:W11D2}
 \end{gather}
 As stated in the introduction, from the relation $W_{1}=\xi_{n}$ this f\/irst integral def\/ines a three-point, second-order ordinary dif\/ference equation in the $n$ direction which depends parametrically on~$m$. From the parametric dependence we f\/ind two dif\/ferent three-point non-autonomous ordinary dif\/ference equations corresponding to~$m$ even and~$m$ odd. We treat them separately.

{\bf Case $\boldsymbol{m=2l}$.} If $m=2l$ we have the following non-autonomous ordinary dif\/ference equation
 \begin{gather*}
{ \Fp{n}} {\frac { [ ( 1+{ \delta_{2}}) u_{{n,2l}}+u_{{n+1,2l}} ] {
 \delta_{1}}-u_{{n,2l}}}{[ ( 1+{ \delta_{2}}) u_{{n,2l}}+u_{{n-1,2l}} ] { \delta_{1}}-u_{{n,2l}} }}
+{ \Fm{n}} ( u_{{n+1,2l}} -u_{{n-1,2l}}) = \xi_{n}, %\label{eq:1D2eQm}
 \end{gather*}
where without loss of generality we have chosen $\alpha=1$ and $\beta=1$. We can easily see, that once solved for $u_{n+1,2l}$ the equation is \emph{linear}
 \begin{gather}
 u_{{n+1,2l}} -{\frac {{ \Fp{n}} ( 1-\delta_{1}\delta_{2}-\delta_{1} ) ( 1 -\xi_{n}) u_{{n,2l}}}{{ \delta_1}}}
- \big(\Fp{n}\xi_{n}+{\Fm{n}}\big) u_{{n-1,2l}} -\Fm{n}\xi_{n}=0. \label{eq:1D2eQm2}
 \end{gather}
 Tackling this equation directly is very dif\/f\/icult, but we can separate again the cases when $n$ is even and odd and convert \eqref{eq:1D2eQm2} into a system using the standard transformation \eqref{eq:genautsubl}
 \begin{subequations} \label{eq:1D2eQmsys}
 \begin{gather}
 w_{{k,l}}- \xi_{2k} w_{k-1,l}= \frac{\delta}{\delta_{1}} \left(1-\xi_{2k}\right)v_{k,l}, \label{eq:1D2eQmk} \\
 v_{{k+1,l}}-v_{{k,l}}={ {\xi_{{2k+1}}}}, \label{eq:1D2eQmkp}
 \end{gather}
 \end{subequations}
where $\delta$ is given by \eqref{eq:deltaconstdef}. Now we have two \emph{first-order} ordinary dif\/ference equations. Equation~\eqref{eq:1D2eQmkp} is uncoupled from equation \eqref{eq:1D2eQmk}. Furthermore, since $\xi_{2k}$ and $\xi_{2k+1}$ are independent functions we can write $\xi_{2k+1} = a_{k+1} - a_{k}$. So the second equation possesses the trivial solution\footnote{From now on we use the convention of naming the arbitrary functions depending on $k$ with Latin letters and the functions depending on $l$ by Greek ones.}
 \begin{gather}
 v_{k,l} = \alpha_{l} + a_{k}. \label{eq:1D2vsol}
 \end{gather}
 Now introduce \eqref{eq:1D2vsol} into \eqref{eq:1D2eQmk} and solve the equation for $w_{k,l}$
 \begin{gather*}
 w_{{k,l}}- \xi_{2k} w_{k-1,l} = \frac{\delta}{\delta_{1}} (1-\xi_{2k})( \alpha_{l} + a_{k}). %\label{eq:1D2eQmkp2}
 \end{gather*}
 We def\/ine $\xi_{2k}=b_{k}/b_{k-1}$ and perform the change of dependent variable: $w_{k,l}=b_{k}W_{k,l}$. Then~$W_{k,l}$ solves the equation
 \begin{gather*}
 W_{k,l}-W_{k-1,l} = \frac{\delta}{\delta_{1}} \left( \frac{1}{b_{k}}-\frac{1}{b_{k-1}} \right)( \alpha_{l}+a_{k}). %\label{eq:1D2eQmkp3}
 \end{gather*}
 The solution of this dif\/ference equation is given by
 \begin{gather*}
 W_{k,l} = \beta_{l} + \frac{\delta}{\delta_{1}}\frac{\alpha_{l}}{b_{k}} + c_{k}, %\label{eq:1D2Wklsol}
 \end{gather*}
 where $c_{k}$ is such that
 \begin{gather}
 c_{k}-c_{k-1} = \frac{\delta}{\delta_{1}} \left( \frac{1}{b_{k}}-\frac{1}{b_{k-1}} \right) a_{k}. \label{eq:1D2ckdef}
 \end{gather}
 Equation \eqref{eq:1D2ckdef} is not a total dif\/ference, but it can be used to def\/ine $a_{k}$ in terms of the arbitrary functions~$b_{k}$ and~$c_{k}$
 \begin{gather}
 a_{k} = -\frac{\delta_1}{\delta} \frac{b_{k} b_{k-1} }{b_{k}-b_{k-1}} (c_{k}-c_{k-1}). \label{eq:1D2akdef}
 \end{gather}
 This means that we have the following solution for the system \eqref{eq:1D2eQmsys}
 \begin{subequations} \label{eq:1D2eQmsol}
 \begin{gather}
 v_{k,l} = \alpha_{l} - \frac{\delta_1}{\delta} \frac{b_{k} b_{k-1} }{b_{k}-b_{k-1}} (c_{k}-c_{k-1}), \label{eq:1D2vsol2} \\
 w_{k,l} = b_{k}(\beta_{l} +c_{k}) + \frac{\delta}{\delta_{1}}\alpha_{l}, \label{eq:1D2wsol}
 \end{gather}
 \end{subequations}

{\bf Case $\boldsymbol{m=2l+1}$.} If $m=2l+1$ we have the following non-autonomous ordinary dif\/ference equation
 \begin{gather}
 { \Fp{n}}{\frac {1+ ( u_{{n+1,2l+1}}-1 ) { \delta_{1}}}{ 1+ ( u_{{n-1,2l+1}}-1 ) { \delta_{1}} }}
+{ \Fm{n}}{\frac { ( u_{{n-1,2l+1}}-u_{{n+1,2l+1}} ) [ 1 -{ \delta_{1}} ( 1-u_{{n,2l+1}} ) ] }{{ \delta_{2}}+u_{{n,2l+1}}}}
 =\xi_{n}, \label{eq:1D2ePm}
 \end{gather}
 We can easily see that the equation is genuinely nonlinear. However we can separate the cases when $n$ is even and odd and convert \eqref{eq:1D2ePm} into a system using the standard transforma\-tion~\eqref{eq:genautsublp}
 \begin{subequations} \label{eq:1D2ePmsys}
 \begin{gather}
 z_{{k,l}}-\xi_{2k} z_{{k-1,l}}= \left(1-\frac{1}{\delta_1}\right)(1-\xi_{2k}), \label{eq:1D2ePmk} \\
 y_{k+1,l}-y_{k,l}=\frac{\xi_{{2k+1}}}{\delta_{1}}{\frac { \delta-1+\delta_{1}(1 - z_{{k,l}}) }{ 1-\delta_{1}(1-z_{{k,l}})}}, \label{eq:1D2ePmkp}
 \end{gather}
 \end{subequations}
 where we used the def\/inition \eqref{eq:deltaconstdef}. This is a system of two \emph{first-order} dif\/ference equation, and equation \eqref{eq:1D2ePmk} is linear and uncoupled from~\eqref{eq:1D2ePmkp}. As $\xi_{2k}=b_{k}/b_{k-1}$ we have that~\eqref{eq:1D2ePmk} is a total dif\/ference
 \begin{gather}
 \frac{z_{k,l}}{b_{k}}-\frac{z_{k-1,l}}{b_{k-1}}= \left( 1-\frac{1}{\delta_{1}} \right)\left(\frac{1}{b_{k}}-\frac{1}{b_{k-1}}\right).
 \label{eq:1D2ePmk2}
 \end{gather}
 Hence the solution of \eqref{eq:1D2ePmk2} is given by
 \begin{gather}
 z_{k,l} = 1-\frac{1}{\delta_{1}}+b_{k}\gamma_{l}. \label{eq:zsol}
 \end{gather}

 Inserting \eqref{eq:zsol} into \eqref{eq:1D2ePmkp} and using the def\/inition of $\xi_{2k+1}$ in terms of $a_{k}$, i.e., $\xi_{2k+1}=a_{k+1}-a_{k}$ we obtain
 \begin{gather}
 y_{k+1,l}-y_{k,l}=-\left(\frac{1}{\delta_{1}} +\frac{\delta}{\delta_{1}^{2}b_{k}\gamma_{l}}\right) ( a_{k+1}-a_{k} ). \label{eq:1D2ePmkp2}
 \end{gather}
 We can then represent the solution of \eqref{eq:1D2ePmkp2} as
 \begin{gather*}
 y_{k,l} = \zeta_{l} + \frac{\delta d_{k}}{\delta_{1}^{2}\gamma_{l}} -\frac{a_{k}}{\delta_{1}}, %\label{eq:1D2ykl}
 \end{gather*}
 where $d_{k}$ satisf\/ies the f\/irst-order linear dif\/ference equation
 \begin{gather}
 d_{k+1}-d_{k}=\frac{a_{k+1}-a_{k}}{b_{k}}. \label{eq:1D2dkeq}
 \end{gather}
 Inserting the value of $a_{k}$ given by \eqref{eq:1D2akdef} inside \eqref{eq:1D2dkeq} we obtain that this equation is a total dif\/ference. Then $d_{k}$ is given by
 \begin{gather*}
 d_{k}= -\frac{(c_{k}-c_{k-1})b_{k-1}\delta_1}{(b_{k}-b_{k-1})\delta} -\frac{\delta_1 c_{k}}{\delta}. %\label{eq:1D2dkdef}
 \end{gather*}
 This means that f\/inally we have the following solutions for the f\/ields $z_{k,l}$ and $y_{k,l}$
 \begin{subequations} \label{eq:1D2zyklsols}
 \begin{gather}
 z_{k,l} = 1-\frac{1}{\delta_{1}}+b_{k}\gamma_{l}, \label{eq:1D2zklsol2} \\
 y_{k,l} = \zeta_{l} -\frac{(c_{k}-c_{k-1})b_{k-1}}{(b_{k}-b_{k-1})\delta_{1}\gamma_{l}} -\frac{c_{k}}{\delta_{1}\gamma_{l}} +\frac{1}{\delta}
 \frac{b_{k} b_{k-1} \left(c_{k}-c_{k-1}\right) }{b_{k}-b_{k-1}}. \label{eq:1D2yklsol2}
 \end{gather}
 \end{subequations}

 Equations (\ref{eq:1D2eQmsol}), (\ref{eq:1D2zyklsols}) provide the value of the four f\/ields, but we have too many arbitrary functions in the $l$ direction, namely $\alpha_{l}$, $\beta_{l}$, $\gamma_{l}$ and $\zeta_{l}$. Introducing (\ref{eq:1D2eQmsol}), (\ref{eq:1D2zyklsols}) into~\eqref{eq:1D2} and separating the terms even and odd in $n$ and $m$ we obtain two independent equations
 \begin{gather}
 (\alpha_{l}+\delta_1\zeta_{l})\gamma_{l}+\beta_{l}=0, \qquad
 (\alpha_{l+1}+\delta_1\zeta_{l})\gamma_{l}+\beta_{l+1}=0, \label{eq:1D2constr}
 \end{gather}
 which allow us to reduce by two the number of independent functions in the $l$ direction. Solving equations \eqref{eq:1D2constr} with respect to $\gamma_{l}$ and $\zeta_{l}$ we obtain
 \begin{gather}
 \gamma_{l} = -\frac{\beta_{l+1}-\beta_{l}}{\alpha_{l+1}-\alpha_{l}},
 \qquad
 \zeta_{l} =\frac{1}{\delta_{1}} \frac{\beta_{l}\alpha_{l+1}-\beta_{l+1}\alpha_{l}}{\beta_{l+1}-\beta_{l}}. \label{eq:1D2gammadeltadef}
 \end{gather}

Inserting \eqref{eq:1D2gammadeltadef} into (\ref{eq:1D2eQmsol}), (\ref{eq:1D2zyklsols}) we obtain that the general solution of $_{1}D_{2}$ equation~\eqref{eq:1D2} is given by \eqref{eq:1D2solfin}, provided that $\delta_{1}\neq0$ and $\delta\neq0$. Indeed the solution~\eqref{eq:1D2solfin} is ill-def\/ined if $\delta_{1}=0$ or $\delta=0$ and we proceed to treat the relevant cases separately.

{\bf Case $\boldsymbol{\delta=0}$.} If $\delta=0$ we can solve \eqref{eq:deltaconstdef} with respect to $\delta_{1}$
 \begin{gather}
 \delta_{1} =\frac{1}{1+\delta_{2}}. \label{eq:deltasolved}
 \end{gather}
 The f\/irst integral \eqref{eq:W11D2} is not singular for $\delta_{1}$ given by \eqref{eq:deltasolved}. The procedure of solution becomes dif\/ferent only when we arrive to the systems of ordinary dif\/ference equations \eqref{eq:1D2eQmsys} and \eqref{eq:1D2ePmsys}. So we will present the solution of the systems in this case.

{\bf Case $\boldsymbol{m=2l}$.} If $\delta_{1}$ is given by equation~\eqref{eq:deltasolved} the system \eqref{eq:1D2eQmsys} becomes
 \begin{subequations} \label{eq:1D2eQmd1ssys}
 \begin{gather}
 w_{k,l}-\xi_{2k}w_{k-1,l}=0, \label{eq:1D2eQmd1sk} \\
 v_{k+1,l}-v_{k,l}=\xi_{2k+1}. \label{eq:1D2eQmd1skp}
 \end{gather}
\end{subequations}
 The system \eqref{eq:1D2eQmd1ssys} is uncoupled and imposing $\xi_{2k}=b_{k}/b_{k-1}$ and $\xi_{2k+1} = a_{k+1}-a_{k}$ it is readily solved to give{\samepage
 \begin{subequations} \label{eq:1D2wvd1sol}
 \begin{gather}
 v_{k, l} = a_{k}+\beta_{l}. \label{eq:1D2vd1sol}\\
 w_{k, l} = b_{k}\alpha_{l}, \label{eq:1D2wd1sol}
 \end{gather}
 \end{subequations}}

{\bf Case $\boldsymbol{m=2l+1}$.} If $\delta_{1}$ is given by equation~\eqref{eq:deltasolved} the system \eqref{eq:1D2ePmsys} becomes
 %\begin{subequations}\label{eq:1D2ePmd1ssys}
 \begin{gather*}
 \frac{\delta_2+z_{k, l}}{b_{k}}= \frac{\delta_2+z_{k-1, m}}{b_{k-1}}, \qquad %\label{eq:1D2ePmd1sk} \\
 -\frac{y_{k+1, l}-y_{k, l}}{1+\delta_2} = a_{k+1}-a_{k}, %\label{eq:1D2ePmd1skp}
 \end{gather*}
 where we used the fact that $\xi_{2k}=b_{k}/b_{k-1}$ and $\xi_{2k+1} = a_{k+1}-a_{k}$. The solution to this system is immediate and it is given by
 \begin{subequations} \label{eq:1D2zyd1sol}
 \begin{gather}
 z_{k, l} = b_{k}\gamma_{l}-\delta_{2}, \label{eq:1D2zd1sol} \\
 y_{k, l}= ( 1+\delta_{2}) (\zeta_{l}-a_{k}). \label{eq:1D2yd1sol}
 \end{gather}
 \end{subequations}

 As in the general case we obtained the expressions of the four f\/ields, but we have too many arbitrary functions in the $l$ direction, namely $\alpha_{l}$,
 $\beta_{l}$, $\gamma_{l}$ and $\zeta_{l}$. Substituting the obtained expressions (\ref{eq:1D2wvd1sol}), (\ref{eq:1D2zyd1sol}) in the equation $_{1}D_{2}$ \eqref{eq:1D2} with $\delta_{1}$ given by equation \eqref{eq:deltasolved} separating the even and odd terms we obtain two compatibility conditions
 \begin{gather*}
 \alpha_{l}+\gamma_{l} \beta_{l}+\gamma_{l}\zeta_{l}=0, \qquad
 \alpha_{l+1}+\gamma_{l}\beta_{l+1}+\gamma_{l}\zeta_{l}=0. %\label{eq:1D2ccd1s}
 \end{gather*}
 We can solve this equation with respect to $\gamma_{l}$ and $\zeta_{l}$ and we obtain
 \begin{gather}
 \gamma_{l} = -\frac{\alpha_{l+1}-\alpha_{l}}{\beta_{l+1}-\beta_{l}}, \qquad
 \zeta_{l} = -\frac{\beta_{l}\alpha_{l+1}-\beta_{l+1}\alpha_{l}}{\alpha_{l+1}-\alpha_{l}}. \label{eq:1D2gamdeld1s}
 \end{gather}

 Inserting \eqref{eq:1D2gamdeld1s} into (\ref{eq:1D2wvd1sol}), (\ref{eq:1D2zyd1sol}) we obtain that the general solution of $_{1}D_{2}$ equation~\eqref{eq:1D2} when $\delta=0$ is given by~\eqref{eq:1D2d1sol}.

{\bf Case $\boldsymbol{\delta_{1}=0}$.} If $\delta_{1}=0$ the f\/irst integral \eqref{eq:W11D2} is singular. Then following \cite{GSY_DarbouxI} the $_{1}D_{2}$ equation~\eqref{eq:1D2} with $\delta_{1}=0$ possesses in the direction~$n$ the following three-point, second-order integral
 \begin{gather}
 W_{1}^{(0,\delta_{2})} =\Fppp \alpha \frac {u_{n+1,m}-u_{n-1,m}}{u_{n,m}}
 -\Fpmm \alpha( u_{{n+1,m}}-u_{{n-1,m}})\nonumber \\
 \hphantom{W_{1}^{(0,\delta_{2})} =}{} +\Fmpm \beta ( u_{{n+1,m}}-u_{{n-1,m}} ) + \Fmmp\beta{\frac {u_{{n-1,m}}-u_{{n+1,m}}}{{ \delta_2}+u_{{n,m}}}}.
 \label{eq:W11D20d2}
 \end{gather}
 In order to solve the $_{1}D_{2}$ equation \eqref{eq:1D2} in this case we use the f\/irst integral~\eqref{eq:W11D20d2}. We start separating the cases even and odd in~$m$.

{\bf Case $\boldsymbol{m=2l}$.} If $m=2l$ we obtain from the f\/irst integral \eqref{eq:W11D20d2}
 \begin{gather}
 \Fp{n}\frac {u_{{n+1,2l}}-u_{{n-1,2l}} }{u_{{n,2l}}} +\Fm{n} ( u_{{n+1,2l}}-u_{{n-1,2l}} ) =\xi_{{n}}, \label{eq:eQm0d2}
 \end{gather}
 where we have chosen without loss of generality $\alpha=\beta=1$. This equation is nonlinear. Applying the transformation \eqref{eq:genautsubl} we transform equation \eqref{eq:eQm0d2} into the system
 \begin{subequations}\label{eq:eQm0d2sys}
 \begin{gather}
 w_{k,l} - w_{k-1,l} = \xi_{2k}v_{k,l}, \label{eq:eQm0d2k} \\
 v_{k+1,l}-v_{k,l} = \xi_{2k+1}. \label{eq:eQm0d2kp}
 \end{gather}
 \end{subequations}
 The system \eqref{eq:eQm0d2sys} is linear and equation \eqref{eq:eQm0d2kp} is uncoupled from equation \eqref{eq:eQm0d2k}. If we put $\xi_{2k+1} = a_{k+1}-a_{k}$ then equation \eqref{eq:eQm0d2kp} has the solution
 \begin{gather*}
 v_{k,l} = a_{k} + \alpha_{l}. %\label{eq:1D20d2vsol}
 \end{gather*}
 Substituting into \eqref{eq:eQm0d2k} we obtain
 \begin{gather}
 w_{k, l}-w_{k-1, l}=\xi_{2k} (a_{k}+\alpha_{l}). \label{eq:eQm0d2k2}
 \end{gather}
 Equation \eqref{eq:eQm0d2k2} becomes a total dif\/ference if we set
 \begin{gather}
 \xi_{2k} = b_{k}-b_{k-1}, \qquad a_{k} = \frac{c_{k}-c_{k-1}}{b_{k}-b_{k-1}} \label{eq:1D20d2lkadef}
 \end{gather}
 and then the solution of the system \eqref{eq:1D2eQmd1ssys} is given by
 \begin{subequations}\label{eq:1D20d2wvsol}
 \begin{gather}
 v_{k,l} = \frac{c_{k}-c_{k-1}}{b_{k}-b_{k-1}} +\alpha_{l}, \label{eq:1D20d2wsol2} \\
 w_{k,l} = c_{k} + \alpha_{l}b_{k} + \beta_{l}. \label{eq:1D20d2wsol}
 \end{gather}
 \end{subequations}

{\bf Case $\boldsymbol{m=2l+1}$.} If $m=2l+1$ we obtain from the f\/irst integral \eqref{eq:W11D20d2}
 \begin{gather}
 {\Fp{n}} \left( u_{{n-1,2l+1}}-u_{{n+1,2l+1}} \right) -{ \Fm{n}}{\frac { \left( u_{{n+1,2l+1}}-u_{{n-1,2l+1}} \right) }{{ \delta_2}+u_{{n,2l+1}}}}= \xi_{{n}}, \label{eq:ePm0d2}
 \end{gather}
where we have chosen without loss of generality $\alpha=\beta=1$. This equation is nonlinear. Applying the transformation \eqref{eq:genautsublp} we transform equation \eqref{eq:ePm0d2} into the system
 \begin{subequations} \label{eq:ePm0d2sys}
 \begin{gather}
 z_{k-1,l}-z_{k,l}=b_{k}-b_{k-1}, \label{eq:ePm0d2k} \\
 y_{k,l}-y_{k+1,l}=(a_{k+1}-a_{k})(\delta_2+z_{k,l}), \label{eq:ePm0d2kp}
 \end{gather}
 \end{subequations}
where we used the values of $\xi_{2k}$ and $\xi_{2k+1}$. The system is now linear and equation \eqref{eq:ePm0d2k} is solved by
 \begin{gather*}
 z_{k,l} = \gamma_{l} - b_{k}. %\label{eq:1D20d2zsol}
 \end{gather*}
 Substituting into \eqref{eq:ePm0d2kp} we obtain
 \begin{gather*}
 y_{k,l}-y_{k+1,l}=(a_{k+1}-a_{k})(\delta_2+\gamma_{l}-b_{k}).% \label{eq:ePm0d2kp2}
 \end{gather*}
 Then we have that $y_{k,l}$ is given by
 \begin{gather*}
 y_{k,l} = \zeta_{l} - ( \gamma_{l}+\delta_{2} )a_{k} +d_{k}, %\label{eq:1D20d2ysol}
 \end{gather*}
 where $d_{k}$ solves the ordinary dif\/ference equation
 \begin{gather}
 d_{k+1}-d_{k} = b_{k+1}\frac{c_{k+1}-c_{k}}{b_{k+1}-b_{k}}-c_{k+1} -b_{k}\frac{c_{k}-c_{k-1}}{b_{k}-b_{k-1}}+c_{k}. \label{eq:1D20d2dkeq}
 \end{gather}
In \eqref{eq:1D20d2dkeq} we inserted the value of $a_{k}$ according to \eqref{eq:1D20d2lkadef}. Equation \eqref{eq:1D20d2dkeq} is a total dif\/ference and then $d_{k}$ is given by
 \begin{gather*}
 d_{k} = b_{k}\frac{c_{k}-c_{k-1}}{b_{k}-b_{k-1}}-c_{k}. %\label{eq:1D20d2dsol}
 \end{gather*}
 Then the solution of the system \eqref{eq:ePm0d2sys} is
 \begin{subequations} \label{eq:1D20d2zysol}
 \begin{gather}
 z_{k,l} = \gamma_{l} - b_{k}, \label{eq:1D20d2zsol2} \\
 y_{k,l} = \zeta_{l} - ( \gamma_{l}+\delta_{2} ) \frac{c_{k}-c_{k-1}}{b_{k}-b_{k-1}} +b_{k}\frac{c_{k}-c_{k-1}}{b_{k}-b_{k-1}}-c_{k}.
 \end{gather}
 \end{subequations}

As in the general case we obtained the expressions of the four f\/ields, but we have too many arbitrary functions in the~$l$ direction, namely~$\alpha_{l}$,~$\beta_{l}$,~$\gamma_{l}$ and~$\zeta_{l}$. Substituting the obtained expressions~(\ref{eq:1D20d2wvsol}), (\ref{eq:1D20d2zysol}) in the equation $_{1}D_{2}$~\eqref{eq:1D2} with $\delta_{1}=0$ separating the even and odd terms we obtain two compatibility conditions
 \begin{gather}
 (\gamma_{l}+\delta_2) \alpha_{l}+\zeta_{l}+\beta_{l}=0, \qquad (\gamma_{l}+\delta_2) \alpha_{l+1}+\beta_{l+1}+\zeta_{l}=0. \label{eq:1D20d2cc}
 \end{gather}
 We can solve equation \eqref{eq:1D20d2cc} with respect to $\gamma_{l}$ and $\zeta_{l}$ and to obtain
 \begin{gather}
 \gamma_{l} = -\delta_2-\frac{\beta_{l+1}-\beta_{l}}{\alpha_{l+1}-\alpha_{l}}, \qquad
 \zeta_{l} = \frac{\beta_{l+1} \alpha_{l} -\alpha_{l+1} \beta_{l}}{\alpha_{l+1}-\alpha_{l}}. \label{eq:1D20d2gamdel}
 \end{gather}

Inserting \eqref{eq:1D20d2gamdel} into (\ref{eq:1D20d2wvsol}), (\ref{eq:1D20d2zysol}) we obtain that the general solution of $_{1}D_{2}$ equation~\eqref{eq:1D2} when $\delta=0$ is given by \eqref{eq:1D20d2sol}.

This discussion exhausts the possible cases. For any value of the parameters we have the general solution of the $_{1}D_{2}$ equation \eqref{eq:1D2} and this concludes the proof.
\end{proof}

\begin{Proposition} The $_{2}D_{2}$ equation \eqref{eq:2D2} is exactly solvable. If $\delta_{1}\neq0$ its general solution is given by
 \begin{subequations} \label{eq:2D2solfin}
 \begin{gather}
 v_{k,l} = b_{k}+\beta_{l} + \delta\frac{c_{k}}{\alpha_{l}} +\frac{1}{\delta_{1}}-1-\delta, \label{eq:2D2vsolfin} \\
 w_{k,l} = \alpha_{l}\frac{b_{k+1}-b_{k}}{c_{k+1}-c_{k}} + \frac{1}{\delta_{1}}-1-\delta, \label{eq:2D2wsolfin} \\
 z_{k,l} = \delta b_{k} +\frac{b_{k+1}-b_{k}}{c_{k+1}-c_{k}} \left[\alpha_{{l}} ( \beta_{{l+1}}-\beta_{{l}} ) +
 {\frac {\alpha_{l}^{2} ( \beta_{{l+1}}-\beta_{{l}} )}{ \alpha_{{l+1}}-\alpha_{{l}} }}-\delta c_{k}\right]\nonumber \\
\hphantom{z_{k,l} =}{} +\delta\frac {\beta_{{l+1}}{ \delta_{1}}-1-\delta_{1}^{2}
 \lambda+\delta { \delta_{1}}+{ \delta_{1}}}{\delta_{1}} +\frac {\delta \alpha_{{l}} ( \beta_{{l+1}}-\beta_{{l}}) }{ \alpha_{{l+1}} -\alpha_{{l}}},
 \label{eq:2D2zklsolfin} \\
 y_{k,l} =-{\frac {\beta_{{l+1}}{ \delta_{1}}-1-\delta_{1}^{2} \lambda+\delta { \delta_{1}}+{ \delta_{1}}}{\delta_{1}^{2}}} -{\frac {\alpha_{{l}} ( \beta_{{l+1}}-\beta_{{l}} ) }{ ( \alpha_{{l+1}} -\alpha_{{l}}) { \delta_{1}}}} - \frac{b_{k}}{\delta_{1}},
 \label{eq:2D2yklsolfin}
 \end{gather}
 \end{subequations}
 where $b_{k}$, $c_{k}$, $\alpha_{l}$ and $\beta_{l}$ are arbitrary functions of their arguments. If $\delta_{1}=0$ then its general solution is given by
 \begin{subequations}\label{eq:2D20d2sol}
 \begin{gather}
 v_{k,l} = a_{k}\alpha_{l}, \label{eq:2D20d2vsolfin} \\
 w_{k,l} = -\delta_{2} -\frac{1}{\alpha_{l}}\frac{b_{k+1}-b_{k}}{a_{k+1}-a_{k}}, \label{eq:2D20d2wsolfin} \\
 y_{k,l} = b_{k} + \beta_{l}, \label{eq:2D20d2ysolfin} \\
 z_{k,l} = a_{k}\frac{b_{k+1}-b_{k}}{a_{k+1} -a_{k}}-b_{k} - \beta_{l}, \label{eq:2D20d2zsolfin}
 \end{gather}
 \end{subequations}
 where $a_{k}$, $b_{k}$, $\alpha_{l}$ and $\beta_{l}$ are arbitrary functions of their arguments.
\end{Proposition}

\begin{proof}
 The proof of the two solution \eqref{eq:2D2solfin} and \eqref{eq:2D20d2sol} proceeds as the one outlined in Proposition~\ref{prop:1D2}. The interested reader can f\/ind it in Appendix~\ref{app:remsols}.
\end{proof}

\begin{Proposition}
 The $_{3}D_{2}$ equation \eqref{eq:3D2} is exactly solvable. If $\delta_{1}\neq0$ and $\delta\neq0$, where $\delta$ is give by~\eqref{eq:deltaconstdef}, its general solution is given by
 \begin{subequations}\label{eq:3D2solfin}
 \begin{gather}
 v_{k,l} = -\frac{\delta_{1}}{\delta}b_{k} + c_{k}\alpha_{l}, \label{eq:3D2vsolfin} \\
 w_{k,l}= \frac{\delta_{1}-1+\delta}{\delta_{1}}+ \frac{\delta(b_{k+1}-b_{k})}{ \delta\alpha_{l}(c_{k+1}-c_{k})-(b_{k+1}-b_{k})\delta_1},\label{eq:3D2wsolfin}
 \\
 z_{k,l} =\frac{b_{k}}{\delta} -\frac{1}{\delta_{1}}\left( \beta_{l}-\lambda\delta_{1} - \alpha_{l}\frac{ \beta_{l+1}-\beta_{l} }{\alpha_{l+1}-\alpha_{l} }\right) +\frac{1}{\delta}\frac{b_{k+1}-b_{k}}{c_{k+1}-c_{k}}
 \left(\frac{ \beta_{l+1}-\beta_{l} }{\alpha_{l+1}-\alpha_{l} }+c_{k} \right), \label{eq:3D2zklsolfin} \\
 y_{k,l} =-b_{k}+\frac{\delta}{\delta_{1}}\left( \beta_{l}-\lambda\delta_{1} -
 \alpha_{l}\frac{ \beta_{l+1}-\beta_{l} }{\alpha_{l+1}-\alpha_{l} }\right), \label{eq:3D2yklsolfin}
 \end{gather}
 \end{subequations}
 where $b_{k}$, $c_{k}$, $\alpha_{l}$ and $\beta_{l}$ are arbitrary functions of their arguments. If $\delta=0$ then its general solution is given by
 \begin{subequations}\label{eq:3D2d1d1sol}
 \begin{gather}
 v_{k,l} = c_{k}+ b_{k}\alpha_{l} + \beta_{l}, \label{eq:3D2d1d2vsolfin} \\
 w_{k,l}= -\delta_{2} + \frac{b_{k+1}-b_{k}}{ c_{k+1}-c_{k} + \alpha_{l}( b_{k+1} -b_{k} )}, \label{eq:3D2d1d2wsolfin} \\
 y_{k,l} = -\frac{\beta_{l+1}-\beta_{l}}{\alpha_{l+1}-\alpha_{l}}-b_{k}, \label{eq:3D2d1d2ysolfin} \\
 z_{k,l} = -(\delta_2+1)\left[ \left( b_{k+1} -\frac{\beta_{l+1}-\beta_{l}}{\alpha_{l+1}-\alpha_{l}}\right) \frac{c_{k+1}-c_{k}}{b_{k+1}-b_{k}}
 -c_{k+1}\right]\nonumber \\
 \hphantom{z_{k,l} =}{} + ( \delta_{2}+1 ) \frac{\beta_{l+1}\alpha_{l}-\alpha_{l+1}\beta_{l}}{\alpha_{l+1}-\alpha_{l}} +\lambda,
 \label{eq:3D2d1d2zsolfin}
 \end{gather}
 \end{subequations}
 where $b_{k}$, $c_{k}$, $\alpha_{l}$ and $\beta_{l}$ are arbitrary functions of their arguments. If $\delta_{1}=0$ then its general solution is given by
 \begin{subequations}\label{eq:3D20d2sol}
 \begin{gather}
 v_{k,l} = \frac{b_{k}+\zeta_{0}}{\alpha_{l}}, \label{eq:3D20d2vsolfin} \\
 w_{k,l} = -\delta_{2} - \alpha_{l}\frac{c_{k+1}-c_{k}}{b_{k+1}-b_{k}}, \label{eq:3D20d2wsolfin} \\
 y_{k,l} = c_{k}+\gamma_{l}, \label{eq:3D20d2ysolfin} \\
 z_{k,l} = \frac{c_{k+1}-c_{k}}{b_{k+1}-b_{k}}\zeta_{0}-\gamma_{l} +\frac{b_{k}c_{k+1}-c_{k}b_{k+1}}{b_{k+1}-b_{k}}, \label{eq:3D20d2zsolfin}
 \end{gather}
\end{subequations}
 where $b_{k}$, $c_{k}$, $\alpha_{l}$ and $\gamma_{l}$ are arbitrary functions of their arguments and $\zeta_{0}$ is a~constant.
\end{Proposition}

\begin{proof} The proof of the three solution \eqref{eq:3D2solfin}, \eqref{eq:3D2d1d1sol} and \eqref{eq:3D20d2sol} proceeds as the one outlined in Proposition~\ref{prop:1D2}. The interested reader can f\/ind it in Appendix~\ref{app:remsols}.
\end{proof}

\subsection[The $D_{3}$ and the $_{i}D_{4}$ equations, $i=1,2$]{The $\boldsymbol{D_{3}}$ and the $\boldsymbol{{}_{i}D_{4}}$ equations, $\boldsymbol{i=1,2}$}\label{sec:otherH6}

We have the following propositions:

\begin{Proposition} \label{prop:D3}
 The $D_{3}$ equation \eqref{eq:D3} is exactly solvable. We have that the expression of the fields $y_{k,l}$ and $v_{k,l}$ is given by
 \begin{subequations} \label{eq:D3yvsolfin}
 \begin{align}
 y_{k,l} &= \alpha_{l}c_{k}+d_{k}+\beta_{l}
 \label{eq:D3ysolfin}
 \\
 v_{k,l} &
 \begin{aligned}[t]
 &=\left( \alpha_{l}c_{k}+d_{k}+\beta_{l} \right)^{2}
 \\
 &+\left[
 \left(\alpha_{l}-\alpha_{l-1}\right)c_{k}+\beta_{l}-\beta_{l-1}\right]
 \left[
 \gamma_{l}+
 \frac{e_{k}-\alpha_{l-1}\left(\alpha_{l-1} c_{k}+2 d_{k}\right)}{ \alpha_{l}-\alpha_{l-1}}\right],
 \end{aligned}
 \label{eq:D3vsolfin}
 \end{align}
\end{subequations}
 where the function functions $c_{k}$, $d_{k}$, $\alpha_{l}$
 and $\beta_{l}$ are arbitrary functions of their arguments,
 whereas the function $e_{k}$ and $\gamma_{l}$ are given through
 the discrete integrations
 \begin{subequations} \label{eq:D3discrint}
 \begin{gather}
 e_{k+1}=e_{k}
 -\frac{\left( d_{k+1}-d_{k} \right)^{2}}{c_{k+1}-c_{k}},
 \label{eq:D3ekeqdef0}
 \\
 \left(\alpha_{{l+1}} -\alpha_{{l}}\right) \gamma_{{l+1}}
 - \left(\alpha_{{l}}-\alpha_{{l-1}} \right) \gamma_{{l}}
 -\alpha_{{l-1}}\beta_{{l-1}}+\alpha_{{l}}\beta_{{l}}
 +\alpha_{{l}}\beta_{{l-1}}-\alpha_{{l-1}}\beta_{{l}}=0.
 \label{eq:D3gammleq0}
 \end{gather}
 \end{subequations}
 The fields $z_{k,l}$ and $w_{k,l}$ are then given in terms of $y_{k,l}$ and $v_{k,l}$ as
 \begin{subequations} \label{eq:D3zwsolfin}
 \begin{align}
 z_{k,l} &= \frac{\left( y_{k+1,l} y_{k,l-1}-y_{k+1,l-1} y_{k,l} \right)y_{k,l}
 -\left( y_{k+1,l}+y_{k,l-1}-y_{k+1,l-1}-y_{k,l} \right)v_{k,l}}{%
 \left( y_{k+1,l}+y_{k,l-1}-y_{k+1,l-1}-y_{k,l} \right) y_{k,l}
 -y_{k+1,l} y_{k,l-1}+y_{k+1,l-1} y_{k,l}}
 \label{eq:D3zsolfin}
 \\
 w_{k,l} &=
 -\frac{y_{k+1,l} y_{k,l-1}-y_{k+1,l-1} y_{k,l}}{y_{k+1,l}+y_{k,l-1}-y_{k+1,l-1}-y_{k,l}}.
 \label{eq:D3wsolfin}
 \end{align}
 \end{subequations}
\end{Proposition}

\begin{Remark} We remark that we can say that equation \eqref{eq:D3gammleq0} def\/ines a discrete integration for the function $\gamma_{l}$ since it can be expressed in the form \eqref{eq:discrintdef}. Def\/ining a the new function $\zeta_{l}$ by
 \begin{gather*}
 \gamma_{l} =-\frac{\alpha_{l-1}\beta_{l-1}+\zeta_{l}}{ \alpha_{l}-\alpha_{l-1}}
 %\label{eq:D3gammaldef}
 \end{gather*}
 we have that $\zeta_{l}$ satisf\/ies the following dif\/ference equation
 \begin{gather}
 \zeta_{l+1} - \zeta_{l}= \alpha_{l} \beta_{l-1}-\alpha_{l-1} \beta_{l}, \label{eq:D3deltaeqdef}
 \end{gather}
 i.e., it is given by a discrete integration.
\end{Remark}

\begin{proof} To f\/ind general of the $D_{3}$ equation \eqref{eq:D3} we start from the equation itself. Applying the general transformation \eqref{eq:genautsub} to the $D_{3}$ equation \eqref{eq:D3} we obtain the following system of four equations:
 \begin{subequations}\label{eq:D3sys}
 \begin{gather}
 v_{k,l}+w_{k,l} y_{k,l}+w_{k,l} z_{k,l}+y_{k,l} z_{k,l}=0, \label{eq:D3a}
 \\
 v_{k,l+1}+y_{k,l} w_{k,l+1}+z_{k,l} w_{k,l+1}+y_{k,l} z_{k,l}=0, \label{eq:D3b}
 \\
 v_{k+1,l}+w_{k,l} y_{k+1,l}+w_{k,l} z_{k,l}+z_{k,l} y_{k+1,l}=0, \label{eq:D3c}
 \\
 v_{k+1,l+1}+y_{k+1,l} w_{k,l+1}+z_{k,l} w_{k,l+1}+z_{k,l} y_{k+1,l}=0. \label{eq:D3d}
 \end{gather}
 \end{subequations}

 From the system \eqref{eq:D3sys} we have four dif\/ferent way for calculating $z_{k,l}$. This means that we have some compatibility conditions. Indeed from \eqref{eq:D3a} and \eqref{eq:D3c} we obtain the following equation for $v_{k+1,l}$:
 \begin{gather}
 v_{k+1,l} = \frac{(w_{k,l}+y_{k+1,l}) v_{k,l}}{w_{k,l}+y_{k,l}} +\frac{(y_{k,l}-y_{k+1,l}) w_{k,l}^2}{w_{k,l}+y_{k,l}}, \label{eq:evk}
 \end{gather}
 while from \eqref{eq:D3b} and \eqref{eq:D3d} we obtain the following equation for $v_{k+1,l+1}$:
 \begin{gather}
 v_{k+1,l+1} = \frac{(w_{k,l+1}+y_{k+1,l}) v_{k,l+1}}{w_{k,l+1}+y_{k,l}} +\frac{(y_{k,l}-y_{k+1,l}) w_{k,l+1}^2}{w_{k,l+1}+y_{k,l}}. \label{eq:evkp}
 \end{gather}

Equations \eqref{eq:evk} and \eqref{eq:evkp} give rise to a compatibility condition between $v_{k+1,l}$ and its shift in the $l$ direction $v_{k+1,l+1}$ which is given by
 \begin{gather*}
 \left(
 \begin{matrix}
y_{k,l} w_{k,l+1}+y_{k+1,l+1} w_{k,l+1}+y_{k+1,l+1} y_{k,l} \\
 -y_{k,l+1} w_{k,l+1}-y_{k+1,l} w_{k,l+1}-y_{k+1,l} y_{k,l+1}
 \end{matrix}
 \right) \big(v_{k,l+1}-w_{k,l+1}^2\big)=0. %\label{eq:vkvkpcc}
 \end{gather*}
 Discarding the trivial solution $v_{k,l}=w_{k,l}^{2}$ we obtain the following value for the f\/ield $w_{k,l}$
 \begin{gather}
 w_{k,l} = -\frac{y_{k+1,l} y_{k,l-1}-y_{k+1,l-1} y_{k,l}}{y_{k+1,l}+y_{k,l-1}-y_{k+1,l-1}-y_{k,l}}, \label{eq:D3wkldef}
 \end{gather}
 which makes \eqref{eq:evk} and \eqref{eq:evkp} compatible. Equation \eqref{eq:D3wkldef} gives $w_{k,l}$ in terms of $y_{k,l}$ alone and therefore it is the f\/irst part of the solution represented by~\eqref{eq:D3wsolfin}. Inserting equation~\eqref{eq:D3wkldef} into~\eqref{eq:evk} we are left with the following equation for~$v_{k,l}$
 \begin{gather}
 v_{k+1,l} = \frac{y_{k+1,l}-y_{k+1,l-1} }{y_{k,l}-y_{k,l-1}}v_{k,l} +\frac{(y_{k+1,l-1} y_{k,l}-y_{k+1,l} y_{k,l-1})^2}{ (y_{k+1,l-1}+y_{k,l}-y_{k+1,l}-y_{k,l-1}) (y_{k,l}-y_{k,l-1})}. \label{eq:evk2}
 \end{gather}
Applying the transformation
 \begin{gather}
 v_{k,l} = ( y_{k,l}-y_{k,l-1} )V_{k,l} + y_{k,l-1}^{2} \label{eq:D3Vkldef}
 \end{gather}
 we can simplify \eqref{eq:evk2} to the equation
 \begin{gather}
 V_{k+1,l}=V_{k,l}+ \frac{(y_{k,l-1}-y_{k+1,l-1})^2}{y_{k+1,l-1}+y_{k,l}-y_{k+1,l}-y_{k,l-1}}. \label{eq:evk3}
 \end{gather}

To go further we need to specify the form of the f\/ield~$y_{k,l}$. This can be obtained from the Darboux integrability of the $D_{3}$ equation~\eqref{eq:D3}. From \cite{GSY_DarbouxI} we know that the $D_{3}$ equation \eqref{eq:D3} possesses the following $W_{1}$ four-point, third-order integral
 \begin{gather}
 W_{1}={ \Fppp}{ \alpha} {\frac { ( u_{{n+1,m}}-u_{{n-1,m}}) ( u_{{n+2,m}}-u_{{n,m}}) }{u_{n+1,m}^{2}-u_{{n,m}}}}\nonumber \\
\hphantom{W_{1}=}{} +{ \Fpmm}{ \alpha} {\frac { ( u_{{n+1,m}}-u_{{n-1,m}} ) ( u_{n+2,m}-u_{{n,m}}) }{u_{{n,m}}+u_{{n-1,m}}}}\nonumber\\
 \hphantom{W_{1}=}{}-{ \Fmpm}{ \beta} \frac { ( u_{{n+1,m}}-u_{{n-1,m}} )(u_{{n+2,m}} -u_{{n,m}}) }{u_{{n+1,m}}- u_{n,m}^{2}}\nonumber \\
\hphantom{W_{1}=}{}+{ \Fmmp}{ \beta}\frac {( u_{{n+1,m}}-u_{{n-1,m}})( u_{{n+2,m}}-u_{{n,m}} ) }{u_{{n+1,m}}+u_{{n+2,m}}}. \label{eq:W1D3}
 \end{gather}
 Consider now the equation $W_{1}=\xi_{n}$ with $W_{1}$ given as in~\eqref{eq:W1D3}. This relation def\/ines a third-order, four-point ordinary dif\/ference equation in the $n$ direction depending parametrically on~$m$. In particular if we choose the case when $m=2l+1$ we have the equation
 \begin{gather}
 \Fp{n} \frac { ( u_{{n+1,2l+1}}-u_{{n-1,2l+1}})
 ( u_{{n+2,2l+1}}-u_{{n,2l+1}} ) }{u_{{n,2l+1}}+u_{{n-1,2l+1}}}\nonumber \\
 \qquad{} +\Fm{n}\frac { ( u_{{n+1,2l+1}}-u_{{n-1,2l+1}} ) ( u_{{n+2,2l+1}}-u_{{n,2l+1}} ) }{u_{{n+1,2l+1}}+u_{{n+2,2l+1}}} =\xi_{n}, \label{eq:D3ePm}
 \end{gather}
where we have taken without loss of generality $\alpha=\beta=1$. Using the transformation \eqref{eq:genautsublp} equation~\eqref{eq:D3ePm} is converted into the system
 \begin{subequations}\label{eq:D3ePmsys}
 \begin{gather}
 (y_{k+1,l}-y_{k,l})(z_{k,l}-z_{k-1,l}) =\xi_{2 k} (y_{k,l}+z_{k-1,l}), \label{eq:D3ePmk} \\
 (y_{k+1,l}-y_{k,l}) (z_{k+1,l}-z_{k,l})=\xi_{2 k+1}(y_{k+1,l}+z_{k+1,l}). \label{eq:D3ePmkp}
 \end{gather}
 \end{subequations}
 This system is nonlinear, but if we solve \eqref{eq:D3ePmkp} with respect to $z_{k+1,l}$ and we substitute it along with its shift in the $k$ direction into \eqref{eq:D3ePmk} we obtain a \emph{linear second-order ordinary difference equation} involving \emph{only} the f\/ield $y_{k,l}$
 \begin{gather}
 \xi_{{2 k-1}}y_{{k+1,m}} - ( \xi_{{2 k}}+ \xi_{{2 k-1}} ) y_{{k,m}} +\xi_{{2 k}}y_{{k-1,m}} +\xi_{{2 k}}\xi_{{2 k-1}} =0. \label{eq:D3ykleq}
 \end{gather}
 We can lower the order of equation \eqref{eq:D3ykleq} by one using the potential transformation
 \begin{gather}
 Y_{k,l} = y_{k+1,l}-y_{k,l}. \label{eq:D3Ykldef}
 \end{gather}
 Then $Y_{k,l}$ solves the equation
 \begin{gather*}
 Y_{k,l}-\frac{\xi_{2k}}{\xi_{2k-1}}Y_{k-1,l}+\xi_{2k}=0. %\label{eq:D3Ykleq}
 \end{gather*}
 Imposing that
 \begin{gather*}
 \xi_{2k}= -a_{k}( b_{k}-b_{k-1} ), \qquad \xi_{2k-1} = -a_{k-1} ( b_{k}-b_{k-1} ) %\label{eq:D3xidef}
 \end{gather*}
 we obtain that $Y_{k,l}$ can be expressed as
 \begin{gather*}
 Y_{k,l} = a_{k} ( b_{k}+\alpha_{l}). %\label{eq:D3Yklsol}
 \end{gather*}
 From \eqref{eq:D3Ykldef} we have that
 \begin{gather*}
 y_{k+1,l}-y_{k,l}= a_{k}( b_{k}+\alpha_{l}). %\label{eq:D3ykleq2}
 \end{gather*}
 Setting
 \begin{gather*}
 a_{k} =c_{k+1}-c_{k}, \qquad b_{k} = \frac{d_{k+1}-d_{k}}{c_{k+1}-c_{k}}, %\label{eq:D3abdef}
 \end{gather*}
 we have that the solution of is given by
 \begin{gather}
 y_{k,l} = \alpha_{l} c_{k} + d_{k} + \beta_{l}, \label{eq:D3yklsol}
 \end{gather}
 i.e., by equation \eqref{eq:D3ysolfin}.

 Inserting now the obtained value of $y_{k,l}$ from \eqref{eq:D3yklsol} into the equation \eqref{eq:evk3} we obtain
 \begin{gather*}
 V_{k+1,l}=V_{k,l} -\frac{(d_{k}+\alpha_{l-1} c_{k} -d_{k+1}-\alpha_{l-1} c_{k+1})^2}{ (\alpha_{l}-\alpha_{l-1}) (c_{k+1}-c_{k})}. %\label{eq:evk4}
 \end{gather*}
 This means that we can write down the following solution for $V_{k,l}$
 \begin{gather}
 V_{k,l} = \gamma_{l} -\alpha_{l-1} \frac{\alpha_{l-1} c_{k}+2 d_{k}}{ \alpha_{l}-\alpha_{l-1}} +\frac{e_{k}}{\alpha_{l}-\alpha_{l-1}}, \label{eq:D3Vklsol}
 \end{gather}
 up to a discrete integration for the function $e_{k}$
 \begin{gather*}
 e_{k+1}=e_{k} -\frac{( d_{k+1}-d_{k})^{2}}{c_{k+1}-c_{k}}, %\label{eq:D3ekeqdef}
 \end{gather*}
i.e., up to the condition \eqref{eq:D3ekeqdef0}. Substituting the value of $V_{k,l}$ from~\eqref{eq:D3Vklsol} and of $y_{k,l}$ from~\eqref{eq:D3yklsol} into equation \eqref{eq:D3Vkldef} we have that $v_{k,l}$ is given by equation~\eqref{eq:D3vsolfin}. Plugging the obtained value of~$v_{k,l}$ we can compute $w_{k,l}$ from~\eqref{eq:D3wkldef}. Finally we can compute $z_{k,l}$ from the original system~\eqref{eq:D3sys} and we obtain a~\emph{single} compatibility condition given by
 \begin{gather}
 (\alpha_{{l+1}} -\alpha_{{l}} ) \gamma_{{l+1}} - (\alpha_{{l}}-\alpha_{{l-1}} ) \gamma_{{l}} -\alpha_{{l-1}}\beta_{{l-1}}+\alpha_{{l}}\beta_{{l}}
 +\alpha_{{l}}\beta_{{l-1}}-\alpha_{{l-1}}\beta_{{l}}=0, \label{eq:D3gammleq}
 \end{gather}
i.e., just by \eqref{eq:D3gammleq0}. Given this conditions all the equations in \eqref{eq:D3sys} are compatible and $z_{k,l}$ is indif\/ferently given by solving one of the equation. E.g., solving~\eqref{eq:D3a} we can say that $z_{k,l}$ is given by equation~\eqref{eq:D3zsolfin}. This ends the procedure of solution of the $D_{3}$ equation~\eqref{eq:D3}.
\end{proof}

\begin{Proposition} The $_{1}D_{4}$ equation \eqref{eq:1D4} is exactly solvable. We have that the expression of the fields $y_{k,l}$ and $v_{k,l}$ is given by
 \begin{subequations} \label{eq:1D4yvsolfin}
 \begin{gather}
 y_{k,l} = c_{k}( \alpha_{l}d_{k}+\beta_{l} ), \label{eq:1D4ysolfin} \\
 v_{k,l} = c_{k}[ (\alpha_{l}-\alpha_{l-1})d_{k}
 +\beta_{l}-\beta_{l-1} ] \left\{ \gamma_{l} +\frac{ \delta_{1} \delta_{3}}{ \beta_{l-1} \alpha_{l}-\beta_{l} \alpha_{l-1} }
 \left[\frac{\alpha_{l-1}}{c_{k}^2 (\alpha_{l-1} d_{k}+\beta_{l-1})} +e_{k}\right] \right\}\nonumber \\
\hphantom{v_{k,l} =}{} +\frac{\delta_{1}\delta_{3}}{c_{k}\left( \alpha_{l-1}d_{k}+\beta_{l-1} \right)}
 -\delta_{2}c_{k}\left( \alpha_{l-1}d_{k}+\beta_{l-1} \right), \label{eq:1D4vsolfin}
 \end{gather}
 \end{subequations}
 where the function functions $c_{k}$, $d_{k}$, $\alpha_{l}$ and $\beta_{l}$ are arbitrary functions of their arguments, whereas the function~$e_{k}$ and~$\gamma_{l}$ are given through the discrete integrations
 \begin{subequations} \label{eq:1D4discrint}
 \begin{gather}
 e_{k+1}-e_{k} =-\frac{(c_{k}-c_{k+1})^2}{ c_{k}^2 c_{k+1}^2 (d_{k+1}-d_{k})}, \label{eq:1D4ekeqdef0} \\
 ( \beta_{{l}}\alpha_{{l+1}}-\beta_{{l+1}}\alpha_{{l}} ) \gamma _{{l+1}} - ( \beta_{{l-1}}\alpha_{{l}}-\beta_{{l}}\alpha_{{l-1}} )
 \gamma _{{l}} = ( \beta_{{l-1}}\alpha_{{l}}-\beta_{{l}}\alpha_{{l-1}} ) {\delta_2}. \label{eq:1D4gammleq0}
 \end{gather}
 \end{subequations}
 The fields $z_{k,l}$ and $w_{k,l}$ are then given in terms of $y_{k,l}$ and $v_{k,l}$ as
 \begin{subequations} \label{eq:1D4zwsolfin}
 \begin{gather}
 z_{k,l} = - \frac{\left[
 \begin{gathered}
 \delta_{3}\left( y_{k+1,l} y_{k,l-1}-y_{k+1,l-1} y_{k,l} \right) \\
 + \delta_{1}\delta_{3}y_{k,l}( y_{k+1,l-1}-y_{k+1,l}-y_{k,l-1}+y_{k,l} ) \end{gathered} \right] }{ \left[
 \begin{gathered}
 ( v_{k,l}+\delta_{2}y_{k,l} ) ( y_{k+1,l} y_{k,l-1}-y_{k+1,l-1} y_{k,l}) \\
 + \delta_{1}\delta_{3}( y_{k+1,l-1}-y_{k+1,l}-y_{k,l-1}+y_{k,l}) \end{gathered} \right] },
 \label{eq:1D4zsolfin} \\
 w_{k,l} = \delta_{3}\frac{y_{k+1,l-1}-y_{k+1,l}-y_{k,l-1}+y_{k,l}}{ y_{k+1,l} y_{k,l-1}-y_{k+1,l-1} y_{k,l}}.
 \label{eq:1D4wsolfin}
 \end{gather}
 \end{subequations}
\end{Proposition}

\begin{Remark} We remark that we can say that equation \eqref{eq:1D4gammleq0} def\/ines a discrete integration for the function $\gamma_{l}$ since it can be expressed in the form \eqref{eq:discrintdef}. Def\/ining a new function $\zeta_{l}$ by
 \begin{gather*}
 \gamma_{l} = \frac{\zeta_{l} \delta_{2}}{ \beta_{l-1} \alpha_{l}-\beta_{l} \alpha_{l-1}} %\label{eq:1D4gammaldef}
 \end{gather*}
 we have that $\zeta_{l}$ is given by the following dif\/ference equation
 \begin{gather}
 \zeta_{l+1} - \zeta_{l}= \alpha_{l} \beta_{l-1}-\alpha_{l-1} \beta_{l}, \label{eq:1D4deltaeqdef}
 \end{gather}
 i.e., it is given by a discrete integration. Note that \eqref{eq:1D4deltaeqdef} is exactly the same as \eqref{eq:D3deltaeqdef}.
\end{Remark}

\begin{proof} The proof of the solution of the $_{1}D_{2}$ equation \eqref{eq:1D4} proceeds as the one outlined in Proposition~\ref{prop:D3}. The interested reader can f\/ind the details in Appendix~\ref{app:remsols}.
\end{proof}

\begin{Proposition} The $_{2}D_{4}$ equation \eqref{eq:2D4} is exactly solvable. We have that the expression of the fields $y_{k,l}$ and $v_{k,l}$ is given by
 \begin{subequations} \label{eq:2D4yvsolfin}
 \begin{gather}
 y_{k,l}= c_{k}( \alpha_{l}d_{k}+\beta_{l} ), \label{eq:2D4ysolfin} \\
 v_{k,l} = c_{k}[ (\alpha_{l}-\alpha_{l-1})d_{k} +\beta_{l}-\beta_{l-1} ] \nonumber \\
 \hphantom{v_{k,l} =}{}\times
 \left\{\gamma_{l} -\frac{\alpha_{l} \delta_{1} \delta_{3}}{ (\alpha_{l} d_{k}+\beta_{l}) c_{k}^2 (\beta_{l} \alpha_{l-1}-\beta_{l-1} \alpha_{l})}
 +\frac{\delta_{3} \delta_{1} e_{k}}{ \beta_{l} \alpha_{l-1}-\beta_{l-1} \alpha_{l}} \right\}\nonumber \\
\hphantom{v_{k,l} =}{} +\frac{\delta_{1}\delta_{3}}{c_{k} ( \alpha_{l}d_{k}+\beta_{l} )} -\delta_{2}c_{k}( \alpha_{l}d_{k}+\beta_{l} ),\label{eq:2D4vsolfin}
 \end{gather}
 \end{subequations}
 where the function functions $c_{k}$, $d_{k}$, $\alpha_{l}$ and $\beta_{l}$ are arbitrary functions of their arguments, whereas the function $e_{k}$ and $\gamma_{l}$ are given through the discrete integrations
 \begin{subequations} \label{eq:2D4discrint}
 \begin{gather}
 e_{k+1}-e_{k} =\frac{(c_{k+1}-c_{k})^2}{ (d_{k+1}-d_{k}) c_{k}^2 c_{k+1}^2}, \label{eq:2D4ekeqdef0} \\
 (\beta_{l} \alpha_{l+1}-\beta_{l+1} \alpha_{l}) \gamma_{l+1} -(\beta_{l-1} \alpha_{l}-\beta_{l} \alpha_{l-1}) \gamma_{l}
 =(\beta_{l} \alpha_{l+1}-\beta_{l+1} \alpha_{l}) \delta_{2}. \label{eq:2D4gammleq0}
 \end{gather}
 \end{subequations}
 The fields $z_{k,l}$ and $w_{k,l}$ are then given in terms of $y_{k,l}$ and $v_{k,l}$ as
 \begin{subequations} \label{eq:2D4zwsolfin}
 \begin{gather}
 z_{k,l} = -\frac{1}{\delta_{1}}
 \frac{\left[ \begin{gathered}
 \delta_{1}\delta_{3}( y_{k,l-1}+y_{k+1,l}-y_{k,l}-y_{k+1,l-1} )
 \\
 + ( v_{k,l}+\delta_{2}y_{k,l} )
( y_{k+1,l-1} y_{k,l}-y_{k+1,l} y_{k,l-1} )
 \end{gathered}
 \right]}{ \left[
 \begin{gathered} y_{k+1,l-1} y_{k,l}-y_{k+1,l} y_{k,l-1} \\
 +y_{k,l}( y_{k,l-1}+y_{k+1,l}-y_{k,l}-y_{k+1,l-1})
 \end{gathered}
 \right]}, \label{eq:2D4zsolfin} \\
 w_{k,l} = \frac{1}{\delta_{1}}
 \frac{y_{k+1,l-1} y_{k,l}-y_{k+1,l} y_{k,l-1}}{ y_{k,l-1}+y_{k+1,l}-y_{k,l}-y_{k+1,l-1}}. \label{eq:2D4wsolfin}
 \end{gather}
 \end{subequations}
\end{Proposition}

\begin{Remark} We remark that we can say that equation \eqref{eq:1D4gammleq0} def\/ines a discrete integration for the function $\gamma_{l}$ since it can be expressed in the form \eqref{eq:discrintdef}. Def\/ining a new function $\zeta_{l}$ by
 \begin{gather*}
 \gamma_{l} = \frac{\zeta_{l} \delta_{2}}{ \beta_{l-1} \alpha_{l}-\beta_{l} \alpha_{l-1}} %\label{eq:2D4gammaldef}
 \end{gather*}
 we have that $\zeta_{l}$ is given by the following dif\/ference equation
 \begin{gather*}
 \zeta_{l+1} -\zeta_{l}= \alpha_{l+1} \beta_{l}-\alpha_{l} \beta_{l+1}, %\label{eq:2D4deltaeqdef}
 \end{gather*}
 i.e., it is given by a discrete integration.
\end{Remark}

\begin{proof} The proof of the solution of the $_{2}D_{2}$ equation \eqref{eq:2D4} proceeds as the one outlined in Proposition~\ref{prop:D3}. The interested reader can f\/ind the details in Appendix~\ref{app:remsols}.
\end{proof}

\subsection[The ${}_{t}H^{\varepsilon}_{2}$ and the ${}_{t}H^{\varepsilon}_{3}$ equations]{The $\boldsymbol{{}_{t}H^{\varepsilon}_{2}}$ and the $\boldsymbol{{}_{t}H^{\varepsilon}_{3}}$ equations}\label{sec:trapH4}

In this subsection we construct a general solution of the \tHeq{2} and the \tHeq{3} equations. As we recalled in the introduction, the solution of the \tHeq{1} through the f\/irst integrals was already presented in \cite{GSY_DarbouxI}, so we will not discuss it again. Moreover we also recall that the general solution of the \tHeq{1} equation was f\/irst found in \cite{GSL_general,GSL_Gallipoli15} without the knowledge of the f\/irst integrals. The f\/irst integrals of the \tHeq{1} equation were f\/irst presented in~\cite{GSL_Pavel}.

The procedure we will follow will make use of the f\/irst integrals, in a similar way than in the cases presented in Section~\ref{sec:otherH6}. The main dif\/ference is in the fact that the $H^{4}$ are non-autonomous only in the direction $m$, i.e., they depend only on the non-autonomous factors~$F_{m}^{(\pm)}$ as given by \eqref{eq:fk}. Therefore instead of the general transformation~\eqref{eq:genautsub} we can use the simplif\/ied transformation
\begin{gather}
 u_{n, 2l} = p_{n, l},\qquad u_{n, 2l+1} = q_{n, l}. \label{eq:h4aut}
\end{gather}
Then to describe the general solution of a $H^{4}$ we only need three arbitrary functions: one in the~$n$ direction and two in the $m$ direction.

We have then that the following propositions hold true:

\begin{Proposition} \label{prop:tH2e} The \tHeq{2}\ equation \eqref{eq:tH2e} is exactly solvable. If $\varepsilon\neq0$ and the field $q_{n,l}$ do not satisfy the discrete wave equation
 \begin{gather}
 q_{{n+1,l+1}}+q_{{n,l}}=q_{{n+1,l}}+q_{{n,l+1}}, \label{eq:tH2qsing}
 \end{gather}
 then the solution of the \tHeq{2}\ equation \eqref{eq:tH2e} is given by
 \begin{subequations} \label{eq:tH2soldefgen}
 \begin{gather}
 q_{n,l} = \beta_{l} + \frac{\gamma_{l}+\zeta_{l}e_{n} + f_{n}}{c_{n}+\zeta_{l}}, \label{eq:tH2qnlsoldefgen} \\
 p_{{n,l}}=\frac { \left\{
 \begin{gathered}( q_{{n,l}}-{ \alpha_3} ) q_{{n+1,l-1}} -( q_{{n,l-1}}-{ \alpha_3} ) q_{{n+1,l}} \\
 {}-( { \alpha_2}+{ \alpha_3}) ( q_{{n,l}}-q_{{n,l-1}}) -\varepsilon\alpha_3^{2}( q_{{n,l}}-q_{{n,l-1}} ) \\
 {}+\varepsilon \big[ {{ \alpha_3}}^{2}
 + 2\alpha_3( q_{{n,l-1}}+ { \alpha_2}) -( q_{{n,l}}-q_{{n,l-1}}) q_{{n+1,l-1}} \big]q_{{n+1,l}} \\
 {}+\varepsilon ( { \alpha_2}+q_{{n,l}}) ( { \alpha_2}+q_{{n,l-1}} ) q_{{n+1,l}} \\
 -\varepsilon q_{{n+1,l-1}}
 \big[ {{ \alpha_3}}^{2}-2( q_{{n,l}}+ { \alpha_2}){\alpha_3}
 + ( { \alpha_2}+q_{{n,l}} )
( { \alpha_2}+q_{{n,l-1}} ) \big]
 \end{gathered}\right\}
 }{ q_{{n+1,l}}-q_{{n,l}}+q_{{n,l-1}}-q_{{n+1,l-1}}}, \label{eq:tH2pnlsoldefgen}
 \end{gather}
 \end{subequations}
 where $c_{n}$, $\zeta_{l}$ and $\beta_{l}$ are arbitrary functions of their arguments and $e_{n}$ is a solution of the equation
 \begin{gather}
 e_{n+1} -\frac{c_{n+1}-c_{n-1}}{c_{n}-c_{n-1}} e_{n} +\frac{c_{n+1}-c_{n}}{c_{n}-c_{n-1}}e_{n-1} - \alpha_2 \frac{c_{n+1}-c_{n-1}}{c_{n}-c_{n-1}} =0, \label{eq:tH2eeq0}
 \end{gather}
 while $f_{n}$ and $\gamma_{l}$ are given by the discrete integrations
 \begin{subequations}\label{eq:tH2discrint0}
 \begin{gather}
 f_{n}-f_{n-1}=e_{n}c_{n-1}-c_{n} e_{n-1}, \label{eq:tH2feq0} \\
 \gamma_{l+1}-\gamma_{l}= -(\zeta_{l+1}-\zeta_{l}) \left(\alpha_2+\beta_{l+1}+\beta_{l}+2\alpha_3-\frac{1}{\epsilon}\right). \label{eq:tH2gameq0}
 \end{gather}
 \end{subequations}
 If $\varepsilon=0$, but the field $q_{n,l}$ do not satisfy the discrete wave equation \eqref{eq:tH2singcc} then the solution of the \tHeq{2}\ equation~\eqref{eq:tH2e} is given by
 \begin{subequations} \label{eq:tH2soldef0}
 \begin{gather}
 q_{n,l} = \beta_{l} + \frac{\gamma_{l}+\zeta_{0}e_{n} + f_{n}}{c_{n}+\zeta_{0}}, \label{eq:tH2qnlsoldef0} \\
 p_{{n,l}}=\frac { \left[
 \begin{gathered}
( q_{{n,l}}-{ \alpha_3} ) q_{{n+1,l-1}} -( q_{{n,l-1}}-{ \alpha_3} ) q_{{n+1,l}} \\
 {}-( { \alpha_2}+{ \alpha_3}) ( q_{{n,l}}-q_{{n,l-1}}) \end{gathered}\right]
 }{ q_{{n+1,l}}-q_{{n,l}}+q_{{n,l-1}}-q_{{n+1,l-1}}}, \label{eq:tH2pnlsoldef0}
 \end{gather}
 \end{subequations}
where $c_{n}$, $\beta_{l}$ and $\gamma_{l}$ are arbitrary functions of their arguments, $\zeta_{0}$ is a constant and $e_{n}$ is a~solution of~\eqref{eq:tH2eeq0} and $f_{n}$ is a solution of \eqref{eq:tH2feq0}. If the field $q_{n,l}$ satisfies the discrete wave equation~\eqref{eq:tH2qsing}
 regardless of the value of the parameter $\varepsilon$ the solution of the \tHeq{2}\ equation \eqref{eq:tH2e} is given by
 \begin{subequations} \label{eq:tH2singdef}
 \begin{gather}
 q_{n,l} = a_{n} +\zeta_{0}, \label{eq:tH2qnlsingdef} \\
 p_{n,l} = b_{n} ( \beta_{l}+c_{n} ), \label{eq:tH2pnlsingdef}
 \end{gather}
 \end{subequations}
 where $b_{n}$ and $\beta_{l}$ are arbitrary functions of their arguments, $\zeta_{0}$ is a constant and $a_{n}$ and $c_{n}$ are given by the discrete
 integration
 \begin{subequations} \label{eq:tH2discrintsing}
 \begin{gather}
 {\frac {a_{{n+1}}-a_{{n}} +{ \alpha_2}}{a_{{n+1}}-a_{{n}}-{ \alpha_2}}}= \frac{b_{n+1}}{b_{n}}, \label{eq:tH2aeq0} \\
 c_{n+1} - c_{n} = {\frac { ( a_{{n}}+\zeta_{0}-{ \alpha_3} ) b_{{n+1}} +b_{{n}} ( { \alpha_2}+{ \alpha_3}-\zeta_{0}-a_{{n}} ) }{ b_{{n}}b_{{n+1}}}}\nonumber \\
\hphantom{c_{n+1} - c_{n} =}{} -\epsilon{\frac { \big[ ( { \alpha_2}+\zeta_{0}+a_{{n}}+{ \alpha_3} ) ^{2}b_{{n+1}}-b_{{n}} ( { \alpha_3}+a_{{n}}+\zeta_{0} ) ^{2} \big] }{b_{{n}}b_{{n+1}}}}. \label{eq:tH2ceq0}
\end{gather}
 \end{subequations}
\end{Proposition}

\begin{Remark} We remark that the function $e_{n}$ can be obtained from \eqref{eq:tH2eeq0} as the result of \emph{two} discrete integrations. Indeed def\/ining
 \begin{gather}
 E_{n} = \frac{e_{n+1}-e_{n}}{c_{n+1}-c_{n}}, \label{eq:tH2Dndef}
 \end{gather}
and substituting in \eqref{eq:tH2eeq0} we obtain that $E_{n}$ must solve the equation
 \begin{gather}
 E_{n}-E_{n-1}=\alpha_2 \left(\frac{1}{c_{n+1}-c_{n}}+\frac{1}{c_{n}-c_{n-1}} \right). \label{eq:tH2Dneq}
 \end{gather}
Note that the right-hand side of \eqref{eq:tH2Dneq} is not a total dif\/ference. So the function $e_{n}$ can be obtained by integrating \eqref{eq:tH2Dneq} and subsequently integrating \eqref{eq:tH2Dndef}. This provides the value of $e_{n}$. The obtained value can be plugged in \eqref{eq:tH2feq0} to give $f_{n}$ after discrete integration. This reasoning shows that we can obtain the non-arbitrary functions $e_{n}$ and $f_{n}$ as result of a f\/inite number of discrete integrations. Therefore we can conclude that the solution of the \tHeq{2}\ equation~\eqref{eq:tH2e} in the general case is given in terms of four discrete integrations. If $\varepsilon=0$ then the general solution is given in terms of three discrete integrations and f\/inally in the singular case, when~$q_{n,l}$ solves the discrete wave equation \eqref{eq:tH2qsing}, we need only two discrete integrations.
\end{Remark}

\begin{proof} We start the procedure of solution of the \tHeq{2}\ equation \eqref{eq:tH2e} by looking at the equation itself. We apply the transformation \eqref{eq:h4aut} to the \tHeq{2} equation \eqref{eq:tH2e} and we obtain the following system of two coupled equations
 \begin{subequations} \label{eq:tH2sys}
 \begin{gather}
 (p_{n,l}-p_{n+1,l}) (q_{n,l}-q_{n+1,l})-\alpha_{2} (p_{n,l}+p_{n+1,l}+q_{n,l}+q_{n+1,l})\nonumber \\
 \qquad{} +\frac{\epsilon \alpha_{2}}{2} (2 q_{n,l}+2 \alpha_{3}+\alpha_{2}) (2 q_{n+1,l}+2 \alpha_{3}+\alpha_{2})
 +\frac{\epsilon \alpha_{2}}{2}(2 \alpha_{3}+\alpha_{2})^2+(\alpha_{2}+\alpha_{3})^2\nonumber \\
 \qquad{} -\alpha_{3}^2-2 \epsilon \alpha_{2} \alpha_{3} (\alpha_{2}+\alpha_{3})=0, \label{eq:tH2l} \\
 (q_{n,l}-q_{n+1,l}) (p_{n,l+1}-p_{n+1,l+1})-\alpha_{2} (q_{n,l}+q_{n+1,l}+p_{n,l+1}+p_{n+1,l+1})\nonumber \\
 \qquad{} +\frac{\epsilon \alpha_{2}}{2} (2 q_{n,l}+2 \alpha_{3}+\alpha_{2}) (2 q_{n+1,l}+2 \alpha_{3}+\alpha_{2})+\frac{\epsilon \alpha_{2}}{2} (2 \alpha_{3}+\alpha_{2})^2 +(\alpha_{2}+\alpha_{3})^2-\alpha_{3}^2\nonumber \\
 \qquad{} -2 \epsilon \alpha_{2} \alpha_{3} (\alpha_{2}+\alpha_{3})=0. \label{eq:tH2lp}
 \end{gather}
 \end{subequations}
 We have that equation \eqref{eq:tH2l} depends on $p_{n,l}$ and $p_{n+1,l}$ and that equation \eqref{eq:tH2lp} depends on~$p_{n,l+1}$ and $p_{n+1,l+1}$.
 So we apply the translation operator $T_{l}$ to \eqref{eq:tH2l} to obtain two equations in terms of $p_{n,l+1}$ and $p_{n+1,l+1}$
 \begin{subequations} \label{eq:tH2sys2}
 \begin{gather}
 (p_{n,l+1}-p_{n+1,l+1}) (q_{n,l+1}-q_{n+1,l+1})-\alpha_{2} (p_{n,l+1}+p_{n+1,l+1}+q_{n,l+1}+q_{n+1,l+1})\nonumber \\
 \qquad{} +\frac{\epsilon \alpha_{2}}{2} (2 q_{n,l+1}+2 \alpha_{3}+\alpha_{2}) (2 q_{n+1,l+1}+2 \alpha_{3}+\alpha_{2})+\frac{\epsilon \alpha_{2}}{2}(2 \alpha_{3}+\alpha_{2})^2+(\alpha_{2}+\alpha_{3})^2\nonumber \\
 \qquad{} -\alpha_{3}^2-2 \epsilon \alpha_{2} \alpha_{3} (\alpha_{2}+\alpha_{3})=0, \label{eq:tH2l2}\\
 (q_{n,l}-q_{n+1,l}) (p_{n,l+1}-p_{n+1,l+1})-\alpha_{2} (q_{n,l}+q_{n+1,l}+p_{n,l+1}+p_{n+1,l+1})\nonumber \\
\qquad{} +\frac{\epsilon \alpha_{2}}{2} (2 q_{n,l}+2 \alpha_{3}+\alpha_{2}) (2 q_{n+1,l}+2 \alpha_{3}+\alpha_{2})+\frac{\epsilon \alpha_{2}}{2} (2 \alpha_{3}+\alpha_{2})^2 +(\alpha_{2}+\alpha_{3})^2-\alpha_{3}^2\nonumber \\
\qquad{} -2 \epsilon \alpha_{2} \alpha_{3} (\alpha_{2}+\alpha_{3})=0. \label{eq:tH2lp2}
 \end{gather}
 \end{subequations}
 The system \eqref{eq:tH2sys2} is equivalent to the original system \eqref{eq:tH2sys}. We can solve \eqref{eq:tH2sys2} with respect to $p_{n,l+1}$ and $p_{n+1,l+1}$
 \begin{subequations} \label{eq:tH2pnlppnplp}
 \begin{gather}
 p_{{n,l+1}}=\frac { \left\{
 \begin{gathered}
( q_{{n,l+1}}-{ \alpha_3}) q_{{n+1,l}} -( q_{{n,l}}-{ \alpha_3}) q_{{n+1,l+1}} \\
 {} - ( { \alpha_2}+{ \alpha_3} )( q_{{n,l+1}}-q_{{n,l}}) -\varepsilon\alpha_3^{2} ( q_{{n,l+1}}-q_{{n,l}} ) \\
 {}+\varepsilon \big[ {{ \alpha_3}}^{2} + 2\alpha_3 ( q_{{n,l}}+ { \alpha_2}) -( q_{{n,l+1}}-q_{{n,l}}) q_{{n+1,l}}\big]q_{{n+1,l+1}} \\
 {}+\varepsilon ( { \alpha_2}+q_{{n,l+1}} ) ( { \alpha_2}+q_{{n,l}} ) q_{{n+1,l+1}} \\
 -\varepsilon q_{{n+1,l}} \big[ {{ \alpha_3}}^{2}-2 ( q_{{n,l+1}}+ { \alpha_2} ){\alpha_3} + ( { \alpha_2}+q_{{n,l+1}})( { \alpha_2}+q_{{n,l}} ) \big]
 \end{gathered}\right\}
 }{ q_{{n+1,l+1}}-q_{{n,l+1}}+q_{{n,l}}-q_{{n+1,l}}}, \label{eq:tH2pnlp}
 \\
 p_{{n+1,l+1}}={\frac {\left\{
 \begin{gathered}
( q_{{n+1,l}}-{ \alpha_3} ) q_{{n,l+1}} + ( { \alpha_3}-q_{{n+1,l+1}} ) q_{{n,l}} \\
 {} - ( { \alpha_2}+{ \alpha_3} ) ( q_{{n+1,l}}-q_{{n+1,l+1}} ) +\varepsilon\alpha_3^{2} (q_{{n+1,l+1}}-q_{{n+1,l}}) \\
 {}+\varepsilon q_{{n,l+1}} \big[ ( q_{{n+1,l+1}}-q_{{n+1,l}} ) q_{{n,l}} -\alpha_3^{2} - 2 {\alpha_3}(q_{{n+1,l}}+ {\alpha_2} ) \big] \\
 {}-\varepsilon q_{{n,l+1}} ( { \alpha_2}+q_{{n+1,l+1}} )( {\alpha_2}+q_{{n+1,l}} ) \\
 {}+\varepsilon q_{{n,l}}\big[ \alpha_3^{2}
 + 2( { \alpha_2}+ q_{{n+1,l+1}} ) { \alpha_3} + ( {\alpha_2}+q_{{n+1,l+1}}) ( { \alpha_2}+q_{{n+1,l}} ) \big]
 \end{gathered}\right\}
 }{ q_{{n+1,l+1}}-q_{{n,l+1}}+q_{{n,l}}-q_{{n+1,l}}}}. \label{eq:tH2pnplp}
 \end{gather}
 \end{subequations}
 We see that the right-hand sides of \eqref{eq:tH2pnlppnplp} are functions only of $q_{n,l}$, $q_{n+1,l}$, $q_{n,l+1}$ and $q_{n+1,l+1}$ and well def\/ined unless $q_{n,l}$ solves the discrete wave equation \eqref{eq:tH2qsing}, which is therefore a \emph{singular case}. Therefore at this point the procedure of solution bifurcates into two cases. We treat them separately.

{\bf Singular case: $\boldsymbol{q_{n,l}}$ solves \eqref{eq:tH2qsing}.} Let us assume that the f\/ield $q_{n,l}$ satisf\/ies the discrete wave equation in the form \eqref{eq:tH2qsing}. Then as discussed in the introduction the discrete wave equation is a simplest example of Darboux integrable equation and its solution is given by the discrete d'Alembert formula
 \begin{gather}
 q_{n,l} = a_{n}+\zeta_{l}, \label{eq:tH2qsingsol}
 \end{gather}
 where both $a_{n}$ and $\zeta_{l}$ are arbitrary functions of their argument. Substituting \eqref{eq:tH2qsingsol} into \eqref{eq:tH2sys2} we obtain the compatibility condition
 \begin{gather}
 \zeta_{l+1} - \zeta_{l} =0, \label{eq:tH2singcc}
 \end{gather}
 i.e., $\zeta_{l}=\zeta_{0}=\text{const}$ and the system \eqref{eq:tH2sys} is now consistent. This yield us the f\/irst part of the solution of this case \eqref{eq:tH2qnlsingdef}. We are therefore left with one equation for $p_{n,l}$, e.g.,~\eqref{eq:tH2l}. Inserting \eqref{eq:tH2qsingsol} with $\zeta_{l}=\zeta_{0}$ in \eqref{eq:tH2l} and solving with respect to $p_{n+1,l}$ we obtain
 \begin{gather*}
 p_{n+1, l} = {\frac {a_{{n+1}}-a_{{n}} +{ \alpha_2}}{a_{{n+1}}-a_{{n}}-{ \alpha_2}}} p_{n, l} +{\frac {{ \alpha_2}
 ( { \alpha_2}-a_{{n}}+2 { \alpha_3}-2 \zeta_{0}-a_{{n+1}} ) }{{ \alpha_2}+a_{{n}}-a_{{n+1}}}} \\
 \hphantom{ p_{n+1, l} =}{} +{\frac {{ \alpha_2} \epsilon \big[ \alpha_2^{2}+ ( 2 \zeta_{0}+a_{{n
+1}}+2 { \alpha_3}+a_{{n}} ) { \alpha_2}+2 ( a_{{n+1}}+\alpha_
{{0}}+{ \alpha_3} ) ( { \alpha_3}+a_{{n}}+\zeta_{0})\big] }{{ \alpha_2}+a_{{n}}-a_{{n+1}}}} .
%\label{eq:tH2psingeq}
 \end{gather*}
 We can introduce a new function $b_{n}$ through discrete integration
 \begin{gather*}
 {\frac {a_{{n+1}}-a_{{n}} +{ \alpha_2}}{a_{{n+1}}-a_{{n}}-{ \alpha_2}}} = \frac{b_{n+1}}{b_{n}}, %\label{eq:tH2bsingdef}
 \end{gather*}
 which is just formula \eqref{eq:tH2aeq0}. Then we have that $p_{n,l}$ must solve the equation
 \begin{gather}
 \frac{p_{n+1, l}}{b_{n+1}} = \frac{p_{n, l}}{b_{n}}
 +{\frac { ( a_{{n}}+\zeta_{0}-{ \alpha_3} ) b_{{n+1}}
 +b_{{n}} ( { \alpha_2}+{ \alpha_3}-\zeta_{0}-a_{{n}} ) }{ b_{{n}}b_{{n+1}}}}\nonumber \\
\hphantom{\frac{p_{n+1, l}}{b_{n+1}} =}{} -\epsilon
{\frac { \big[ ( { \alpha_2}+\zeta_{0}+a_{{n}}+{ \alpha_3}
) ^{2}b_{{n+1}}-b_{{n}} ( { \alpha_3}+a_{{n}}+\zeta_{0}
) ^{2} \big] }{b_{{n}}b_{{n+1}}}}. \label{eq:tH2psingeq2}
 \end{gather}
 The solution of equation \eqref{eq:tH2psingeq2} is given by
 \begin{gather*}
 p_{n,l} = b_{n} ( \beta_{l}+c_{n}), %\label{eq:tH2psingsol}
 \end{gather*}
 where $c_{n}$ is given by the discrete integration
 \begin{gather*}
 c_{n+1} - c_{n} = {\frac {( a_{{n}}+\zeta_{0}-{ \alpha_3}) b_{{n+1}}
 +b_{{n}} ( { \alpha_2}+{ \alpha_3}-\zeta_{0}-a_{{n}} ) }{ b_{{n}}b_{{n+1}}}}\nonumber \\
\hphantom{c_{n+1} - c_{n} =}{} -\epsilon
{\frac { \big[ ( { \alpha_2}+\zeta_{0}+a_{{n}}+{ \alpha_3}
) ^{2}b_{{n+1}}-b_{{n}} ( { \alpha_3}+a_{{n}}+\zeta_{0}) ^{2} \big] }{b_{{n}}b_{{n+1}}}},%\label{eq:tH2csingdef}
 \end{gather*}
 i.e., through formula \eqref{eq:tH2ceq0}. This yields the solution of the \tHeq{2} equation \eqref{eq:tH2e} when $q_{n,l}$ satisfy the discrete wave equation \eqref{eq:tH2qsing}.

{\bf General case: $\boldsymbol{q_{n,l}}$ do not solve \eqref{eq:tH2qsing}.} When $q_{n,l}$ is not a solution of the discrete wave equation
 \eqref{eq:tH2qsing} the equations~\eqref{eq:tH2pnlppnplp} are well def\/ined. Moreover we have that \eqref{eq:tH2pnlp} and \eqref{eq:tH2pnplp} must be compatible. To impose the compatibility condition we apply $T_{l}^{-1}$ to \eqref{eq:tH2pnplp} and we impose to the obtained expression to be equal to \eqref{eq:tH2pnlp}. We f\/ind that $q_{n,l}$ must solve the following equation
 \begin{gather}
 { \alpha_2}
( q_{{n-1,l+1}}-q_{{n-1,l}}-q_{{n+1,l+1}}+q_{{n+1,l}}) + ( q_{{n,l}}-q_{{n+1,l}} ) q_{{n-1,l+1}}\nonumber\\
\qquad{} - ( q_{{n,l}}-q_{{n-1,l}} ) q_{{n+1,l+1}} +q_{{n,l+1}}( q_{{n+1,l}}-q_{{n-1,l}} )\nonumber\\
\qquad{} +\varepsilon\alpha_2^{2} (q_{{n+1,l+1}} -q_{{n+1,l}}+q_{{n-1,l}}-q_{{n-1,l+1}}) \nonumber\\
 \qquad{} +\varepsilon\alpha_2( q_{{n+1,l+1}}-q_{{n+1,l}}+q_{{n-1,l}}-q_{{n-1,l+1}} ) ( q_{{n,l+1}}+2{ \alpha_3}+q_{{n,l}} )\nonumber \\
 \qquad{} +\varepsilon ( q_{{n+1,l}}-q_{{n-1,l}} ) q_{{n+1,l+1}} q_{{n-1,l+1}} +\varepsilon ( q_{{n,l+1}}-q_{{n,l}}+q_{{n+1,l}} ) q_{{n-1,l}}q_{{n-1,l+1}} \nonumber\\
 \qquad{} +\varepsilon \big[2 { \alpha_3} q_{{n+1,l}}-( q_{{n,l+1}}+2 \alpha_3) q_{{n,l}} \big] q_{{n-1,l+1}} \nonumber\\
 \qquad{} + \varepsilon\big[ ( q_{{n,l+1}}+2\alpha_3 ) q_{{n,l}}- ( 2 { \alpha_3}+q_{{n+1,l}} ) q_{{n-1,l}}
 -( q_{{n,l+1}}-q_{{n,l}}) q_{{n+1,l}} \big] q_{{n+1,l+1}} \nonumber\\
 \qquad{} -\varepsilon ( 2 \alpha_3+ q_{{n,l}}) ( q_{{n+1,l}}-q_{{n-1,l}} ) q_{{n,l+1}}=0. \label{eq:tH2qnleq}
 \end{gather}
 This partial dif\/ference equation for $q_{n,l}$ is not def\/ined on the square quad graph of Fig.~\ref{fig:geomquad}, but it is def\/ined on the six-point lattice shown in Fig.~\ref{fig:6pointslattice}.

 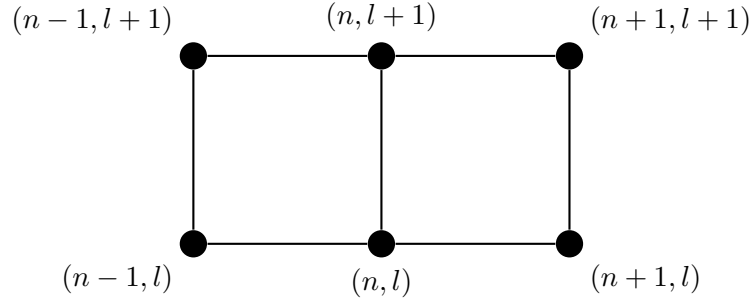
\begin{figure}[hbt]
 \centering
 \begin{tikzpicture}
 \node (nml) at (0,0) [circle,fill,label=-135:{$(n-1,l)$}] {};
 \node (nmlp) at (0,2.5) [circle,fill,label=135:{$(n-1,l+1)$}] {};
 \node (nl) at (2.5,0) [circle,fill,label=-90:{$(n,l)$}] {};
 \node (nlp) at (2.5,2.5) [circle,fill,label=90:{$(n,l+1)$}] {};
 \node (nplp) at (5,2.5) [circle,fill,label=45:{$(n+1,l+1)$}] {};
 \node (npl) at (5,0) [circle,fill,label=-45:{$(n+1,l)$}] {};
 \draw[thick] (nlp)--(nmlp)--(nml)--(nl)--(npl)--(nplp)--(nlp)--(nl);
 \end{tikzpicture}
 \caption{The six-point lattice.} \label{fig:6pointslattice}
 \end{figure}

In the general case we have proved that the \tHeq{2} equation \eqref{eq:tH2e} is equivalent to the system~\eqref{eq:tH2sys} which in turn is equivalent to the solution of equations \eqref{eq:tH2pnlp} and~\eqref{eq:tH2qnleq}. However~\eqref{eq:tH2pnlp} merely def\/ines $p_{n,l+1}$ in terms of $q_{n,l}$ and its shifts. Therefore if we f\/ind the general solution of equation \eqref{eq:tH2qnleq} the value of $p_{n,l}$ will follow. Applying $T_{l}^{-1}$ to equation \eqref{eq:tH2l2} we obtain then the value of $p_{n,l}$ as displayed in \eqref{eq:tH2pnlsoldefgen}. To f\/ind the solution for $q_{n,l}$ solution we turn to the f\/irst integrals. Like in the case of the $H^{6}$ equations \eqref{eq:h6} we will f\/ind an expression for $q_{n,l}$ using the f\/irst integrals, and then we will insert it into \eqref{eq:tH2qnleq} to reduce the number of arbitrary functions to the right one. From \cite{GSY_DarbouxI} we know that \tHeq{2} equation \eqref{eq:tH2e} possesses the following four-point, third-order integral in the $n$ direction
 \begin{gather}
 W_{1}=\Fp{m}{\frac {( u_{{n+1,m}}-u_{{n-1,m}})( u_{{n+2,m}}-u_{{n,m}}) }{ {\epsilon}^{2}\alpha_{2}^{4}+4\epsilon\alpha_{2}^{3}+
 \big[ (8{ \alpha_{3}} -2u_{{n,m}}-2u_{{n+1,m}} ) \epsilon-1 \big]
 \alpha_{2}^{2}+( u_{{n,m}}-u_{{n+1,m}} ) ^{2}}}\nonumber \\
\hphantom{W_{1}=}{} -\Fm{m}{\frac {( -u_{{n+1,m}}+u_{{n-1,m}} )
( u_{{n,m}}-u_{{n+2,m}} ) }{ ( -u_{{n-1,m}}+u_{{n,m}}+{ \alpha_{2}}) ( u_{{n+1,m}}+{ \alpha_{2}}-u_{{n+2,m}} )}}.
 \label{eq:W1tH2e}
 \end{gather}
 We consider the equation $W_{1}=\xi_{n}$, where $W_{1}$ is given by \eqref{eq:W1tH2e}, with $m=2l+1$
 \begin{gather*}
 {\frac { ( u_{{n-1,2l+1}}-u_{{n+1,2l+1}} )
( u_{{n+2,2l+1}}-u_{{n,2l+1}} ) }{ (u_{{n,2l+1}} -u_{{n-1,2l+1}}+{ \alpha_{2}} )
 ( u_{{n+1,2l+1}}-u_{{n+2,2l+1}}+\alpha_{2} )}}=\xi_{n}.% \label{eq:tH2kp}
 \end{gather*}
 Using the substitutions \eqref{eq:h4aut} we have
 \begin{gather}
 {\frac { \left( q_{{n-1,l}}-q_{{n+1,l}} \right)
 \left( q_{{n+2,l}}-q_{{n,l}} \right) }{ \left( q_{{n,l}}-q_{{n-1,l}}+{ \alpha_{2}} \right)
 \left( q_{{n+1,l}}-q_{{n+2,l}}+\alpha_{2} \right)}}=\xi_{n}.
 \label{eq:tH2kp2}
 \end{gather}
 This equation contains only $q_{n,l}$ and its shifts. From equation \eqref{eq:tH2kp2} it is very simple to obtain a discrete Riccati equation. Indeed the transformation
 \begin{gather}
 Q_{n,l} = \frac{q_{n, l}-q_{n-1, l}+\alpha_{2}}{ q_{n+1, l}-q_{n-1, l}} \label{eq:tH2Qdef}
 \end{gather}
 brings \eqref{eq:tH2kp2} into:
 \begin{gather}
 Q_{n+1, l} +\frac{1}{\xi_{n}Q_{n, l}}=1
 \label{eq:tH2kp3}
 \end{gather}
which is a discrete Riccati equation. Let us assume $a_{n}$ to be a particular solution of \eqref{eq:tH2kp3}, then we express $\xi_{n}$ as
 \begin{gather}
 \xi_{n} = \frac{1}{a_{n}( 1-a_{n+1})}. \label{eq:tH2lam}
 \end{gather}
 Using the standard linearization of the discrete Riccati equation
 \begin{gather}
 Q_{n,l} = a_{n} + \frac{1}{Z_{n,l}} \label{eq:tH2Zdef}
 \end{gather}
 from \eqref{eq:tH2lam} we obtain the following equation for $Z_{n,l}$
 \begin{gather*}
 Z_{n+1,l} = \frac{a_{n}Z_{n, l}+1}{1-a_{n+1}}. %\label{eq:tH2kp4}
 \end{gather*}
 Introducing
 \begin{gather}
 a_{n} = \frac{b_{n-1}}{b_{n}+b_{n-1}} \label{eq:tH2adef}
 \end{gather}
 we obtain
 \begin{gather}
 Z_{n+1,l}-\frac{b_{n-1} b_{n+1}+b_{n-1} b_{n}}{b_{n} b_{n+1}+b_{n-1} b_{n+1}} Z_{n,l}
 =\frac{b_{n-1} b_{n+1}+b_{n-1} b_{n}+b_{n} b_{n+1}+b_{n}^2}{b_{n} b_{n+1}+b_{n-1} b_{n+1}}. \label{eq:tH2kp5}
 \end{gather}
 If we assume that \eqref{eq:tH2kp5} can be written as a total dif\/ference, i.e.,
 \begin{gather*}
( T_{n}-\Id ) ( d_{n}Z_{n,l} - c_{n} ) =0, %\label{eq:tH2kp6}
 \end{gather*}
 we obtain
 \begin{gather}
 b_{n}=c_{n+1}-c_{n}, \qquad d_{n} = \frac{(c_{n+1}-c_{n})(c_{n}-c_{n-1})}{c_{n+1}-c_{n-1}}. \label{eq:tH2bddef}
 \end{gather}
 So $b_{n}$ must be a total dif\/ference and therefore we can represent $Z_{n,l}$ as
 \begin{gather*}
 Z_{n,l} = \frac{(c_{n+1}-c_{n-1})(c_{n}+\zeta_{l})}{ (c_{n+1}-c_{n})(c_{n}-c_{n-1})}. %\label{eq:tH2Zsol}
 \end{gather*}
 From \eqref{eq:tH2Zdef} and \eqref{eq:tH2adef} we obtain the form of $Q_{n,l}$
 \begin{gather}
 Q_{n,l}=\frac{ (c_{n}-c_{n-1}) (c_{n+1}+\zeta_{l})}{ (c_{n}+\zeta_{l}) (c_{n+1}-c_{n-1})}. \label{eq:tH2Qsol}
 \end{gather}
 Introducing the value of $Q_{n,l}$ from \eqref{eq:tH2Qsol} into \eqref{eq:tH2Qdef} we obtain the following equation for $q_{n,l}$
 \begin{gather*}
 \frac{q_{n+1, l}-q_{n-1, l}}{q_{n, l}-q_{n-1, l}+\alpha_{2}} =
 \frac{(c_{n}+\zeta_{l}) (c_{n+1}-c_{n-1})}{(c_{n}-c_{n-1}) (c_{n+1}+\zeta_{l})}. %\label{eq:tH2kp7}
 \end{gather*}
 Performing the transformation
 \begin{gather}
 R_{n,l} = ( c_{n}+\zeta_{l} )q_{n,l} \label{eq:tH2rdef}
 \end{gather}
 we obtain the following second-order ordinary dif\/ference equation for the f\/ield $R_{n,l}$
 \begin{gather}
 R_{n+1,l} -\frac{c_{n+1}-c_{n-1}}{c_{n}-c_{n-1}} R_{n,l}+\frac{c_{n+1}-c_{n}}{c_{n}-c_{n-1}}R_{n-1,l} -
 \alpha_2 (c_{n}+\zeta_{l}) \frac{c_{n+1}-c_{n-1}}{c_{n}-c_{n-1}} =0. \label{eq:tH2req}
 \end{gather}
 Then we can represent the solutions of the equation \eqref{eq:tH2req} as
 \begin{gather}
 R_{n,l} = P_{n,l} + \zeta_{l}e_{n} + f_{n}, \label{eq:tH2Pdef}
 \end{gather}
 where $e_{n}$ and $f_{n}$ are particular solutions of
 \begin{subequations} \label{eq:tH2deeq}
 \begin{gather}
 e_{n+1} -\frac{c_{n+1}-c_{n-1}}{c_{n}-c_{n-1}} e_{n} +\frac{c_{n+1}-c_{n}}{c_{n}-c_{n-1}}e_{n-1} - \alpha_2 \frac{c_{n+1}-c_{n-1}}{c_{n}-c_{n-1}} =0, \label{eq:tH2deq} \\
 f_{n+1} -\frac{c_{n+1}-c_{n-1}}{c_{n}-c_{n-1}} f_{n} +\frac{c_{n+1}-c_{n}}{c_{n}-c_{n-1}}f_{n-1} - \alpha_2 c_{n}\frac{c_{n+1}-c_{n-1}}{c_{n}-c_{n-1}}=0.
 \label{eq:tH2eeq}
 \end{gather}
 \end{subequations}
 $P_{n,l}$ will be then solve the following equation
 \begin{gather}
 P_{n+1,l} -\frac{c_{n+1}-c_{n-1}}{c_{n}-c_{n-1}} P_{n,l} +\frac{c_{n+1}-c_{n}}{c_{n}-c_{n-1}}P_{n-1,l} =0. \label{eq:tH2Peq}
 \end{gather}

 The equations \eqref{eq:tH2deq} and \eqref{eq:tH2eeq} are not independent. Indeed def\/ining
 \begin{gather}
 A_{n} = \frac{e_{n}c_{n-1}-e_{n-1}c_{n}-f_{n}+f_{n-1}}{c_{n}-c_{n-1}}, \label{eq:tH2Adef}
 \end{gather}
 and using \eqref{eq:tH2deeq} it is possible to show that the function $A_{n}$ lies in the kernel of the operator $T_{n}-\Id$. This implies that $A_{n}=A_{0}=\text{const}$. We can without loss of generality assume the constant $A_{0}$ to be zero, since if we perform the transformation
 \begin{gather}
 e_{n} = \tilde{e}_{n}-A_{0}, \label{eq:tH2dtdef}
 \end{gather}
 the equation \eqref{eq:tH2Adef} is mapped into
 \begin{gather}
 \frac{\tilde{e}_{n}c_{n-1}-\tilde{e}_{n-1}c_{n}-f_{n}+f_{n-1}}{c_{n}-c_{n-1}}=0. \label{eq:tH2Adef2}
 \end{gather}
 Furthermore since \eqref{eq:tH2deq} is invariant under the transformation \eqref{eq:tH2dtdef} we can safely drop the tilde in \eqref{eq:tH2Adef2} and assume that the functions $e_{n}$ and $f_{n}$ are solutions of the equations{\samepage
 \begin{subequations} \label{eq:tH2deeq2}
 \begin{gather}
 e_{n+1} -\frac{c_{n+1}-c_{n-1}}{c_{n}-c_{n-1}} e_{n} +\frac{c_{n+1}-c_{n}}{c_{n}-c_{n-1}}e_{n-1} - \alpha_2 \frac{c_{n+1}-c_{n-1}}{c_{n}-c_{n-1}} =0, \label{eq:tH2deq2} \\
 f_{n}-f_{n-1}=e_{n}c_{n-1}-c_{n} e_{n-1}. \label{eq:tH2eeq2}
 \end{gather}
 \end{subequations}
 The system \eqref{eq:tH2deeq2} is just gives the constraints expressed in formulas \eqref{eq:tH2eeq0} and \eqref{eq:tH2feq0}.}

 Now we turn to the solution of the homogeneous equation \eqref{eq:tH2Peq}. We can reduce \eqref{eq:tH2Peq} to a total dif\/ference using the potential substitution $T_{n,l}=P_{n,l}-P_{n-1,l}$
 \begin{gather*}
 \frac{T_{n+1,l}}{c_{n+1}-c_{n}} -\frac{T_{n,l}}{c_{n}-c_{n-1}} =0. %\label{eq:tH2Teq}
 \end{gather*}
 This clearly implies
 \begin{gather*}
 \frac{P_{n,l}-P_{n-1,l}}{c_{n}-c_{n-1}} = \beta_{l},% \label{eq:tH2Peq2}
 \end{gather*}
 where $\beta_{l}$ is an arbitrary function. The solution to this equation is given by\footnote{The arbitrary functions are taken in a convenient way.}
 \begin{gather}
 P_{n,l} = ( c_{n}+\zeta_{l} )\beta_{l}+\gamma_{l}, \label{eq:tH2Psol}
 \end{gather}
 where $\gamma_{l}$ is an arbitrary function. Using (\ref{eq:tH2rdef}), (\ref{eq:tH2Pdef}), (\ref{eq:tH2Psol}) we obtain then the following expression for $q_{n,l}$
 \begin{gather}
 q_{n,l} = \beta_{l} + \frac{\gamma_{l}+\zeta_{l}e_{n} + f_{n}}{c_{n}+\zeta_{l}}, \label{eq:tH2qsol}
 \end{gather}
 where $e_{n}$ and $f_{n}$ are solutions of \eqref{eq:tH2deeq2}.

 Equation \eqref{eq:tH2qsol} is formally the solution presented in \eqref{eq:tH2qnlsoldefgen}, but depends on three arbitrary functions in the $l$ direction, namely $\zeta_{l}$, $\beta_{l}$ and $\gamma_{l}$. This means that there is a constraint between these functions, which can be recovered by plugging \eqref{eq:tH2qsol} into \eqref{eq:tH2qnleq}. At this point we have a second bifurcation, depending on the value of $\varepsilon$.

{\bf Case $\boldsymbol{\varepsilon\neq0}$.} If $\varepsilon\neq0$ inserting \eqref{eq:tH2qsol} into \eqref{eq:tH2qnleq} and factorizing out the $n$ dependent part away we are left with
 \begin{gather*}
 \gamma_{l+1}-\gamma_{l} = - (\zeta_{l+1}-\zeta_{l}) \left(\alpha_2+\beta_{l+1}+\beta_{l}+2\alpha_3-\frac{1}{\epsilon}\right).% \label{eq:tH2gameq}
 \end{gather*}
 This equation tells us that the function $\gamma_{l}$ can be expressed after a discrete integration in terms of the two arbitrary functions $\zeta_{l}$ and $\beta_{l}$. This is just the f\/inal constraint expressed in formula~\eqref{eq:tH2gameq0}. This yields the general solution of the \tHeq{2}\ equation~\eqref{eq:tH2e}.

{\bf Case $\boldsymbol{\varepsilon=0}$.} If $\varepsilon=0$ inserting \eqref{eq:tH2qsol} into \eqref{eq:tH2qnleq} we obtain the compatibility condition
$\zeta_{l+1}-\zeta_{l}=0$, i.e., $\zeta_{l}=\zeta_{0}=\text{const}$. It is easy to check that the obtained value of $q_{n,l}$ through formula \eqref{eq:tH2qsol} is consistent with the substitution of $\varepsilon=0$ in~\eqref{eq:tH2sys}. This means that in the case $\varepsilon=0$ the value of $q_{n,l}$ is given by
 \begin{gather}
 q_{n,l} = \beta_{l} + \frac{\gamma_{l}+\zeta_{0}e_{n} + f_{n}}{c_{n}+\zeta_{0}}, \label{eq:tH2qsoleps0}
 \end{gather}
 where the functions $e_{n}$ and $f_{n}$ are def\/ined implicitly and can be found by discrete integration from \eqref{eq:tH2deeq2}, i.e., from formula~\eqref{eq:tH2qnlsoldef0}. Since formula~\eqref{eq:tH2pnlsoldefgen} is not singular with respect to $\varepsilon$ the value of $p_{n,l}$ can be recovered just by substituting $\varepsilon=0$ and the form of $q_{n,l}$ found in~\eqref{eq:tH2qsoleps0}. This yields equation~\eqref{eq:tH2pnlsoldef0}.
 This concludes the procedure of solution of the \tHeq{2} equation \eqref{eq:tH2e} in the case when $\varepsilon=0$.
\end{proof}

\begin{Proposition} \label{prop:tH3e} The \tHeq{3}\ equation \eqref{eq:tH3e} is exactly solvable.
 If $\delta\neq0$ and the field $q_{n,l}$ do not satisfy
 the equation
 \begin{gather}
 q_{{n+1,l+1}}q_{{n,l}}=q_{{n+1,l}}q_{{n,l+1}} \label{eq:tH3qsing}
 \end{gather}
 then the solution of the \tHeq{3}\ equation \eqref{eq:tH3e}
 is given by
 \begin{subequations} \label{eq:tH3soldefgen}
 \begin{gather}
 q_{n,l} = \gamma_{l}e_{n}\frac{f_{n}+\beta_{l}}{c_{n}+\zeta_{l}}, \label{eq:tH3qnlsoldefgen} \\
 p_{{n,l}} ={\frac {
 \left[ \begin{gathered}
 { \alpha_2} ( q_{{n+1,l}}-q_{{n+1,l-1}})
 \big( {\epsilon}^{2}q_{{n,l-1}}q_{{n,l}}+{\delta}^{2}{{\alpha_3}}^{2} \big) \\
 {}+{\delta}^{2} \alpha_2^{2}\alpha_3^{2} ( q_{{n,l-1}}-q_{{n,l}} ) +{\epsilon}^{2}q_{{n+1,l}}q_{{n+1,l-1}} ( q_{{n,l-1}}-q_{{n,l}} )
 \end{gathered}\right]}{ ( q_{{n+1,l}}q_{{n,l-1}}-q_{{n+1,l-1}}q_{{n,l}} ) {\alpha_3} { \alpha_2}}}, \label{eq:tH3pnlsoldefgen}
 \end{gather}
 \end{subequations}
 where $e_{n}$, $\beta_{l}$ and $\gamma_{l}$ are arbitrary functions of their arguments and~$c_{n}$ is a solution of the equation
 \begin{gather}
 \frac {c_{{n+1}}-c_{{n}}}{c_{{n}}-c_{{n-1}}} = {\frac {e_{{n+1}}-{\alpha_2}e_{{n}}}{{\alpha_2}e_{{n}}-e_{{n-1}}}}, \label{eq:tH3eeq0}
 \end{gather}
 while $f_{n}$ and $\zeta_{l}$ are given by the discrete integrations
 \begin{subequations} \label{eq:tH3discrint0}
 \begin{gather}
 f_{n}-f_{n-1}=\frac{c_{n-1}-c_{n}}{e_{n}e_{n-1}}, \label{eq:tH3feq0} \\
 \zeta_{l+1} - \zeta_{l}= \frac{\epsilon^2}{\alpha_2 \delta^2 \alpha_3^2} \gamma_{l+1}\gamma_{l} (\beta_{l+1}-\beta_{l}). \label{eq:tH3gameq0}
 \end{gather}
 \end{subequations}
 If $\varepsilon=0$, but the field $q_{n,l}$ do not satisfy the equation \eqref{eq:tH3singcc} then the solution of the \tHeq{3}\ equation~\eqref{eq:tH3e} is given by
 \begin{subequations} \label{eq:tH3soldef0}
 \begin{gather}
 q_{n,l} = \gamma_{l}e_{n}\frac{f_{n}+\beta_{0}}{c_{n}+\zeta_{l}}, \label{eq:tH3qnlsoldef0} \\
 p_{{n,l}} =\frac{\varepsilon^{2}}{\alpha_{2}\alpha_{3}}\frac { { \alpha_2} ( q_{{n+1,l}}-q_{{n+1,l-1}})q_{{n,l-1}}q_{{n,l}}
 +q_{{n+1,l}}q_{{n+1,l-1}} ( q_{{n,l-1}}-q_{{n,l}} )}{ q_{{n+1,l}}q_{{n,l-1}}-q_{{n+1,l-1}}q_{{n,l}}}, \label{eq:tH3pnlsoldef0}
 \end{gather}
 \end{subequations}
 where $e_{n}$, $\zeta_{l}$ and $\gamma_{l}$ are arbitrary functions of their arguments $\beta_{0}$ is a constant and $c_{n}$ is a~solution of~\eqref{eq:tH3eeq0} and $f_{n}$ is a solution of~\eqref{eq:tH3feq0}. If the field $q_{n,l}$ satisfies the equation \eqref{eq:tH3qsing} regardless of the value of the parameter $\varepsilon$ the solution of the \tHeq{3}\ equation \eqref{eq:tH3e} is given by
 \begin{subequations} \label{eq:tH3singdef}
 \begin{gather}
 q_{n,l} = \zeta_{0} a_{n}, \label{eq:tH3qnlsingdef} \\
 p_{n,l} = b_{n}( \beta_{l}+c_{n} ), \label{eq:tH3pnlsingdef}
 \end{gather}
 \end{subequations}
 where $b_{n}$ and $\beta_{l}$ are arbitrary functions of their arguments, $\zeta_{0}$ is a constant and $a_{n}$ and $c_{n}$ are given by the discrete integration
 \begin{subequations}\label{eq:tH3discrintsing}
 \begin{gather}
 \frac{\alpha_2 a_{n+1}-a_{n}}{a_{n+1}-\alpha_2 a_{n}} = \frac{b_{n+1}}{b_{n}}, \label{eq:tH3aeq0} \\
 c_{n+1} - c_{n} = {\frac {{\delta}^{2}\alpha_3^{2}\alpha_2^{2}b_{{n}}-b_{{n+1}}
 \big( {\delta}^{2}\alpha_3^{2}+{\epsilon}^{2}a_{{n}}^{2}\zeta_{0}^{2} \big) { \alpha_2}+{\epsilon}^{2}a_{{n}}^{2}\zeta_{0}^{2}b_{{n}}}{%
 b_{{n}}a_{{n}}\zeta_{0}{ \alpha_2} { \alpha_3} b_{{n+1}}}}. \label{eq:tH3ceq0}
 \end{gather}
 \end{subequations}
\end{Proposition}

\begin{Remark} We remark that the function $c_{n}$ can be obtained from \eqref{eq:tH3eeq0} as the result of \emph{two} discrete integrations. Indeed def\/ining
 \begin{gather}
 z_{n} = c_{n+1}-c_{n}, \label{eq:tH3Dndef}
 \end{gather}
 and substituting in \eqref{eq:tH3eeq0} we obtain that $z_{n}$ must solve the equation
 \begin{gather}
 \frac{z_{n}}{z_{n-1}}= \frac {e_{{n+1}}-{\alpha_2}e_{{n}}}{{\alpha_2}e_{{n}}-e_{{n-1}}}. \label{eq:tH3Dneq}
 \end{gather}
 Note that the right-hand side of \eqref{eq:tH3Dneq} is not a total dif\/ference. So the function $c_{n}$ can be obtained by integrating \eqref{eq:tH3Dneq} and subsequently integrating~\eqref{eq:tH3Dndef}. This provides the value of~$c_{n}$. The obtained value can be plugged in~\eqref{eq:tH3feq0} to give~$f_{n}$ after discrete integration. This reasoning shows that we can obtain the non-arbitrary functions~$c_{n}$ and~$f_{n}$ as result of a~f\/inite number of
 discrete integrations. Therefore we can conclude that the solution of the \tHeq{3}\ equation~\eqref{eq:tH3e} in the general case is given in terms of four discrete integrations. If $\varepsilon=0$ the general solution is given in terms of three discrete integration and f\/inally in the singular case, when $q_{n,l}$ solves the equation \eqref{eq:tH3qsing}, we need only two discrete integrations.
\end{Remark}

\begin{proof} The proof of the solution of the \tHeq{3}\ equation \eqref{eq:tH3e} proceeds as the one outlined in Proposition~\ref{prop:tH2e}. The interested reader can f\/ind the details in Appendix~\ref{app:remsols}.
\end{proof}

\section{Conclusions}\label{sec:concl}

In this paper we presented a detailed procedure to construct the general solutions for all the~\Hvier\ and \Hsechs~equations. As stated in the introduction, these general solutions were obtained in three dif\/ferent ways, but the common feature is that they can be found through some \emph{linear or linearizable
$($discrete Riccati$)$ equations}. This is the great advantage of the f\/irst integral approach over the direct one which was pursued in~\cite{GSL_general}. The Darboux integrability therefore yields extra information that it is useful to get the f\/inal result, i.e., the general solutions. Moreover linearization arises very naturally from f\/irst integrals also in the most complicated cases, whereas in the direct approach can be quite tricky, see, e.g., the examples in~\cite{GSL_general}. The linearization of the f\/irst integrals is another proof of the deep linear nature of the \Hvier\ and~\Hsechs~equations. This result is even stronger than Darboux integrability alone, since \emph{a~priori} the f\/irst integrals do not need to def\/ine linearizable equations.

We also note that our procedure of construction of the general solution, based on the ideas from \cite{GarifullinYamilov2015}, is likely to be the discrete version of the procedure of linearization and solutions for continuous Darboux integrable equations presented in \cite{ZhiberSokolov2011}. The preeminent r\^ole of the discrete Riccati equation in the solutions is reminiscent of the importance of the usual Riccati equation in the continuous case. Recall, e.g., that the f\/irst integrals of the Liouville equation~\cite{Liouville1853}
\begin{gather*}
 u_{xt} = e^{u}, \label{eq:liouville}
\end{gather*}
which is the most famous Darboux integrable system, are Riccati equations. Many other examples of solutions presented in \cite{ZhiberSokolov2011} use the reduction of higher-order dif\/ferential equations to Riccati-like equation in order to obtain the solution, as we have done in the discrete case.

In this paper we constructed the general solutions of the trapezoidal \Hvier~and of the \Hsechs~equations. Therefore we possess an almost complete theory about these equations ranging from the geometrical background to their analytic properties. For a discussion of the open problems in this f\/ield we refer to our previous paper \cite{GSY_DarbouxI}.

\appendix

\section{Procedure to f\/ind the general solution in the remaining cases}\label{app:remsols}

\subsection[$_{2}D_{2}$ equation (\ref{eq:2D2})]{$\boldsymbol{{}_{2}D_{2}}$ equation (\ref{eq:2D2})}

From \cite{GSY_DarbouxI} we know that the f\/irst integrals of the $_{2}D_{2}$ equation \eqref{eq:2D2} can be dif\/ferent depending on the value of the parameter~$\delta_{1}$. For this reason we treat separately the various cases.

{\bf Case $\boldsymbol{\delta_{1}\neq0}$.} If $\delta_{1}\neq0$ the $W_{1}$ f\/irst integrals of the $_{2}D_{2}$ equation \eqref{eq:2D2} is given by~\cite{GSY_DarbouxI}
 \begin{gather*}
 W_{1} ={ \Fppp}{\alpha} {\frac {{ \delta_{2}}+u_{{n+1,m}}}{{ \delta_{2}}+u_{{n-1,m}}}}
 +{ \Fpmm}{\alpha} {\frac { [ 1- ( 1+{ \delta_{2}} ) {
 \delta_{1}} ] u_{{n,m}}+u_{{n+1,m}}}{ [ 1 - ( 1+{ \delta_{2}} ) { \delta_{1}}] u_{{n,m}}+u_{{n-1,m}}}} \\
\hphantom{W_{1} =}{} +{ \Fmpm}{\beta} {\frac { (u_{{n+1,m}} -u_{{n-1,m}} ) ( u_{{n,m}}+{ \delta_{2}}) }{1+ ( u_{{n,m}}-1 ) { \delta_{1}}}}
-{ \Fmmp}{\beta} ( u_{{n+1,m}} -u_{{n-1,m}}).% \label{eq:W12D2}
 \end{gather*}
As stated in the introduction, from the relation $W_{1}=\xi_{n}$ this f\/irst integral def\/ines a three-point, second-order ordinary dif\/ference equation in the $n$ direction which depends parametrically on~$m$. From this parametric dependence we f\/ind two dif\/ferent three-point non-autonomous ordinary dif\/ference equations corresponding to $m$ even and $m$ odd. We treat them separately.

{\bf Case $\boldsymbol{m=2l}$.} If $m=2l$ we have the following non-autonomous nonlinear ordinary dif\/ference equation
 \begin{gather}
{ \Fp{n}}{\alpha}{\frac {{ \delta_{2}}+u_{{n+1,2l}}}{{ \delta_{2}}+u_{{n-1,2l}}}}-{ \Fm{n}}{\beta}{\frac {( u_{{n-1,2l}}-u_{{n+1,2l}} )
( u_{{n,2l}}+{ \delta_{2}} ) }{1+ ( u_{{n,2l}}-1 ) { \delta_{1}}}} =\xi_{n}. \label{eq:2D2eQm}
 \end{gather}
 Without loss of generality we set $\alpha=1$ and $\beta=\delta_{1}$. Then making the transformation
 \begin{gather}
 u_{n,2l} = U_{n,2l} - \delta_{2} \label{eq:2D2eQmtransf}
 \end{gather}
 and putting
 \begin{gather}
 \delta=\frac{1-\delta_{1}-\delta_{1}\delta_{2}}{\delta_{1}} \label{eq:deltaconstver2}
 \end{gather}
 equation \eqref{eq:2D2eQm} is mapped to
 \begin{gather}
 \Fp{n}\frac {U_{n+1,2l}}{U_{n-1,2l}} -\Fm{n}\frac { ( U_{{n-1,2l}}-U_{{n+1,2l}} ) U_{{n,2l}}}{ U_{n,2l}+\delta}=\xi_{n}. \label{eq:2D2eQm2}
 \end{gather}
From the def\/inition \eqref{eq:genautsubl} applied to $U_{n,2l}$ instead of $u_{n,2l}$\footnote{We will denote the corresponding f\/ields with capital letters.} we can separate again the even and the odd part in~\eqref{eq:2D2eQm2}. We obtain the following system of two coupled \emph{first-order} ordinary dif\/ference equations
 \begin{subequations} \label{eq:2D2eQm2sys}
 \begin{gather}
 W_{k,l} - \xi_{2k} W_{k-1,l} =0, \label{eq:2D2eQm2k2} \\
 V_{k+1,l} - V_{k,l} = \xi_{2k+1} \left( 1+\frac{\delta}{W_{k,l}} \right). \label{eq:2D2eQm2kp2}
 \end{gather}
 \end{subequations}
 Putting $\xi_{2k}=a_{k}/a_{k-1}$ the solution to \eqref{eq:2D2eQm2k2} is given by
 \begin{gather}
 W_{k,l} = a_{k}\alpha_{l}. \label{eq:2D2Wklsol}
 \end{gather}
 Inserting the value of $W_{k,l}$ from \eqref{eq:2D2Wklsol} into \eqref{eq:2D2eQm2kp2} we obtain
 \begin{gather}
 V_{k+1,l} - V_{k,l} = \xi_{2k+1} \left( 1+\frac{\delta}{a_{k}\alpha_{l}} \right). \label{eq:2D2eQm2kp3}
 \end{gather}
 If we def\/ine
 \begin{gather}
 \xi_{2k+1}=b_{k+1}-b_{k}, \qquad a_{k} = \frac{b_{k+1}-b_{k}}{c_{k+1}-c_{k}}, \label{eq:2D2akdef}
 \end{gather}
 then
 \eqref{eq:2D2eQm2kp3} becomes a total dif\/ference. So we obtain the following solutions for the~$W_{k,l}$ and the~$V_{k,l}$ f\/ields
 %\begin{subequations}\label{eq:2D2WVklsol2}
 \begin{gather*}
 W_{k,l} = \alpha_{l}\frac{b_{k+1}-b_{k}}{c_{k+1}-c_{k}}, \qquad %\label{eq:2D2Wklsol2} \\
 V_{k,l} = b_{k}+\beta_{l} + \delta\frac{c_{k}}{\alpha_{l}}. %\label{eq:2D2Vklsol2}
 \end{gather*}
 %\end{subequations}
 Inverting the transformation \eqref{eq:2D2eQmtransf} we obtain for the f\/ields $w_{k,l}$ and $v_{k,l}$
 \begin{subequations} \label{eq:2D2wvklsol2}
 \begin{gather}
 w_{k,l} = \alpha_{l}\frac{b_{k+1}-b_{k}}{c_{k+1}-c_{k}} + \frac{1}{\delta_{1}}-1-\delta, \label{eq:2D2wklsol2} \\
 v_{k,l}= b_{k}+\beta_{l} + \delta\frac{c_{k}}{\alpha_{l}} +\frac{1}{\delta_{1}}-1-\delta. \label{eq:2D2vklsol2}
 \end{gather}
 \end{subequations}

{\bf Case $\boldsymbol{m=2l+1}$.} If $m=2l+1$ we have the following non-autonomous ordinary dif\/ference equation
 \begin{gather*}
 \Fp{n}\frac { \delta\delta_1 u_{{n,2l+1}}+u_{{n+1,2l+1}}}{\delta \delta_1 u_{{n,2l+1}}+u_{{n-1,2l+1}}}
 +\Fm{n}\delta_1 ( u_{{n-1,2l+1}}-u_{{n+1,2l+1}} ) =\xi_{n}, %\label{eq:2D2ePm}
 \end{gather*}
 where we already substituted $\delta$ as def\/ined in \eqref{eq:deltaconstver2}. Using the standard transformation \eqref{eq:genautsublp} to get rid of the two-periodic factors we obtain
 \begin{subequations} \label{eq:2D2ePmsys}
 \begin{gather}
 \frac { \delta\delta_1 y_{k,l}+z_{k,l}}{ \delta \delta_1 y_{k,l}+z_{k-1,l}} =\xi_{2k},\label{eq:2D2ePmk} \\
 \delta_{1} ( y_{k,l}-y_{k+1,l} ) = \xi_{2k+1}. \label{eq:2D2ePmkp}
 \end{gather}
 \end{subequations}
 Both equations in \eqref{eq:2D2ePmsys} are linear in $z_{k,l}$, $y_{k,l}$ and their shifts. As $\xi_{2k+1}=b_{k+1}-b_{k}$ we have that the solution of equation \eqref{eq:2D2ePmkp} is given by
 \begin{gather}
 y_{k,l} = \gamma_{l} - \frac{b_{k}}{\delta_{1}}. \label{eq:2D2yklsol}
 \end{gather}
 As $\xi_{2k}=a_{k}/a_{k-1}$ and $y_{k,l}$ is given by \eqref{eq:2D2yklsol} we obtain
 \begin{gather*}
 \frac{z_{k, l}}{a_{k}}-\frac{z_{k-1, l}}{a_{k-1}}= \left( \frac{1}{a_{k}}-\frac{1}{a_{k-1}} \right) (\delta b_{k}- \delta\delta_1\gamma_{l} ).% \label{eq:2D2ePmkp2}
 \end{gather*}
 Recalling the def\/inition of $a_{k}$ in \eqref{eq:2D2akdef} we represent $z_{k,l}$ as
 \begin{gather}
 z_{k,l} = \delta b_{k} +\frac{b_{k+1}-b_{k}}{c_{k+1}-c_{k}}(\zeta_{l} -\delta c_{k}) -\delta\delta_{1}\gamma_{l}. \label{eq:2D2zklsol}
 \end{gather}

Equations (\ref{eq:2D2wvklsol2}), (\ref{eq:2D2yklsol}), (\ref{eq:2D2zklsol}) provide the value of the four f\/ields, but we have too many arbitrary functions in the $l$ direction, namely $\alpha_{l}$, $\beta_{l}$, $\gamma_{l}$ and $\zeta_{l}$. Inserting (\ref{eq:2D2wvklsol2}), (\ref{eq:2D2yklsol}), (\ref{eq:2D2zklsol}) into \eqref{eq:2D2} and separating the terms even and odd in $n$ and $m$ we obtain two independent equations
\begin{subequations} \label{eq:2D2deltagammaeq}
 \begin{gather}
 { \delta_{1}} \zeta_{l}+\alpha_{{l}}\delta_{1}^{2}\gamma _{{l}} -{{ \delta_{1}}}^{2}\lambda \alpha_{{l}}+\beta_{{l}}\alpha_{{l}}{\delta_{1}} +\delta \alpha_{{l}}{ \delta_{1}}-\alpha_{{l}}+\alpha_{{l}}{\delta_{1}}=0, \\
 { \delta_{1}} \zeta_{l}+\alpha_{{l+1}}\delta_{1}^{2}\gamma_{{l}} -\delta_{1}^{2}\lambda \alpha_{{l+1}}+\beta_{{l+1}}\alpha_{{l+1}}\delta_{1} +\delta \alpha_{{l+1}}{ \delta_{1}}-\alpha_{{l+1}}+\alpha_{{l+1}}{ \delta_{1}}=0,
 \end{gather}
\end{subequations}
which allow us to reduce by two the number of independent functions in the $l$ direction. Sol\-ving~\eqref{eq:2D2deltagammaeq} with respect to $\gamma_{l}$ and $\zeta_{l}$ we f\/ind
\begin{subequations} \label{eq:2D2deltagammadef}
 \begin{gather}
 \gamma_{l} =-{\frac {\beta_{{l+1}}{ \delta_{1}}-1-\delta_{1}^{2} \lambda+\delta { \delta_{1}}+{ \delta_{1}}}{\delta_{1}^{2}}}
 -{\frac {\alpha_{{l}} ( \beta_{{l+1}}-\beta_{{l}} ) }{ ( \alpha_{{l+1}} -\alpha_{{l}}) { \delta_{1}}}}, \label{eq:2D2gammadef} \\
 \zeta_{l} =\alpha_{{l}} ( \beta_{{l+1}}-\beta_{{l}} ) + {\frac {\alpha_{l}^{2} ( \beta_{{l+1}}-\beta_{{l}} )}{ \alpha_{{l+1}}-\alpha_{{l}} }}.\label{eq:2D2deltadef}
 \end{gather}
\end{subequations}

Inserting \eqref{eq:2D2deltagammadef} into equations (\ref{eq:2D2wvklsol2}), (\ref{eq:2D2yklsol}), (\ref{eq:2D2zklsol}) we have the general solution~\eqref{eq:2D2solfin} of the~$_{2}D_{2}$ equation~\eqref{eq:2D2} provided that $\delta_{1}\neq0$. Indeed the solution of the $_{2}D_{2}$ equation~\eqref{eq:2D2} given by~\eqref{eq:2D2solfin} is ill-def\/ined if $\delta_{1}=0$. Therefore we now discuss this case separately.

{\bf Case $\boldsymbol{\delta_{1}=0}$.} Following~\cite{GSY_DarbouxI} we have the $_{2}D_{2}$ equation \eqref{eq:2D2} with $\delta_{1}=0$ possesses the following two-point, f\/irst-order f\/irst integral in the direction $n$
\begin{gather}
 W_{1}^{(0,\delta_{2})}= \Fppp \alpha ( { \delta_2}+u_{{n+1,m}}) u_{{n,m}} -\Fpmm \alpha ( u_{{n+1,m}}+u_{{n,m}} )\nonumber \\
\hphantom{W_{1}^{(0,\delta_{2})}=}{} + \Fmpm \beta (\delta_2+ u_{{n,m}} )u_{{n+1,m}} -\Fmmp \beta ( u_{{n+1,m}}+u_{{n,m}}). \label{eq:W12D20d2}
\end{gather}
To solve the $_{2}D_{2}$ equation \eqref{eq:2D2} with $\delta_{1}=0$ we use the f\/irst integral \eqref{eq:W12D20d2}. Again we start separating the cases $m$ even and odd in~\eqref{eq:W12D20d2}.

{\bf Case $\boldsymbol{m=2l}$.} If $m=2l$ we obtain from the f\/irst integral~\eqref{eq:W12D20d2}
 \begin{gather}
 \Fp{n} ( { \delta_2}+u_{{n+1,2l}} ) u_{{n,2l}} + \Fm{n} (\delta_2+ u_{{n,2l}} )u_{{n+1,2l}}=\xi_{n}, \label{eq:2D20d2eQm}
 \end{gather}
 where we have chosen without loss of generality $\alpha=\beta=1$. Applying the transformation \eqref{eq:genautsubl} equation \eqref{eq:2D20d2eQm} becomes the system
 \begin{subequations} \label{eq:2D20d2eQmsys}
 \begin{gather}
 v_{k,l} (\delta_2+w_{k,l}) = \xi_{2k}, \label{eq:2D20d2eQmk} \\
 v_{k+1,l}(\delta_2+w_{k,l}) = \xi_{2k+1}. \label{eq:2D20d2eQmkp}
 \end{gather}
 \end{subequations}
 In this case the system \eqref{eq:2D20d2eQmsys} do not consist of purely dif\/ference equations. Indeed from \eqref{eq:2D20d2eQmk} we can derive immediately the value of the f\/ield~ $w_{k,l}$
 \begin{gather}
 w_{k,l} = -\delta_{2} +\frac{\xi_{2k}}{v_{k,l}}. \label{eq:2D20d2wsol}
 \end{gather}
 Inserting \eqref{eq:2D20d2wsol} into \eqref{eq:2D20d2eQmkp} we obtain that $v_{k,l}$ solves the equation
 \begin{gather}
 v_{k+1,l} - \frac{\xi_{2k+1}}{\xi_{2k}}v_{k,l} = 0. \label{eq:2D20d2eQmkp2}
 \end{gather}
 Def\/ining
 \begin{gather}
 \xi_{2k+1} = \frac{a_{k+1}}{a_{k}}\xi_{2k}, \label{eq:2D20d2xikp}
 \end{gather}
 we have that \eqref{eq:2D20d2eQmkp2} becomes a total dif\/ference. So we have that the system \eqref{eq:2D20d2eQmsys} is solved by
 \begin{subequations} \label{eq:2D20d2wvsol}
 \begin{gather}
 v_{k,l}= a_{k}\alpha_{l}, \label{eq:2D20d2vsol} \\
 w_{k,l} = -\delta_{2} +\frac{\xi_{2k}}{a_{k}\alpha_{l}}. \label{eq:2D20d2wsol2}
 \end{gather}
 \end{subequations}

{\bf Case $\boldsymbol{m=2l+1}$.} If $m=2l+1$ we obtain from the f\/irst integral \eqref{eq:W12D20d2}
 \begin{gather}
 \Fp{n}( u_{{n+1,2l+1}}+u_{{n,2l+1}}) +\Fm{n} ( u_{{n+1,2l+1}}+u_{{n,2l+1}} )=-\xi_{n}. \label{eq:2D20d2ePm}
 \end{gather} Applying the transformation \eqref{eq:genautsublp} equation \eqref{eq:2D20d2ePm} becomes the system
 \begin{subequations} \label{eq:2D20d2ePmsys}
 \begin{gather}
 y_{k,l}+z_{k,l}= -\xi_{2k}, \label{eq:2D20d2ePmk} \\
 z_{k,l} + y_{k+1,l} = -\frac{a_{k+1}}{a_{k}}\xi_{2k}, \label{eq:2D20d2ePmkp}
 \end{gather}
 \end{subequations}
 where $\xi_{2k+1}$ is given by~\eqref{eq:2D20d2xikp}. Equation \eqref{eq:2D20d2ePmk} is not a dif\/ference equation and can be solved to give
 \begin{gather*}
 z_{k,l} = -\xi_{2k}-y_{k,l},% \label{eq:2D20d2zsol}
 \end{gather*}
 which inserted in \eqref{eq:2D20d2ePmkp} gives
 \begin{gather}
 y_{k+1,l} - y_{k,l} = \left( 1-\frac{a_{k+1}}{a_{k}} \right) \xi_{2k}. \label{eq:2D20d2ePmkp2}
 \end{gather}
 Def\/ining
 \begin{gather*}
 \xi_{2k} = - a_{k}\frac{b_{k+1}-b_{k}}{a_{k+1}-a_{k}}%\label{eq:2D20d2xik}
 \end{gather*}
 equation \eqref{eq:2D20d2ePmkp2} becomes a total dif\/ference. Therefore we can write the solution of the system~\eqref{eq:2D20d2ePmsys} as
 \begin{subequations} \label{eq:2D20d2zysol}
 \begin{gather}
 y_{k,l}= b_{k} + \beta_{l}, \label{eq:2D20d2ysol} \\
 z_{k,l} = a_{k}\frac{b_{k+1}-b_{k}}{a_{k+1}-a_{k}} -b_{k} - \beta_{l}. \label{eq:2D20d2zsol2}
 \end{gather}
 \end{subequations}

In this case we have the right number of arbitrary functions in both directions. So the solution of the $_{2}D_{2}$ equation with $\delta_{1}=0$ is just given by combining \eqref{eq:2D20d2wvsol} and \eqref{eq:2D20d2zysol}, gi\-ving~\eqref{eq:2D20d2sol}. It can be directly checked that \eqref{eq:2D20d2sol} is the general solution by inserting it into the~$_{2}D_{2}$ equation~\eqref{eq:2D2} with $\delta_{1}=0$.

\subsection[$_{3}D_{2}$ equation (\ref{eq:3D2})]{$\boldsymbol{{}_{3}D_{2}}$ equation (\ref{eq:3D2})}

From \cite{GSY_DarbouxI} we know that the $_{3}D_{2}$ equation \eqref{eq:3D2} is Darboux integrable, and that the form of the f\/irst integral depends on the value of the parameter $\delta$. We will begin with the general case when $\delta_{1}\neq0$ and $\delta\neq0$ and then consider the particular cases.

{\bf Case $\boldsymbol{\delta_{1}\neq0}$ and $\boldsymbol{\delta\neq0}$.} In this case we know that the $W_{1}$ f\/irst integrals of the $_{3}D_{2}$ equation~\eqref{eq:3D2} is given by~\cite{GSY_DarbouxI}
\begin{gather}
 W_{1} ={ \Fppp}{\alpha } {\frac { ( u_{{n-1,m}}+{ \delta_{2}})
[ 1+ ( u_{{n+1,m}}-1 ) { \delta_{1}} ] }{
 ( u_{{n+1,m}}+{ \delta_{2}} ) [ 1+ ( u_{{n-1,m}}-1 ) { \delta_{1}} ] }}\nonumber \\
 \hphantom{W_{1} =}{} +{\Fpmm}{ \alpha}\frac {u_{{n,m}}+( 1-\delta_{1}-\delta_{1}\delta_{2})u_{{n-1,m}}}{%
 u_{{n,m}}+( 1-\delta_{1}-\delta_{1}\delta_{2})u_{{n+1,m}}}\nonumber \\
\hphantom{W_{1} =}{} +{\Fmpm}{ \beta} (u_{{n+1,m}}- u_{{n-1,m}}) ( { \delta_{2}}+u_{{n,m}} )-{ \Fmmp}{ \beta}(u_{{n+1,m}} -u_{{n-1,m}}). \label{eq:W13D2}
\end{gather}
As stated in the introduction, from the relation $W_{1}=\xi_{n}$ this f\/irst integral def\/ines a three-point, second-order ordinary dif\/ference equation in the~$n$ direction which depends parametrically on~$m$. From this parametric dependence we f\/ind two dif\/ferent three-point non-autonomous ordinary dif\/ference equations corresponding to $m$ even and $m$ odd. We treat them separately.

{\bf Case $\boldsymbol{m=2l}$.} If $m=2l$ we have the following non-autonomous nonlinear ordinary dif\/ference equation
 \begin{gather}
{ \Fp{n}} {\frac { ( u_{{n-1,2l}}+{ \delta_{2}})[ 1+ ( u_{{n+1,2l}}-1 ) { \delta_{1}} ] }{ ( u_{{n+1,2l}}+{ \delta_{2}} )
[ 1+ ( u_{{n-1,2l}}-1) { \delta_{1}} ] }}+{\Fm{n}} (u_{{n+1,2l}}- u_{{n-1,2l}}) ( { \delta_{2}}+u_{{n,2l}} ) = \xi_{n}, \label{eq:3D2eQm}
 \end{gather}
 where we have chosen without loss of generality $\alpha=\beta=1$. We can apply the usual transformation \eqref{eq:genautsubl} in order to separate the even and odd part in \eqref{eq:3D2eQm}
 \begin{subequations} \label{eq:3D2eQmsys}
 \begin{gather}
 \frac {1+ ( w_{k,l}-1 ) { \delta_{1}}}{ w_{k,l}+{ \delta_{2}}} =\xi_{2k} \frac{1+ ( w_{k-1,l}-1 ) { \delta_{1}}}{ w_{k-1,l}+{ \delta_{2}}},\label{eq:3D2eQmk}\\
 v_{{k+1,l}}- v_{{k,l}} =\frac{ \xi_{2k+1}}{( { \delta_{2}}+w_{{k,l}} )}. \label{eq:3D2eQmkp}
 \end{gather}
 \end{subequations}
 This system of equations is still nonlinear, but the equation \eqref{eq:3D2eQmk} is uncoupled from \eqref{eq:3D2eQmkp}. Moreover equation \eqref{eq:3D2eQmk} is a discrete Riccati equation which can be linearized through the M\"obius transformation
 \begin{gather}
 w_{k,l} = -\delta_{2} +\frac{1}{W_{k,l}}, \label{eq:3D2mob}
 \end{gather}
 into
 \begin{subequations} \label{eq:3D2eQmsys2}
 \begin{gather}
 W_{k, l}-\xi_{2k} W_{k-1, l}= \frac{\delta_1}{\delta}(\xi_{2k}-1), \label{eq:3D2eQmk2} \\
 v_{{k+1,l}}- v_{{k,l}}=\xi_{2k+1}W_{k,l}, \label{eq:3D2eQmkp2}
 \end{gather}
\end{subequations}
 where $\delta$ is given by equation~\eqref{eq:deltaconstdef}. Putting $\xi_{2k}=a_{k}/a_{k-1}$ we have the following solution for~\eqref{eq:3D2eQmk2}
 \begin{gather}
 W_{k,l} = a_{k}\alpha_{l} - \frac{\delta_{1}}{\delta}. \label{eq:3D2Wklsol}
 \end{gather}
 Plugging \eqref{eq:3D2Wklsol} into equation \eqref{eq:3D2eQmkp2} and def\/ining
 \begin{gather}
 \xi_{2k+1}=b_{k+1}-b_{k}, \qquad a_{k} =\frac{c_{k+1}-c_{k}}{b_{k+1}-b_{k}}, \label{eq:3D2akdef}
 \end{gather}
 we have that equation \eqref{eq:3D2eQmkp2} becomes a total dif\/ference. Then the solution of \eqref{eq:3D2eQmkp2} can be written as
 \begin{gather*}
 v_{k,l} = -\frac{\delta_{1}}{\delta}b_{k} + c_{k}\alpha_{l} +\beta_{l}.% \label{eq:3D2vklsol}
 \end{gather*}
 So using \eqref{eq:3D2mob} we obtain the following
 solution for the original system \eqref{eq:3D2eQmsys}:
 \begin{subequations} \label{eq:3D2eQmsol}
 \begin{gather}
 w_{k,l} = \frac{\delta_{1}-1+\delta}{\delta_{1}}+ \frac{\delta(b_{k+1}-b_{k})}{ \delta\alpha_{m}(c_{k+1}-c_{k})-(b_{k+1}-b_{k})\delta_1}, \\
 v_{k,l} = -\frac{\delta_{1}}{\delta}b_{k} + c_{k}\alpha_{l} +\beta_{l}.
 \end{gather}
 \end{subequations}

{\bf Case $\boldsymbol{m=2l+1}$.} If $m=2l+1$ we have the following non-autonomous ordinary dif\/ference equation
 \begin{gather*}
 {\Fp{n}}\frac {u_{{n,2l+1}}+\delta u_{{n-1,2l+1}}}{ u_{{n,2l+1}}+\delta u_{{n+1,2l+1}}} -{ \Fm{n}}\left(u_{{n+1,2l+1}} -u_{{n-1,2l+1}}\right)=\xi_{n}.
 %\label{eq:3D2ePm}
 \end{gather*}
where without loss of generality $\alpha=\beta=1$ and $\delta$ is given by \eqref{eq:deltaconstdef}. Solving with respect to~$u_{n+1,2l+1}$ it is immediate to see that the resulting equation is linear. Then separating the even and the odd part using the transformation~\eqref{eq:genautsublp} we obtain the following system of linear, f\/irst-order ordinary dif\/ference equations
 \begin{subequations} \label{eq:3D2ePmsys}
 \begin{gather}
 z_{k,l} - \frac{1}{\xi_{2k}}z_{k-1,l} = \frac{1}{\delta}\left( 1-\frac{1}{\xi_{2k}} \right)y_{k,l}, \label{eq:3D2ePmk} \\
 y_{k+1}-y_{k}=-\xi_{2k+1}. \label{eq:3D2ePmkp}
 \end{gather}
 \end{subequations}
 As $\xi_{2k+1}=b_{k+1}-b_{k}$ we obtain the solution of equation \eqref{eq:3D2ePmkp}
 \begin{gather}
 y_{k,l} = -b_{k}+\gamma_{l}. \label{eq:3D2yklsol}
 \end{gather}
 Substituting $y_{k,l}$ given by~\eqref{eq:3D2yklsol} into equation~\eqref{eq:3D2ePmk} being $\xi_{2k}=a_{k}/a_{k-1}$, we obtain
 \begin{gather*}
 a_{k}z_{k,l}-a_{k-1}z_{k-1,l} = \frac{a_{k}-a_{k-1}}{\delta}( b_{k}-\gamma_{l}).% \label{eq:3D2ePmk2}
 \end{gather*}
 Then, in the usual way, we can represent the solution as
 \begin{gather}
 z_{k,l} = \frac{b_{k}-\gamma_{l}}{\delta} +\frac{b_{k+1}-b_{k}}{c_{k+1}-c_{k}} \left(\zeta_{l}-\frac{c_{k}}{\delta} \right), \label{eq:3D2zklsol}
 \end{gather}
 where we have used the explicit def\/inition of~$a_{k}$ given in~\eqref{eq:3D2akdef}. So we have the explicit expression for both f\/ields~$y_{k,l}$ and~$z_{k,l}$.

\begin{subequations} \label{eq:3D2deltagammaeq}
Equations (\ref{eq:3D2eQmsol}), (\ref{eq:3D2yklsol}), (\ref{eq:3D2zklsol}) provide the value of the four f\/ields, but we have too many arbitrary functions
in the $l$ direction, namely $\alpha_{l}$, $\beta_{l}$, $\gamma_{l}$ and~$\zeta_{l}$. Inserting (\ref{eq:3D2eQmsol}), (\ref{eq:3D2yklsol}), (\ref{eq:3D2zklsol}) into \eqref{eq:3D2} and separating the terms even and odd in $n$ and $m$ we obtain we obtain two equations
 \begin{gather}
 \zeta_{l}{\delta}^{2}\alpha_{{l}}+ ( \beta_{{l}} -{ \delta_{1}} \lambda ) \delta-{ \delta_{1}} \gamma _{{l}} =0, \\
 \zeta_{l}{\delta}^{2}\alpha_{{l+1}}+ ( \beta_{{l+1}} -{ \delta_{1}} \lambda ) \delta-{ \delta_{1}} \gamma _{{l}} =0,
 \end{gather}
\end{subequations}
which allow us to reduce by two the number of independent functions in the $l$ direction. Indeed solving \eqref{eq:3D2deltagammaeq} with respect to~$\gamma_{l}$ and~$\zeta_{l}$ we f\/ind
\begin{subequations} \label{eq:3D2deltagammadef}
 \begin{gather}
 \gamma_{l} = \frac{\delta}{\delta_{1}}\left( \beta_{l}-\lambda\delta_{1} - \alpha_{l}\frac{ \beta_{l+1}-\beta_{l} }{\alpha_{l+1}-\alpha_{l} }\right),
 \label{eq:3D2gammadef} \\
 \zeta_{l} =-\frac{1}{\delta}\frac{ \beta_{l+1}-\beta_{l} }{\alpha_{l+1}-\alpha_{l} }. \label{eq:3D2deltadef}
 \end{gather}
\end{subequations}

Inserting \eqref{eq:3D2deltagammadef} into (\ref{eq:3D2eQmsol}), (\ref{eq:3D2yklsol}), (\ref{eq:3D2zklsol}) we obtain the general solution~\eqref{eq:3D2solfin} of the $_{3}D_{2}$ equation~\eqref{eq:3D2} provided that $\delta_{1}\neq0$ and $\delta\neq0$. It is easy to see that the solution~\eqref{eq:3D2solfin} is ill-def\/ined if $\delta_{1}=0$ and if $\delta=0$. We will treat these two particular cases separately.

{\bf Case $\boldsymbol{\delta=0}$.} If $\delta=0$ we have that $\delta_{1}$ is given by equation \eqref{eq:deltasolved}. In this case the f\/irst integral~\eqref{eq:W13D2} is singular since the coef\/f\/icient of $\alpha$ goes to a constant. Following \cite{GSY_DarbouxI} we have that the $_{3}D_{2}$
equation with $\delta_{1}$ given by \eqref{eq:deltasolved} possesses the following f\/irst integral in the direction~$n$
\begin{gather}
 W_{1}^{( ( 1+\delta_{2})^{-1},\delta_{2})}= {\Fppp} { \alpha} {\frac {u_{{n+1,m}}-u_{{n-1,m}}}{
( {\delta_2}+u_{{n+1,m}} ) ( { \delta_{2}}+u_{{n-1,m}}) }} +{\Fpmm} { \alpha} {\frac {u_{{n+1,m}}-u_{{n-1,m}}}{( { \delta_{2}}+1 ) u_{{n,m}}}}\nonumber \\
\hphantom{ W_{1}^{( ( 1+\delta_{2})^{-1},\delta_{2})}=}{} -{\Fmpm} { \beta} ( u_{{n-1,m}}-u_{{n+1,m}} )
( { \delta_{2}}+u_{{n,m}} )\nonumber \\
\hphantom{ W_{1}^{( ( 1+\delta_{2})^{-1},\delta_{2})}=}{}-{ \Fmmp} { \beta} ( u_{{n+1,m}}-u_{{n-1,m}} ). \label{eq:W13D2d1d2}
\end{gather}
This f\/irst integral is a three-point, second-order f\/irst integral. As in the general case we consider separately the $m$ even and odd cases.

{\bf Case $\boldsymbol{m=2l}$.} If $m=2l$ then the f\/irst integral \eqref{eq:W13D2d1d2} becomes the following nonlinear three-point, second-order dif\/ference equation
 \begin{gather*}
Fp{n} {\frac {u_{{n+1,2l}}-u_{{n-1,2l}}}{ ( {\delta_2}+u_{{n+1,2l}}) ( { \delta_{2}}+u_{{n-1,2l}} ) }}
 -\Fm{n} ( u_{{n-1,2l}}-u_{{n+1,2l}} )( { \delta_{2}}+u_{{n,2l}}) =\xi_{n}, \label{eq:3D2d1d2eQm}
 \end{gather*}
 where without loss of generality $\alpha=\beta=1$. If we separate the even and the odd part using the general transformation given by \eqref{eq:genautsubl} we obtain the system
 \begin{subequations} \label{eq:3D2d1d2eQmsys}
 \begin{gather}
 \frac{w_{k,l}-w_{k-1,l}}{(w_{k,l}+\delta_2)(w_{k-1,l}+\delta_2)}=\xi_{2k}, \label{eq:3D2d1d2eQmk} \\
 (v_{k+1,l}-v_{k,l})(w_{k,l}+\delta_2)=\xi_{2k+1}. \label{eq:3D2d1d2eQmkp}
 \end{gather}
 \end{subequations}
 This is a system of f\/irst-order nonlinear dif\/ference equations. However \eqref{eq:3D2d1d2eQmk} is uncoupled from \eqref{eq:3D2eQmkp}, and it is a discrete Riccati equation which can be linearized through the M\"obius transformation \eqref{eq:3D2mob}. This linearize the system \eqref{eq:3D2d1d2eQmsys} to
 \begin{subequations} \label{eq:3D2d1d2eQmsys2}
 \begin{gather}
 W_{k,l}-W_{k-1,l}=\xi_{2k}, \label{eq:3D2d1d2eQmk2} \\
 v_{k+1,l}-v_{k,l}=\xi_{2k+1} W_{k,l}. \label{eq:3D2d1d2eQmkp2}
 \end{gather}
 \end{subequations}
 Def\/ining $\xi_{2k}=a_{k}-a_{k-1}$ equation \eqref{eq:3D2d1d2eQmk2} is solved by
 \begin{gather}
 W_{k,l} = a_{k}+\beta_{l}. \label{eq:3D2d1d2Wsol}
 \end{gather}
 Introducing \eqref{eq:3D2d1d2Wsol} into equation \eqref{eq:3D2d1d2eQmkp2} we have
 \begin{gather}
 v_{k+1,l}-v_{k,l} = \xi_{2k+1} ( a_{k}+\alpha_{l}). \label{eq:3D2d1d2eQmkp3}
 \end{gather}
 Equation \eqref{eq:3D2d1d2eQmkp3} becomes a total dif\/ference if
 \begin{gather}
 \xi_{2k+1} = b_{k+1}-b_{k}, \qquad a_{k} = \frac{c_{k+1}-c_{k}}{b_{k+1}-b_{k}}. \label{eq:3D2d1d2lkpadef}
 \end{gather}
 This yields the following solution of the system~\eqref{eq:3D2d1d2eQmsys}
 \begin{subequations} \label{eq:3D2d1d2wvsol}
 \begin{gather}
 v_{k,l} = c_{k}+ b_{k}\alpha_{l} + \beta_{l}, \label{eq:3D2d1d2vsol} \\
 w_{k,l} = -\delta_{2} + \frac{b_{k+1}-b_{k}}{ c_{k+1}-c_{k} + \alpha_{l}( b_{k+1} -b_{k} )}. \label{eq:3D2d1d2wsol}
 \end{gather}
 \end{subequations}

{\bf Case $\boldsymbol{m=2l+1}$.} If $m=2l+1$ the f\/irst integral \eqref{eq:W13D2d1d2} becomes the following nonlinear, three-point, second-order dif\/ference equation
 \begin{gather}
 {\Fp{n}} {\frac {u_{{n+1,2l+1}}-u_{{n-1,2l+1}}}{( { \delta_{2}}+1) u_{{n,2l+1}}}} -{ \Fm{n}} ( u_{{n+1,2l+1}}-u_{{n-1,2l+1}} )=\xi_{n}, \label{eq:3D2d1d2ePm}
 \end{gather}
 where without loss of generality $\alpha=\beta=1$. As usual we can separate the even and odd part in $n$ using the transformation \eqref{eq:genautsublp}. This transformation brings equation \eqref{eq:3D2d1d2ePm} into the following linear system
 \begin{subequations} \label{eq:3D2d1d2ePmsys}
 \begin{gather}
 z_{k,l}-z_{k-1,l} = (\delta_2+1)y_{k,l} \left( \frac{c_{k}-c_{k-1}}{b_{k}-b_{k-1}}-\frac{c_{k+1}-c_{k}}{b_{k+1}-b_{k}} \right), \label{eq:3D2d1d2ePmk} \\
 y_{k+1,l}-y_{k,l} =-b_{k+1}+b_{k}, \label{eq:3D2d1d2ePmkp}
 \end{gather}
 \end{subequations}
 where we used \eqref{eq:3D2d1d2lkpadef} and the def\/inition $\xi_{2k+1} = a_{k+1}-a_{k}$. Equation \eqref{eq:3D2d1d2ePmkp} is readily solved and gives
 \begin{gather}
 y_{k,l} = \gamma_{l}-b_{k}. \label{eq:3D2d1d2ysol}
 \end{gather}
 Inserting \eqref{eq:3D2d1d2ysol} into \eqref{eq:3D2d1d2ePmk} we obtain
 \begin{gather*}
 z_{k,l}-z_{k-1,l} = (\delta_2+1)( \gamma_{l}-b_{k}) \left( \frac{c_{k}-c_{k-1}}{b_{k}-b_{k-1}}-\frac{c_{k+1}-c_{k}}{b_{k+1}-b_{k}} \right).
 %\label{eq:3D2d1d2ePmk2}
 \end{gather*}
 We can then write for $z_{k,l}$ the following expression
 \begin{gather*}
 z_{k,l} = -(\delta_2+1)\left( \gamma_{l} \frac{c_{k+1}-c_{k}}{b_{k+1}-b_{k}} +d_{k}\right)+\zeta_{l},
% \label{eq:3D2d1d2zsol}
 \end{gather*}
 where $d_{k}$ solves the equation
 \begin{gather}
 d_{k}-d_{k-1} = -b_{k} \left( \frac{c_{k}-c_{k-1}}{b_{k}-b_{k-1}}-\frac{c_{k+1}-c_{k}}{b_{k+1}-b_{k}} \right). \label{eq:3D2d1d2dkeq}
 \end{gather}
 Equation \eqref{eq:3D2d1d2dkeq} is a total dif\/ference with $d_{k}$ given by
 \begin{gather*}
 d_{k} = b_{k+1} \frac{c_{k+1}-c_{k}}{b_{k+1}-b_{k}} -c_{k+1}. %\label{eq:3D2d1d2dksol}
 \end{gather*}
 Therefore we have the following solution to the system \eqref{eq:3D2d1d2ePmsys}
 \begin{subequations} \label{eq:3D2d1d2zysol}
 \begin{gather}
 y_{k,l} = \gamma_{l}-b_{k}, \label{eq:3D2d1d2ysol2} \\
 z_{k,l}= -(\delta_2+1)\left[ ( \gamma_{l}+b_{k+1} ) \frac{c_{k+1}-c_{k}}{b_{k+1}-b_{k}} -c_{k+1}\right]+\zeta_{l}. \label{eq:3D2d1d2zsol2}
 \end{gather}
 \end{subequations}

Equations (\ref{eq:3D2d1d2wvsol}), (\ref{eq:3D2d1d2zysol}) provide the value of the four f\/ields, but we have too many arbitrary functions in the~$l$ direction, namely $\alpha_{l}$, $\beta_{l}$, $\gamma_{l}$ and $\zeta_{l}$. Inserting (\ref{eq:3D2d1d2wvsol}), (\ref{eq:3D2d1d2zysol}) into~\eqref{eq:3D2} with~$\delta_{1}$ given by~\eqref{eq:deltasolved} and separating the terms even and odd in~$n$ and~$m$ we obtain we obtain two equations
%\begin{subequations} \label{eq:3D2d1d2cc}
 \begin{gather*}
 \gamma_{l} (\delta_2+1) \alpha_{l}-\lambda+\zeta_{l}+\beta_{l} (\delta_2+1) =0, \\
 \gamma_{l} (\delta_2+1) \alpha_{l+1}-\lambda+\zeta_{l}+\beta_{l+1} (\delta_2+1) =0.
 \end{gather*}
%\end{subequations}
Solving this compatibility condition with respect to $\gamma_{l}$ and $\zeta_{l}$ we obtain
\begin{subequations} \label{eq:3D2d1d2gamdel}
 \begin{gather}
 \gamma_{l} = -\frac{\beta_{l+1}-\beta_{l}}{\alpha_{l+1}-\alpha_{l}}, \label{eq:3D2d1d2gam} \\
 \zeta_{l} = ( \delta_{2}+1 ) \frac{\beta_{l+1}\alpha_{l}-\alpha_{l+1}\beta_{l}}{\alpha_{l+1}-\alpha_{l}}+\lambda . \label{eq:3D2d1d2del}
 \end{gather}
\end{subequations}

Inserting then \eqref{eq:3D2d1d2gamdel} into (\ref{eq:3D2d1d2wvsol}), (\ref{eq:3D2d1d2zysol}) we obtain general solution~\eqref{eq:3D2d1d1sol} of the $_{3}D_{2}$ equation~\eqref{eq:3D2} provided that $\delta=0$.

{\bf Case $\boldsymbol{\delta_{1}=0}$.} The f\/irst integral \eqref{eq:W13D2} is non-singular when $\delta_{1}=0$. Therefore the procedure of solution becomes dif\/ferent only when we arrive to the systems of ordinary dif\/ference equations~\eqref{eq:3D2eQmsys} and~\eqref{eq:3D2ePmsys}. So we present the solution of the systems in this case.

{\bf Case $\boldsymbol{m=2l}$.} If $\delta_{1}=0$ the system~\eqref{eq:3D2eQmsys} becomes
 \begin{subequations} \label{eq:3D20d2eQmsys}
 \begin{gather}
 w_{k,l}+{ \delta_{2}}= \frac{w_{k-1,l}+{ \delta_{2}}}{\xi_{2k}}, \label{eq:3D20d2eQmk} \\
 v_{{k+1,l}}- v_{{k,l}}=\frac{ \xi_{2k+1}}{( { \delta_{2}}+w_{{k,l}})}. \label{eq:3D20d2eQmkp}
 \end{gather}
 \end{subequations}
 The system \eqref{eq:3D20d2eQmsys} is nonlinear, but equation \eqref{eq:3D20d2eQmk} is uncoupled from equation \eqref{eq:3D20d2eQmk}. Def\/ining $\xi_{2k}=a_{k-1}/a_{k}$ equation \eqref{eq:3D20d2eQmk} is solved by
 \begin{gather}
 w_{k,l} = -\delta_{2} + a_{k}\alpha_{l}. \label{eq:3D20d2wsol}
 \end{gather}
 Substituting $w_{k,l}$ given by \eqref{eq:3D20d2wsol} into equation \eqref{eq:3D20d2eQmkp} we obtain
 \begin{gather}
 v_{{k+1,l}}- v_{{k,l}} =\frac{ \xi_{2k+1}}{ a_{k}\alpha_{l}}. \label{eq:3D20d2eQmkp2}
 \end{gather}
 Def\/ining
 \begin{gather}
 \xi_{2k+1} = - a_{k}( b_{k+1}-b_{_k} ), \label{eq:3D20d2xikp}
 \end{gather}
 we have that equation \eqref{eq:3D20d2eQmkp2} is a total dif\/ference. Therefore we have the following solution of the system \eqref{eq:3D20d2eQmsys}
 \begin{subequations} \label{eq:3D20d2wvsol}
 \begin{gather}
 v_{k,l} = \beta_{l} + \frac{b_{k}}{\alpha_{l}}, \label{eq:3D20d2vsol} \\
 w_{k,l} = -\delta_{2} + a_{k}\alpha_{l}. \label{eq:3D20d2wsol2}
 \end{gather}
 \end{subequations}

{\bf Case $\boldsymbol{m=2l+1}$.} If $\delta_{1}=0$ the system \eqref{eq:1D2ePmsys} becomes
 \begin{subequations} \label{eq:3D20d2ePmsys}
 \begin{gather}
 z_{k,l}-\frac{a_{k}}{a_{k-1}}z_{k-1,l} =\left(\frac{a_{k}}{a_{k-1}}-1\right)y_{k,l}, \label{eq:3D20d2ePmk} \\
 y_{k+1,l}-y_{k,l} =-a_{k}( b_{k+1}-b_{k} ), \label{eq:3D20d2ePmkp}
 \end{gather}
 \end{subequations}
 where we used \eqref{eq:3D20d2xikp} and $\xi_{2k}=a_{k-1}/a_{k}$. The system is linear and equation~\eqref{eq:3D20d2ePmkp} is uncoupled from~\eqref{eq:3D20d2ePmk}. If we put
 \begin{gather*}
 a_{k} = -\frac{c_{k+1}-c_{k}}{b_{k+1}-b_{k}},% \label{eq:3D20d2ak}
 \end{gather*}
 then equation \eqref{eq:3D20d2ePmk} becomes a~total dif\/ference whose solution is
 \begin{gather}
 y_{k,l} = c_{k}+\gamma_{l}. \label{eq:3D20d2ysol}
 \end{gather}
Substituting $y_{k,l}$ given by \eqref{eq:3D20d2ysol} into equation \eqref{eq:3D20d2ePmk} we obtain
 \begin{gather*}
 \frac{b_{k+1}-b_{k}}{c_{k+1}-c_{k}} z_{k,l} -\frac{b_{k}-b_{k-1}}{c_{k}-c_{k-1}} z_{k-1,l} =\left( \frac{b_{k}-b_{k-1}}{c_{k}-c_{k-1}} -\frac{b_{k+1}-b_{k}}{c_{k+1}-c_{k}} \right) ( c_{k}+\gamma_{l} ).% \label{eq:3D20d2ePmk2}
 \end{gather*}
 We can therefore represent the solution as
 \begin{gather*}
 z_{k,l} = \frac{c_{k+1}-c_{k}}{b_{k+1}-b_{k}} ( d_{k}+\zeta_{l} ) -\gamma_{l},% \label{eq:3D20d2zsol}
 \end{gather*}
 where $d_{k}$ solves the equation
 \begin{gather}
 d_{k}-d_{k-1}=b_{k+1} - \frac{b_{k+1}-b_{k}}{c_{k+1}-c_{k}}c_{k+1} -b_{k-1} + \frac{b_{k}-b_{k-1}}{c_{k}-c_{k-1}}c_{k-1}. \label{eq:3D20d2dkeq}
 \end{gather}
 Equation \eqref{eq:3D20d2dkeq} is a total dif\/ference and $d_{k}$ is given by
 \begin{gather*}
 d_{k} = b_{k} - \frac{b_{k+1}-b_{k}}{c_{k+1}-c_{k}}c_{k}.% \label{eq:3D20d2dksol}
 \end{gather*}
 Therefore we have that the solution of the system \eqref{eq:3D20d2ePmsys} is given by
 \begin{subequations} \label{eq:3D20d2zysol}
 \begin{gather}
 y_{k,l} = c_{k}+\gamma_{l}, \label{eq:3D20d2ysol2} \\
 z_{k,l} = \frac{c_{k+1}-c_{k}}{b_{k+1}-b_{k}}\zeta_{l}-\gamma_{l} +\frac{b_{k}c_{k+1}-c_{k}b_{k+1}}{b_{k+1}-b_{k}}. \label{eq:3D20d2zsol2}
 \end{gather}
 \end{subequations}

From equations (\ref{eq:3D20d2wvsol}), (\ref{eq:3D20d2zysol}) we have the value of the four f\/ields, but we have too many arbitrary functions in the $l$ direction, namely $\alpha_{l}$, $\beta_{l}$, $\gamma_{l}$ and $\zeta_{l}$. Inserting (\ref{eq:3D20d2wvsol}), (\ref{eq:3D20d2zysol}) into \eqref{eq:3D2} with $\delta_{1}=0$ and separating the terms even and odd in $n$ and $m$ we obtain we obtain two equations
%\begin{subequations} \label{eq:3D20d2cc}
 \begin{gather*}
 \beta_{l}\alpha_{l}-\zeta_{l}=0, \qquad \beta_{l+1}\alpha_{l+1}-\zeta_{l}=0.
 \end{gather*}
%\end{subequations}
We can solve this compatibility conditions with respect to $\beta_{l}$ and $\zeta_{l}$ we obtain
%\begin{subequations}
 \begin{gather}\label{eq:3D20d2bedel}
 \beta_{l} = \frac{\zeta_{0}}{\alpha_{l}}, \qquad \zeta_{l}=\zeta_{0},
 \end{gather}
where $\zeta_{0}$ is a constant.

Inserting then \eqref{eq:3D20d2bedel} into (\ref{eq:3D20d2wvsol}), (\ref{eq:3D20d2zysol}) we obtain general solution~\eqref{eq:3D20d2sol} of the~$_{3}D_{2}$ equation~\eqref{eq:3D2} provided that $\delta_{1}=0$.

This discussion exhausts the possible cases. So for any value of the parameters we have the general solution of the $_{3}D_{2}$ equation \eqref{eq:3D2}.

\subsection[$_{1}D_{4}$ equation (\ref{eq:1D4})]{$\boldsymbol{{}_{1}D_{4}}$ equation (\ref{eq:1D4})}

To f\/ind the general solution we start from the $_{1}D_{4}$ equation \eqref{eq:1D4} itself. Applying the general transformation \eqref{eq:genautsub} we transform the $_{1}D_{4}$ into the following system
\begin{subequations} \label{eq:1D4sys}
 \begin{gather}
 v_{k,l} z_{k,l}+w_{k,l} y_{k,l}+\delta_{1} w_{k,l} z_{k,l} +\delta_{2} y_{k,l} z_{k,l}+\delta_{3} =0, \label{eq:1D4a} \\
 y_{k,l} w_{k,l+1}+z_{k,l} v_{k,l+1} +\delta_{1} z_{k,l} w_{k,l+1}+\delta_{2} y_{k,l} z_{k,l}+\delta_{3} =0, \label{eq:1D4b} \\
 w_{k,l} y_{k+1,l}+v_{k+1,l} z_{k,l} +\delta_{1} w_{k,l} z_{k,l}+\delta_{2} z_{k,l} y_{k+1,l}+\delta_{3} =0, \label{eq:1D4c} \\
 z_{k,l} v_{k+1,l+1}+y_{k+1,l} w_{k,l+1} +\delta_{1} z_{k,l} w_{k,l+1}+\delta_{2} z_{k,l} y_{k+1,l}+\delta_{3} =0. \label{eq:1D4d}
 \end{gather}
\end{subequations}

From the system \eqref{eq:1D4sys} we have four dif\/ferent way for calculating $z_{k,l}$. This means that we have some compatibility conditions. Indeed from~\eqref{eq:1D4a} and \eqref{eq:1D4c} we obtain the following equation for~$v_{k+1,l}$
\begin{gather}
 v_{k+1,l} = \frac{\delta_{3}+w_{k,l} y_{k+1,l}}{ \delta_{3}+w_{k,l} y_{k,l}}v_{k,l} +\frac{(y_{k+1,l}-y_{k,l}) \big(\delta_{1} w_{k,l}^2-\delta_{2} \delta_{3}\big)}{ \delta_{3}+w_{k,l} y_{k,l}}, \label{eq:1D4evk}
\end{gather}
while from \eqref{eq:1D4b} and \eqref{eq:1D4d} we obtain the following equation for $v_{k+1,l+1}$
\begin{gather}
 v_{k+1,l+1} = \frac{\delta_{3}+y_{k+1,l} w_{k,l+1}}{ \delta_{3}+y_{k,l} w_{k,l+1}}v_{k,l+1} +\frac{(y_{k+1,l}-y_{k,l}) \big(\delta_{1} w_{k,l+1}^2-\delta_{2} \delta_{3}\big)}{y_{k,l} w_{k,l+1}+\delta_{3}}. \label{eq:1D4evkp}
\end{gather}
Equations \eqref{eq:1D4evk} and \eqref{eq:1D4evkp} give rise to a compatibility condition between $v_{k+1,l}$ and its shift in the $l$ direction $v_{k+1,l+1}$ which is given by
\begin{gather*}
 \left[
 \begin{aligned}
(y_{k+1,l+1} y_{k,l} -y_{k+1,l} y_{k,l+1} )w_{k,l+1} \\
{} +\delta_{3} (y_{k+1,l+1}+ y_{k,l}
 - y_{k,l+1}- y_{k+1,l} ) \end{aligned}
 \right]
\big(v_{k,l+1} w_{k,l+1}-\delta_{2} \delta_{3}+\delta_{1} w_{k,l+1}^2\big) =0. %\label{eq:1D4vkvkpcc}
\end{gather*}
Discarding the trivial solution
\begin{gather*}
 v_{k,l} = -\delta_{1} w_{k,l} +\frac{\delta_{2} \delta_{3}}{w_{k,l}}
\end{gather*}
we obtain the following value for the f\/ield $w_{k,l}$
\begin{gather}
 w_{k,l} = \delta_{3} \frac{y_{k+1,l-1}-y_{k+1,l}-y_{k,l-1}+y_{k,l}}{ y_{k+1,l} y_{k,l-1}-y_{k+1,l-1} y_{k,l}}, \label{eq:1D4wkldef}
\end{gather}
which makes \eqref{eq:1D4evk} and \eqref{eq:1D4evkp} compatible. This gives us the f\/irst piece of the solution in~\eqref{eq:1D4wsolfin}. Then we have to solve the following equation for~$v_{k,l}$
\begin{gather*}
 v_{k+1,l} =\frac{y_{k+1,l}-y_{k+1,l-1}}{y_{k,l}-y_{k,l-1}}v_{k,l}+ \frac{\delta_{1}\delta_{3} (y_{k+1,l-1}-y_{k,l-1}-y_{k+1,l}+y_{k,l})^2 }{ (y_{k+1,l-1} y_{k,l}-y_{k+1,l} y_{k,l-1}) (y_{k,l}-y_{k,l-1})} \\
\hphantom{v_{k+1,l} =}{} -\frac{\delta_{2} (y_{k+1,l-1} y_{k,l}-y_{k+1,l} y_{k,l-1})}{ y_{k,l}-y_{k,l-1}}.% \label{eq:1D4evk2}
\end{gather*}
Making the transformation
\begin{gather}
 v_{k,l} = ( y_{k,l}-y_{k,l-1} )V_{k,l} +\frac{\delta_{1}\delta_{3}}{y_{k,l-1}}-\delta_{2}y_{k,l-1} \label{eq:1D4Vkldef}
\end{gather}
we obtain that $V_{k,l}$ satisf\/ies the dif\/ference equation
\begin{gather}
V_{k+1,l}=V_{k,l}+ \frac{\delta_{1} \delta_{3} (y_{k,l-1}-y_{k+1,l-1})^2}{ y_{k,l-1} y_{k+1,l-1}(y_{k+1,l-1} y_{k,l} -y_{k+1,l} y_{k,l-1}) }. \label{eq:1D4evk3}
\end{gather}

To go further we need to specify the form of the f\/ield $y_{k,l}$. This can be obtained from the Darboux integrability of the $_{1}D_{4}$ equation~\eqref{eq:1D4}. From \cite{GSY_DarbouxI} we know that the $_{1}D_{4}$ equation~\eqref{eq:1D4} we have the following four-point, third-order
$W_{1}$ integral
\begin{gather*}
 W_{1} ={ \Fppp}{\alpha } \frac {u_{n+1,m}^{2}{ \delta_{1}}+u_{{n+1,m}}u_{{n+2,m}}+u_{{n-1,m}}
( u_{{n,m}}-u_{{n+2,m}} ) -{ \delta_{2}}{ \delta_{3}}}{ u_{n+1,m}( \delta_{1}+u_{n,m}) -\delta_{2}\delta_{3}} \\
 \hphantom{W_{1} =}{} +{ \Fpmm}{\alpha} \frac { ( u_{{n,m}}-u_{{n+2,m}}+{\delta_{1}}u_{{n+1,m}} )
 u_{{n-1,m}}+u_{{n+1,m}}u_{{n+2,m}}}{ ( u_{{n,m}}+{ \delta_{1}}u_{{n-1,m}} ) u_{{n+1,m}}} \\
 \hphantom{W_{1} =}{} +{ \Fmpm}{\beta} \frac { (u_{{n+1,m}} -u_{{n-1,m}} )
 (u_{{n+2,m}}- u_{{n,m}} ) }{ u_{n,m}^{2}{ \delta_{1}}+u_{{n+1,m}}u_{{n,m}} -{ \delta_{2}}{ \delta_{3}}}\\
 \hphantom{W_{1} =}{} +{ \Fmmp}{ \beta} {\frac {( u_{{n+1,m}}-u_{{n-1,m}} )
( u_{{n+2,m}}-u_{{n,m}}) }{
 u_{{n,m}} ( u_{{n+2,m}}{ \delta_{1}}+u_{{n+1,m}}) }},% \label{eq:W11D4}
\end{gather*}
This f\/irst integral def\/ines as always the relation $W_{1}=\xi_{n}$ which is a third-order, four-point ordinary dif\/ference equation in the $n$ direction depending parametrically on $m$. In particular when $m=2l+1$ we have the equation
\begin{gather}
{\Fp{n}}\frac{( u_{{n,2l+1}}-u_{{n+2,2l+1}} +{\delta_{1}}u_{{n+1,2l+1}}) u_{{n-1,2l+1}}+u_{{n+1,2l+1}}u_{{n+2,2l+1}}}{
( u_{{n,2l+1}}+{ \delta_{1}}u_{{n-1,2l+1}} ) u_{{n+1,2l+1}}}\nonumber \\
\qquad{} +\Fm{n}{\frac { ( u_{{n+1,2l+1}}-u_{{n-1,2l+1}} ) ( u_{{n+2,2l+1}}-u_{{n,2l+1}} ) }{
 u_{{n,2l+1}} ( u_{{n+2,2l+1}}{ \delta_{1}}+u_{{n+1,2l+1}}) }} =\xi_{n}. \label{eq:1D4ePm}
\end{gather}
where we have chosen without loss of generality $\alpha=\beta=1$. Using the transformation \eqref{eq:genautsublp} then~\eqref{eq:1D4ePm} is converted into the system
\begin{subequations} \label{eq:1D4ePmsys}
 \begin{gather}
 (y_{k,l}-y_{k+1,l}+\delta_{1} z_{k,l}) z_{k-1,l} +y_{k+1,l} z_{k,l}=\xi_{2 k}(y_{k,l}+\delta_{1} z_{k-1,l}) z_{k,l}, \label{eq:1D4ePmk} \\
 (y_{k,l}-y_{k+1,l}) (z_{k,l}-z_{k+1,l})=\xi_{2 k+1}z_{k,l} (z_{k+1,l} \delta_{1}+y_{k+1,l}). \label{eq:1D4ePmkp}
 \end{gather}
\end{subequations}
If we solve \eqref{eq:1D4ePmkp} with respect to $z_{k+1,l}$ and then substitute into \eqref{eq:1D4ePmk} we obtain a \emph{linear}, second-order ordinary dif\/ference equation for $y_{k,l}$
\begin{gather}
 \xi_{2 k-1}y_{k+1,l} +(1-\xi_{2k}-\xi_{2k} \xi_{2 k-1}) y_{k,l} -(1-\xi_{2k}) y_{k-1,l}=0. \label{eq:1D4ePmk2}
\end{gather}
We can f\/ind the solution to this equation in a similar manner than in the case of the $D_{3}$ equation. First of all let us consider $Y_{k,l}=a_{k}y_{k,l}+b_{k}y_{k-1,l}$ such that $Y_{k+1,l}-Y_{k,l}$ is equal to the left-hand side of \eqref{eq:1D4ePmk2}. To this end we def\/ine
\begin{gather*}
 \xi_{2k} = -\frac{b_{k+1}-b_{k}-a_{k}}{a_{k+1}}, \qquad \xi_{2k-1} = \frac{a_{k+1}-a_{k}+b_{k+1}-b_{k}}{b_{k}}. %\label{eq:1d4lkkpdef}
\end{gather*}
Therefore $y_{k,l}$ must solve the following equation
\begin{gather}
 a_{k}y_{k,l}+b_{k}y_{k-1,l}=\alpha_{l}. \label{eq:1D4ePmk3}
\end{gather}
Equation \eqref{eq:1D4ePmk3} is reduced to a total dif\/ference if we impose
\begin{gather*}
 a_{k} = \frac{1}{c_{k}}\frac{1}{d_{k}-d_{k-1}}, \qquad b_{k} =-\frac{1}{c_{k-1}}\frac{1}{d_{k}-d_{k-1}}. %\label{eq:1D4abdef}
\end{gather*}
Then the solution of \eqref{eq:1D4ePmk3} is then given by
\begin{gather}
 y_{k,l}=c_{k}( \alpha_{l}d_{k}+\beta_{l} ). \label{eq:1D4yklsol}
\end{gather}
This is just equation \eqref{eq:1D4ysolfin}.

Inserting \eqref{eq:1D4yklsol} into \eqref{eq:1D4evk3} we obtain
\begin{gather*}
 V_{k+1,l} = V_{k,l} +\frac{\delta_{1}\delta_{3}\alpha_{l-1}}{ \beta_{l-1} \alpha_{l}-\beta_{l} \alpha_{l-1}}
 \left[ \frac{1}{(\alpha_{l-1} d_{k+1}+\beta_{l-1}) c_{k+1}^2} -\frac{1}{(\alpha_{l-1} d_{k}+\beta_{l-1}) c_{k}^2} \right] \\
\hphantom{V_{k+1,l} =}{} -\frac{\delta_{1}\delta_{3}}{\beta_{l-1} \alpha_{l}-\beta_{l} \alpha_{l-1}}
 \frac{(c_{k}-c_{k+1})^2}{ c_{k}^2 c_{k+1}^2 (d_{k+1}-d_{k})}.% \label{eq:1D4evk4}
\end{gather*}
This means that the solution of $V_{k,l}$ can be represented as
\begin{gather}
 V_{k,l}= \gamma_{l} +\frac{ \delta_{1} \delta_{3}}{ \beta_{l-1} \alpha_{l}-\beta_{l} \alpha_{l-1} }
 \left[\frac{\alpha_{l-1}}{c_{k}^2 (\alpha_{l-1} d_{k}+\beta_{l-1})} +e_{k}\right], \label{eq:1D4Vklsol}
\end{gather}
where $e_{k}$ is def\/ined by the discrete integration
\begin{gather*}
 e_{k+1}-e_{k}= -\frac{(c_{k}-c_{k+1})^2}{ c_{k}^2 c_{k+1}^2 (d_{k+1}-d_{k})}, %\label{eq:1D4ekdef}
\end{gather*}
which is just equation \eqref{eq:1D4ekeqdef0}. Inserting the value of $V_{k,l}$ from \eqref{eq:1D4Vklsol} and the value of~$y_{k,l}$~\eqref{eq:1D4yklsol} into equation \eqref{eq:1D4Vkldef} we obtain equation \eqref{eq:1D4vsolfin}. From the obtained value of~$v_{k,l}$ we can compute~$w_{k,l}$ using \eqref{eq:1D4wkldef}. So f\/inally we can compute $z_{k,l}$ from the original system~\eqref{eq:1D4sys}, and we obtain a \emph{single} compatibility condition given by
\begin{gather*}
 \left( \beta_{{l}}\alpha_{{l+1}}-\beta_{{l+1}}\alpha_{{l}} \right) \gamma _{{l+1}} - \left( \beta_{{l-1}}\alpha_{{l}}-\beta_{{l}}\alpha_{{l-1}} \right) \gamma _{{l}} = \left( \beta_{{l-1}}\alpha_{{l}}-\beta_{{l}}\alpha_{{l-1}} \right) {\delta_2}, %\label{eq:1D4gammleq}
\end{gather*}
which is just equation \eqref{eq:1D4gammleq0}. Since now the system \eqref{eq:1D4sys} is compatible we can use any of its equations to compute $z_{k,l}$. E.g., using equation \eqref{eq:1D4a} and the value of $w_{k,l}$ from equation \eqref{eq:1D4wkldef} we obtain equation~\eqref{eq:1D4wsolfin}. This concludes the procedure of solution of the $_{1}D_{4}$ equation~\eqref{eq:1D4}.

\subsection[$_{2}D_{4}$ equation (\ref{eq:2D4})]{$\boldsymbol{{}_{2}D_{4}}$ equation (\ref{eq:2D4})}

To f\/ind the general solution of the $_{2}D_{4}$ equation \eqref{eq:2D4} we start from the equation itself. Applying the general transformation \eqref{eq:genautsub} the $_{2}D_{4}$ equation we obtain the following system
\begin{subequations} \label{eq:2D4sys}
 \begin{gather}
 v_{k,l} w_{k,l}+\delta_{2} w_{k,l} y_{k,l} +\delta_{1} w_{k,l} z_{k,l}+y_{k,l} z_{k,l}+\delta_{3} =0, \label{eq:2D4a} \\
 v_{k,l+1} w_{k,l+1}+\delta_{2} y_{k,l} w_{k,l+1} +\delta_{1} z_{k,l} w_{k,l+1}+y_{k,l} z_{k,l}+\delta_{3} =0, \label{eq:2D4b} \\
 w_{k,l} v_{k+1,l}+\delta_{2} w_{k,l} y_{k+1,l} +\delta_{1} w_{k,l} z_{k,l}+z_{k,l} y_{k+1,l}+\delta_{3} =0, \label{eq:2D4c} \\
 w_{k,l+1} v_{k+1,l+1}+\delta_{2} y_{k+1,l} w_{k,l+1} +\delta_{1} z_{k,l} w_{k,l+1}+z_{k,l} y_{k+1,l}+\delta_{3} =0. \label{eq:2D4d}
 \end{gather}
\end{subequations}

From the system \eqref{eq:2D4sys} we have four dif\/ferent ways to compute $z_{k,l}$. This means that we have some compatibility conditions. Indeed from \eqref{eq:2D4a} and \eqref{eq:2D4c} we obtain the following equation for~$v_{k+1,l}$
\begin{gather}
 v_{k+1,l} = \frac{\delta_{1} w_{k,l}+y_{k+1,l}}{ \delta_{1} w_{k,l}+y_{k,l}} v_{k,l} -\frac{(y_{k+1,l}-y_{k,l}) (\delta_{2} w_{k,l}^2 \delta_{1}-\delta_{3})}{ (\delta_{1} w_{k,l}+y_{k,l}) w_{k,l}}, \label{eq:2D4evk}
\end{gather}
while from \eqref{eq:2D4b} and \eqref{eq:2D4d} we obtain the following equation for $v_{k+1,l+1}$
\begin{gather}
 v_{k+1,l+1} = \frac{\delta_{1} w_{k,l+1}+y_{k+1,l}}{ \delta_{1} w_{k,l+1}+y_{k,l}} v_{k,l+1} -\frac{(y_{k+1,l}-y_{k,l}) (\delta_{2} w_{k,l+1}^2 \delta_{1}-\delta_{3})}{ (\delta_{1} w_{k,l+1}+y_{k,l}) w_{k,l+1}}. \label{eq:2D4evkp}
\end{gather}
Equations \eqref{eq:2D4evk} and \eqref{eq:2D4evkp} give rise to
a compatibility condition between $v_{k+1,l}$ and its shift
in the $l$ direction $v_{k+1,l+1}$ which is given by
\begin{gather*}
 \left[
 \begin{aligned}
y_{{k+1,l+1}}y_{{k,l}} +{\delta_1}(y_{{k,l}} w_{{k,l+1}} +y_{{k+1,l+1}} w_{{k,l+1}}) \\
{} -y_{{k+1,l}}y_{{k,l+1}} -{\delta_1} (y_{{k+1,l}} w_{{k,l+1}} -y_{{k,l+1}} w_{{k,l+1}} )
 \end{aligned}
 \right]
 \big(v_{{k,l+1}}w_{{k,l+1}}+{\delta_3}-\delta_1\delta_2 w_{{k,l+1}}^{2}\big)=0.% \label{eq:2D4vkvkpcc}
\end{gather*}
Discarding the trivial solution
\begin{gather*}
 v_{k,l} = \delta_{1}\delta_{2} w_{k,l} -\frac{\delta_{3}}{w_{k,l}}
\end{gather*}
we obtain the following value for the f\/ield $w_{k,l}$
\begin{gather}
 w_{k,l} = \frac{1}{\delta_{1}} \frac{y_{k+1,l-1} y_{k,l}-y_{k+1,l} y_{k,l-1}}{ y_{k,l-1}+y_{k+1,l}-y_{k,l}-y_{k+1,l-1}}, \label{eq:2D4wkldef}
\end{gather}
which makes \eqref{eq:2D4evk} and \eqref{eq:2D4evkp} compatible. This gives us the f\/irst part of our solution, as displayed in equation~\eqref{eq:2D4wsolfin}. Inserting \eqref{eq:2D4wkldef} into \eqref{eq:2D4evk} we are left to solve the following equation for $v_{k,l}$
\begin{gather*}
 v_{k+1,l}=\frac{y_{k+1,l}-y_{k+1,l-1}}{y_{k,l}-y_{k,l-1}}v_{k,l}
 -\frac{\delta_{1}\delta_{3} (y_{k,l-1}-y_{k+1,l-1})^2 (y_{k,l}-y_{k,l-1})}{ (y_{k+1,l} y_{k,l-1}-y_{k+1,l-1} y_{k,l}) y_{k,l-1}^2} \\
\hphantom{v_{k+1,l}=}{} + \frac{y_{k,l-1}^2 y_{k+1,l} \delta_{2}-\delta_{1}\delta_{3} y_{k+1,l}
 +\delta_{1} \delta_{3} y_{k+1,l-1}- \delta_{2} y_{k+1,l-1} y_{k,l-1}^2}{ (y_{k,l}-y_{k,l-1}) y_{k,l-1}} \\
\hphantom{v_{k+1,l}=}{}-\frac{\delta_{2} y_{k+1,l-1} y_{k,l-1}^2 +\delta_{1} \delta_{3} y_{k+1,l-1}-2\delta_{3} \delta_{1} y_{k,l-1} }{y_{k,l-1}^2}.% \label{eq:2D4evk2}
\end{gather*}
Making the transformation
\begin{gather}
 v_{k,l} = ( y_{k,l}-y_{k,l-1})V_{k,l} +\frac{\delta_{1}\delta_{3}}{y_{k,l}}-\delta_{2}y_{k,l} \label{eq:2D4Vkldef}
\end{gather}
we obtain that $V_{k,l}$ satisf\/ies the dif\/ference equation
\begin{gather}
 V_{k+1,l}=V_{k,l} -\frac{\delta_{1} \delta_{3} (y_{k,l}-y_{k+1,l})^2}{ y_{k,l} y_{k+1,l} (y_{k+1,l-1} y_{k,l}-y_{k+1,l} y_{k,l-1})}. \label{eq:2D4evk3}
\end{gather}

To go further we need to specify the form of the f\/ield $y_{k,l}$. This can be obtained from the Darboux integrability of the $_2D_4$ equation \eqref{eq:2D4}. From~\cite{GSY_DarbouxI} we know that in the case of the $_{2}D_{4}$ equation~\eqref{eq:2D4} we have the following four-point, third-order $W_{1}$ f\/irst integral
\begin{gather*}
 W_{1} ={ \Fppp} \alpha {\frac { \left[
 \begin{gathered}
( u_{{n,m}}-u_{{n+2,m}} -{ \delta_{1}}{ \delta_{2}}u_{{n-1,m}} )u_{n+1,m}^{2} \\
 {} +u_{{n+1,m}}u_{{n+2,m}}u_{{n-1,m}} +{ \delta_{3}}u_{{n-1,m}}
 \end{gathered}\right]
 }{ ( { \delta_{2}}u_{n+1,m}^{2}{ \delta_{1}}- { \delta_{3}}-u_{{n,m}}u_{{n+1,m}} ) u_{{n-1,m}}}} \\
\hphantom{W_{1}=}{} -{ \Fpmm}\alpha {\frac {u_{{n+2,m}}u_{{n-1,m}}+ ( -u_{{n+2,m}}+u_{{n,m}} )
 u_{{n+1,m}}+{ \delta_{3}}}{u_{{n-1,m}}u_{{n,m}}+{ \delta_{3}}}} \\
\hphantom{W_{1}=}{} -{\Fmpm}\beta {\frac { ( u_{{n+1,m}}-u_{{n-1,m}} )
( u_{{n+2,m}}-u_{{n,m}} ) u_{{n,m}}}{u_{{n+2,m}} ( { \delta_{2}}{ \delta_{1}}{u_{{n,m}}}^{2}-u_{{n,m}}u_{{n+1,m}}-
 { \delta_{3}} ) }} \\
\hphantom{W_{1}=}{}+{ \Fmmp}\beta {\frac { ( u_{{n+1,m}}-u_{{n-1,m}} )( u_{{n+2,m}}-u_{{n,m}} ) }{u_{{n+1,m}}u_{{n+2,m}}+{ \delta_{3}}}}.% \label{eq:W12D4}
\end{gather*}
This f\/irst integral implies the relation $W_{1}=\xi_{n}$ which is a third-order, four-point ordinary dif\/ference equation in the $n$ direction depending parametrically on $m$. In particular if we choose the case when $m=2l+1$ we have the equation
\begin{gather}
-{ \Fp{n}} {\frac {u_{{n+2,2l+1}}u_{{n-1,2l+1}}-( u_{{n+2,2l+1}}-u_{{n,2l+1}}) u_{{n+1,2l+1}}+{ \delta_{3}}}{u_{{n-1,2l+1}}u_{{n,2l+1}}+{ \delta_{3}}}}
 \nonumber\\
\qquad{} +{ \Fm{n}} {\frac { ( u_{{n+1,2l+1}}-u_{{n-1,2l+1}} ) ( u_{{n+2,2l+1}}-u_{{n,2l+1}} ) }{ u_{{n+1,2l+1}}u_{{n+2,2l+1}}+{ \delta_{3}}}}=\xi_{n}.
 \label{eq:2D4ePm}
\end{gather}
where we have chosen without loss of generality $\alpha=\beta=1$. Using the transformation \eqref{eq:genautsublp} then equation \eqref{eq:2D4ePm} is converted into the system
\begin{subequations} \label{eq:2D4ePmsys}
 \begin{gather}
 (y_{k+1,l}-y_{k,l}) z_{k,l}-y_{k+1,l} z_{k+1,l}-\delta_{3} =\xi_{2k}(z_{k+1,l} y_{k,l}+\delta_{3}), \label{eq:2D4ePmk} \\
 (y_{k+1,l}-y_{k,l}) (z_{k+1,l}-z_{k,l}) =\xi_{2k+1}(y_{k+1,l} z_{k+1,l}+\delta_{3}). \label{eq:2D4ePmkp}
 \end{gather}
\end{subequations}
If we solve \eqref{eq:2D4ePmkp} with respect to $z_{k+1,l}$ and substitute into \eqref{eq:2D4ePmk} we obtain a \emph{linear}, second-order ordinary dif\/ference equation for $y_{k,l}$
\begin{gather}
 \xi_{2 k-1}y_{k+1,l} -(1+\xi_{2k}-\xi_{2k} \xi_{2 k-1}) y_{k,l} +(1+\xi_{2k}) y_{k-1,l}=0. \label{eq:2D4ePmk2}
\end{gather}
We can f\/ind the solution of equation \eqref{eq:2D4ePmk2} exploiting the arbitrariness of the functions $\xi_{2k}$ and $\xi_{2k+1}$ as in the case of the $_{1}D_{4}$ equation~\eqref{eq:1D4}. Let us introduce the f\/ield $Y_{k,l} = a_{k}y_{k,l}+b_{k}y_{k-1,l}$ and assume that $Y_{k+1,l}-Y_{k,l}$ equals the left-hand side of \eqref{eq:2D4ePmk2}. Then we must have
\begin{gather*}
 \xi_{2 k} = \frac{b_{k+1}-b_{k}+-a_{k}}{a_{k+1}}, \qquad \xi_{2 k-1} = -\frac{b_{k+1}-b_{k}+a_{k+1}-a_{k}}{b_{k}}. %\label{eq:2D4xikkpdef}
\end{gather*}
This implies that $y_{k,l}$ will solve the equation
\begin{gather}
 a_{k}y_{k,l}+b_{k}y_{k-1,l}=\alpha_{l}. \label{eq:2D4ePmk4}
\end{gather}
If we def\/ine
\begin{gather*}
 a_{k}=\frac{1}{c_{k}}\frac{1}{d_{k}-d_{k-1}}, \qquad b_{k}=-\frac{1}{c_{k-1}}\frac{1}{d_{k}-d_{k-1}}, %\label{eq:2D4abkdef}
\end{gather*}
equation \eqref{eq:2D4ePmk4} is solved by
\begin{gather}
 y_{k,l} = c_{k} ( \alpha_{l}d_{k}+\beta_{l}). \label{eq:2D4yklsol}
\end{gather}
This gives the second part of our solution displayed in equation \eqref{eq:2D4ysolfin}.

Inserting the value of $y_{k,l}$ from \eqref{eq:2D4yklsol} into equation \eqref{eq:2D4evk3} we obtain
\begin{gather}
V_{k+1,l}=V_{k,l} +\frac{\delta_{3} \delta_{1} (c_{k+1}-c_{k})^2}{ (d_{k+1}-d_{k}) (\beta_{l} \alpha_{l-1}-\beta_{l-1} \alpha_{l}) c_{k}^2 c_{k+1}^2}\nonumber \\
\hphantom{V_{k+1,l}=}{} -\frac{\alpha_{l} \delta_{1} \delta_{3}}{ \beta_{l} \alpha_{l-1}-\beta_{l-1} \alpha_{l}} \left[ \frac{1}{(\alpha_{l} d_{k+1}+\beta_{l}) c_{k+1}^2}
 -\frac{1}{(\alpha_{l} d_{k}+\beta_{l}) c_{k}^2} \right]. \label{eq:2D4evk4}
\end{gather}
Therefore we can represent the solution to this equation as
\begin{gather}
 V_{k,l} = \gamma_{l} -\frac{\alpha_{l} \delta_{1} \delta_{3}}{ (\alpha_{l} d_{k}+\beta_{l}) c_{k}^2 (\beta_{l} \alpha_{l-1}-\beta_{l-1} \alpha_{l})}
 +\frac{\delta_{3} \delta_{1} e_{k}}{ \beta_{l} \alpha_{l-1}-\beta_{l-1} \alpha_{l}}, \label{eq:2D4Vklsol}
\end{gather}
where $e_{k}$ is given by the discrete integration
\begin{gather*}
 e_{k+1}-e_{k} =\frac{(c_{k+1}-c_{k})^2}{ (d_{k+1}-d_{k}) c_{k}^2 c_{k+1}^2}, %\label{eq:2D4ekdef}
\end{gather*}
which is just equation \eqref{eq:2D4ekeqdef0}. Inserting the value of $V_{k,l}$ from equation \eqref{eq:2D4Vklsol} and the value of $y_{k,l}$ from equation \eqref{eq:2D4yklsol} into equation \eqref{eq:2D4Vkldef} we obtain equation \eqref{eq:2D4vsolfin} which is the third part of our solution. Using this value alongside with the value of $w_{k,l}$ from equation \eqref{eq:1D4wkldef} we can compute $z_{k,l}$ from the original system \eqref{eq:2D4sys}. We obtain a \emph{single} compatibility condition given by
\begin{gather*}
 (\beta_{l} \alpha_{l+1}-\beta_{l+1} \alpha_{l}) \gamma_{l+1} -(\beta_{l-1} \alpha_{l}-\beta_{l} \alpha_{l-1}) \gamma_{l} =(\beta_{l} \alpha_{l+1}-\beta_{l+1} \alpha_{l}) \delta_{2}, %\label{eq:2D4gammleq}
\end{gather*}
which is just equation \eqref{eq:2D4gammleq0}. This concludes the procedure of solution of the $_{2}D_{4}$ equation~\eqref{eq:2D4}.

\subsection[${}_tH_3^\varepsilon$ equation (\ref{eq:tH3e})]{$\boldsymbol{{}_tH_3^\varepsilon}$ equation (\ref{eq:tH3e})}

\begin{subequations} \label{eq:tH3sys}
To solve the \tHeq{3} equation \eqref{eq:tH3e} we start from the equation itself. We apply the transforma\-tion~\eqref{eq:h4aut} to write down the \tHeq{3} equation \eqref{eq:tH3e} as the following system of two coupled equations
 \begin{gather}
 \alpha_{2} (p_{n,l} q_{n+1,l}+p_{n+1,l} q_{n,l})-p_{n,l} q_{n,l}-p_{n+1,l} q_{n+1,l} \nonumber\\
\qquad {}-\alpha_{3} \big(\alpha_{2}^2-1\big)
 \left(\delta^2+\varepsilon^2 \frac{q_{n,l} q_{n+1,l}}{\alpha_{3}^2 \alpha_{2}}\right) =0, \label{eq:tH3l} \\
 \alpha_{2} (q_{n,l} p_{n+1,l+1}+q_{n+1,l} p_{n,l+1})-q_{n,l} p_{n,l+1}-q_{n+1,l} p_{n+1,l+1}\nonumber\\
 \qquad{} -\alpha_{3} \big(\alpha_{2}^2-1\big)
 \left(\delta^2+\varepsilon^2 \frac{q_{n,l} q_{n+1,l}}{\alpha_{3}^2 \alpha_{2}}\right) =0. \label{eq:tH3lp}
 \end{gather}
\end{subequations}
As in the case of the \tHeq{2} equation \eqref{eq:tH2e} we have that equation \eqref{eq:tH3l} depends on~$p_{n,l}$ and~$p_{n+1,l}$ and that equation~\eqref{eq:tH3lp} depends on $p_{n,l+1}$ and~$p_{n+1,l+1}$. So we can apply the translation operator $T_{l}$ to~\eqref{eq:tH3l} to obtain two equations in terms of $p_{n,l+1}$ and $p_{n+1,l+1}$
\begin{subequations} \label{eq:tH3sys2}
 \begin{gather}
 \alpha_{2} (p_{n,l+1} q_{n+1,l+1}+p_{n+1,l+1} q_{n,l+1})-p_{n,l+1} q_{n,l+1}-p_{n+1,l+1} q_{n+1,l+1}\nonumber \\
\qquad{} -\alpha_{3}\big(\alpha_{2}^2-1\big) \left(\delta^2+\varepsilon^2 \frac{q_{n,l+1} q_{n+1,l+1}}{\alpha_{3}^2 \alpha_{2}}\right) =0, \label{eq:tH3l2} \\
 \alpha_{2} (q_{n,l} p_{n+1,l+1}+q_{n+1,l} p_{n,l+1})-q_{n,l} p_{n,l+1}-q_{n+1,l} p_{n+1,l+1}\nonumber \\
\qquad{} -\alpha_{3} \big(\alpha_{2}^2-1\big) \left(\delta^2+\varepsilon^2 \frac{q_{n,l} q_{n+1,l}}{\alpha_{3}^2 \alpha_{2}}\right) =0. \label{eq:tH3lp2}
 \end{gather}
\end{subequations}
The system \eqref{eq:tH3sys2} is equivalent to the original system \eqref{eq:tH3sys}. Then since we can assume $\alpha_{2},\alpha_{3}\neq0$\footnote{If $\alpha_{2}=0$ or $\alpha_{3}=0$ in \eqref{eq:tH3sys} we have that the system becomes trivially equivalent to $q_{n,l}=0$ and $p_{n,l}$ is left unspecif\/ied. Therefore we can discard this trivial case.}
we can solve \eqref{eq:tH3sys2} with respect to $p_{n,l+1}$ and $p_{n+1,l+1}$:
\begin{subequations} \label{eq:tH3pnlppnplp}
 \begin{gather}
 p_{{n,l+1}} ={\frac {
 \left[
 \begin{gathered}
 { \alpha_2}
( q_{{n+1,l+1}}-q_{{n+1,l}})
 \big( {\epsilon}^{2}q_{{n,l}}q_{{n,l+1}}+{\delta}^{2}{{\alpha_3}}^{2} \big)
 \\
 {} +{\delta}^{2} \alpha_2^{2}\alpha_3^{2}
 ( q_{{n,l}}-q_{{n,l+1}} )
 +{\epsilon}^{2}q_{{n+1,l+1}}q_{{n+1,l}}
 ( q_{{n,l}}-q_{{n,l+1}} )
 \end{gathered}\right]}{ ( q_{{n+1,l+1}}q_{{n,l}}-q_{{n+1,l}}q_{{n,l+1}} ) {\alpha_3} { \alpha_2}}}, \label{eq:tH3pnlp} \\
 p_{{n+1,l+1}}={\frac {\left[
 \begin{gathered}
 { \alpha_2} \big( {\epsilon}^{2}q_{{n+1,l+1}}q_{{n+1,l}}+{\delta}^{2}{{ \alpha_3}}^{2} \big) ( q_{{n,l}}-q_{{n,l+1}} ) \\
 {} +{\delta}^{2} \alpha_2^{2}\alpha_3^{2} ( q_{{n+1,l+1}}-q_{{n+1,l}} )
 +{\epsilon}^{2}q_{{n,l}}q_{{n,l+1}} ( q_{{n+1,l+1}}-q_{{n+1,l}} )
\end{gathered}\right]
 }{ ( q_{{n+1,l+1}}q_{{n,l}}-q_{{
 n+1,l}}q_{{n,l+1}} ) { \alpha_3} { \alpha_2}}}. \label{eq:tH3pnplp}
 \end{gather}
\end{subequations}
We see that the right-hand sides of \eqref{eq:tH3pnlppnplp} are functions only of $q_{n,l}$, $q_{n+1,l}$, $q_{n,l+1}$ and $q_{n+1,l+1}$ and are well def\/ined as long as $q_{n,l}$ is not a solution of equation \eqref{eq:tH3qsing}, which is therefore a~\emph{singular case}. Therefore at this point the procedure of solution bifurcates into two cases. We treat them separately.

{\bf Singular case: $\boldsymbol{q_{n,l}}$ solve (\ref{eq:tH3qsing}).} If $q_{n,l}$ solves equation \eqref{eq:tH3qsing} we f\/irst solve this equation with respect to $q_{n,l}$ and then use the system \eqref{eq:tH3sys} to specify $p_{n,l}$. Indeed equation \eqref{eq:tH3qsing} is a~trivial Darboux integrable equation, since it possesses the following two-point, f\/irst-order f\/irst integrals
\begin{subequations} \label{eq:tH3qsingW1W2}
 \begin{gather}
 W_{1}= \frac{q_{n+1,l}}{q_{n,l}}, \label{eq:tH3qsingW1} \\
 W_{2} = \frac{q_{n,l+1}}{q_{n,l}}. \label{eq:tH3qsingW2}
 \end{gather}
\end{subequations}
As remarked in the introduction the existence of a two-point, f\/irst-order f\/irst integral means that the equation is itself a f\/irst integral. Therefore the equation \eqref{eq:tH3qsing} can be alternatively written as $\left( T_{l}-\Id \right)W_{1}$ or $\left( T_{n}-\Id \right)W_{2}$ with $W_{1}$ and $W_{2}$ given by \eqref{eq:tH3qsingW1W2}. From \eqref{eq:tH3qsingW1} we obtain
\begin{gather}
 q_{n+1,l}=\xi_{n}q_{n,l} \label{eq:tH2singsol0}
\end{gather}
being $\xi_{n}$ an arbitrary function of its argument. Equation \eqref{eq:tH2singsol0} is casted into total dif\/ference form by def\/ining $\xi_{n}= a_{n+1}/a_{n}$, with $\alpha_{n}$ a new arbitrary function of its argument. Then we obtain that the general solution of equation~\eqref{eq:tH3qsing} is
\begin{gather}
 q_{n,l} = a_{n}\zeta_{l}, \label{eq:tH3qsingsol}
\end{gather}
where $\zeta_{l}$ is an arbitrary function of its argument.

Let us remark that equation \eqref{eq:tH3qsing} is the \emph{logarithmic discrete wave equation}, since it can be mapped into the discrete wave equation~\eqref{eq:tH2qsing} exponentiating~\eqref{eq:tH2qsing} and then taking $q_{n,l}\to e^{q_{n,l}}$, and it is a discretization of the hyperbolic
partial dif\/ferential equation
\begin{gather*}
 u u_{xt} - u_{x} u_{t} =0, %\label{eq:tH3qsingcont}
\end{gather*}
which is obtained from the wave equation $v_{xt}=0$ using the transformation $v=\log u$. This fact is worth to note since being the transformation connecting \eqref{eq:tH2qsing} and \eqref{eq:tH3qsing} not bi-rational, integrability properties, in this case linearization and Darboux integrability, are not \emph{a priori} preserved~\cite{Grammaticos2005}.

Substituting \eqref{eq:tH3qsingsol} into \eqref{eq:tH3sys2} we obtain the compatibility condition
\begin{gather}
 \zeta_{l+1} - \zeta_{l} =0, \label{eq:tH3singcc}
\end{gather}
i.e., $\zeta_{l}=\zeta_{0}=\text{const}$ and the system \eqref{eq:tH3sys} is now consistent. With this we f\/ind the f\/irst piece of the general solution in this case given by \eqref{eq:tH3qnlsingdef}. Therefore we are left with one equation for~$p_{n,l}$, e.g.,~\eqref{eq:tH3l}. Therefore inserting \eqref{eq:tH3qsingsol} with $\zeta_{l}=\zeta_{0}$ in \eqref{eq:tH3l} and solving with respect to $p_{n+1,l}$ we obtain
\begin{gather*}
 p_{n+1, l} - \frac{\alpha_2 a_{n+1}-a_{n}}{a_{n+1}-\alpha_2 a_{n}} p_{n, l} =\big(\alpha_2^{2}-1\big)
 \frac{\delta^2 \alpha_3^2 \alpha_2+\epsilon^2 a_{n} \zeta_{0}^2 a_{n+1}}{ \alpha_3 \alpha_2 \zeta_{0} (\alpha_2 a_{n}-a_{n+1})}. %\label{eq:tH3psingeq}
\end{gather*}
Def\/ining
\begin{gather}
 \frac{\alpha_2 a_{n+1}-a_{n}}{a_{n+1}-\alpha_2 a_{n}} = \frac{b_{n+1}}{b_{n}}, \label{eq:tH3bsingdef}
\end{gather}
we have that $p_{n,l}$ solves the equation
\begin{gather}
 \frac{p_{n+1, l}}{b_{n+1}} - \frac{p_{n, l}}{b_{n}} ={\frac {{\delta}^{2}\alpha_3^{2}\alpha_2^{2}b_{{n}}-b_{{n+1}}
\big( {\delta}^{2}\alpha_3^{2}+{\epsilon}^{2}a_{{n}}^{2}\zeta_{{0}}^{2} \big)
{ \alpha_2}+{\epsilon}^{2}a_{{n}}^{2}\zeta_{{0}}^{2}b_{{n}}}{ b_{{n}}a_{{n}}\zeta_{{0}}{ \alpha_2} { \alpha_3} b_{{n+1}}}}. \label{eq:tH3psingeq2}
\end{gather}
Note that $b_{n}$ in \eqref{eq:tH3bsingdef} is given in terms of $a_{n}$ and $a_{n+1}$ through discrete integration and it is the constraint given in equation~\eqref{eq:tH3aeq0}. Equation \eqref{eq:tH3psingeq2} is solved by
\begin{gather*}
 p_{n,l} = b_{n} ( \beta_{l}+c_{n}), %\label{eq:tH3psingsol}
\end{gather*}
where $c_{n}$ is given by the discrete integration
\begin{gather*}
 c_{n+1} = c_{n} +{\frac {{\delta}^{2}\alpha_3^{2}\alpha_2^{2}b_{{n}}-b_{{n+1}}\big( {\delta}^{2}\alpha_3^{2}+{\epsilon}^{2}a_{{n}}^{2}\zeta_{{0}}^{2} \big){ \alpha_2}+{\epsilon}^{2}a_{{n}}^{2}\zeta_{{0}}^{2}b_{{n}}}{ b_{{n}}a_{{n}}\zeta_{{0}}{ \alpha_2} { \alpha_3} b_{{n+1}}}}, %\label{eq:tH3csingdef}
\end{gather*}
i.e., as in equation \eqref{eq:tH3ceq0}. This yields the solution of the \tHeq{3} equation \eqref{eq:tH3e} when $q_{n,l}$ satisfy equation~\eqref{eq:tH3qsing}.

{\bf General case: $\boldsymbol{q_{n,l}}$ do not solve \eqref{eq:tH3qsing}.} If the f\/ield $q_{n,l}$ do not solve \eqref{eq:tH3qsing}
we have that we can def\/ine $p_{n,l+1}$ and $p_{n+1,l+1}$ as in \eqref{eq:tH3pnlp} and \eqref{eq:tH3pnplp} respectively. Furthermore these two equations must be compatible. The compatibility condition is obtained applying $T_{l}^{-1}$ to~\eqref{eq:tH3pnplp} and imposing to the obtained expression to be equal to~\eqref{eq:tH3pnlp}. We then f\/ind that $q_{n,l}$ must solve the following equation
\begin{gather}
 \delta^{2}\alpha_2^{2}\alpha_3^{2} \big[ q_{{n+1,l+1}}q_{{n,l}}-q_{{n,l}}q_{{n-1,l+1}} +q_{{n,l+1}} ( q_{{n-1,l}}-q_{{n+1,l}} ) \big]\nonumber \\
\qquad{} - \alpha_2\big( {\epsilon}^{2}q_{{n,l}}q_{{n,l+1}}+\delta^{2}\alpha_3^{2} \big) ( q_{{n+1,l+1}}q_{{n-1,l}}-q_{{n-1,l+1}}q_{{n+1,l}} )\nonumber \\
\qquad{} +{\epsilon}^{2} \big[ q_{{n,l}} q_{{n+1,l+1}} ( q_{{n-1,l}}-q_{{n+1,l}} ) -q_{{n,l+1}}q_{{n-1,l}}q_{{n+1,l}} \big] q_{{n-1,l+1}}\nonumber \\
\qquad{} +{\epsilon}^{2}q_{{n+1,l+1}}q_{{n+1,l}}q_{{n,l+1}}q_{{n-1,l}}=0. \label{eq:tH3qnleq}
\end{gather}
As in the case of the \tHeq{2} equation \eqref{eq:tH2e} the partial dif\/ference equation for $q_{n,l}$ is not def\/ined on the square quad graph of Fig.~\ref{fig:geomquad}, but it is def\/ined on the six-point lattice shown in Fig.~\ref{fig:6pointslattice}. Moreover once equation \eqref{eq:tH3qnleq} is solved we can use indif\/ferently \eqref{eq:tH3pnlp} or \eqref{eq:tH3pnplp} to obtain the value of the f\/ield $p_{n,l}$ since these two merely def\/ines $p_{n,l+1}$ in terms of $q_{n,l}$ and its shifts. Therefore if we f\/ind the general solution of~\eqref{eq:tH3qnleq} the value of $p_{n,l}$ will follow. E.g., if we solve~\eqref{eq:tH3qnleq} applying $T_{l}^{-1}$ to \eqref{eq:tH3pnlp} we will obtain \eqref{eq:tH3pnlsoldefgen} which is then the f\/irst part of the general solution. To f\/ind the solution of equation~\eqref{eq:tH3qnleq} we turn to the f\/irst integrals. Like in the case of the \tHeq{2} equation~\eqref{eq:tH2e} we will f\/ind an expression for $q_{n,l}$ using the f\/irst integrals, and then we will insert it into~\eqref{eq:tH3qnleq} to reduce the number of arbitrary functions to the right one. From~\cite{GSY_DarbouxI} we know that the \tHeq{3} equation \eqref{eq:tH3e} possesses a four-point, third-order integral in the $n$ direction
\begin{gather}
 W_{1} =\Fp{m} \frac {(u_{{n+1,m}}- u_{{n-1,m}})( u_{{n+2,m}}-u_{{n,m}} ) }{ \alpha_{2}^{4}{\epsilon}^{2}{\delta}^{2}
 - \alpha_{2}^{3}u_{{n+1,m}}u_{{n,m}} +\alpha_{2}^{2} \big( u_{{n,m}}^{2}+u_{n+1,m}^{2} -2{\epsilon}^{2}{\delta}^{2}\big)
 -{ \alpha_{2}}u_{{n,m}}u_{{n+1,m}}+{\epsilon}^{2}{\delta}^{2}} \nonumber\\
\hphantom{W_{1} =}{} -\Fm{m} {\frac { (u_{{n+1,m}} -u_{{n-1,m}})( u_{{n+2,m}}-u_{{n,m}} ) }{{
 \alpha_{2}} ( -u_{{n-1,m}}+{ \alpha_{2}}u_{{n,m}} ) ( -u_{ {n+2,m}}+u_{{n+1,m}}{ \alpha_{2}}) }}. \label{eq:W1tH3e}
\end{gather}
We consider the equation $W_{1}=\xi_{n}/\alpha_{2}$\footnote{The extra $\alpha_{2}$ is due to the arbitrariness of $\xi_{n}$ and is inserted in order to simplify the formulas.}, where $W_{1}$ is given by \eqref{eq:W1tH3e}, with $m=2l+1$
\begin{gather*}
 \frac {(u_{n+1,2l+1} -u_{n-1,2l+1} ) ( u_{n+2,2l+1}-u_{n,2l+1} ) }{ ( \alpha_{2}u_{n,2l+1}-u_{n-1,2l+1} ) ( u_{n+2,2l+1}-\alpha_{2}u_{n+1,2l+1} ) }
 =\xi_{n}.% \label{eq:tH3kp}
\end{gather*}
Using the substitutions \eqref{eq:h4aut} we have
\begin{gather}
 \frac {(q_{n+1,l} -q_{n-1,l} ) ( q_{n+2,l}-q_{n,l} ) }{( \alpha_{2}q_{n,l}-q_{n-1,l} )( q_{n+2,l}-\alpha_{2}q_{n+1,l}) } =\xi_{n}. \label{eq:tH3kp2}
\end{gather}
This equation contains only $q_{n,l}$ and its shifts. By the transformation
\begin{gather}
 Q_{n,l} = \frac{\alpha_2 q_{n, l}-q_{n-1, l}}{ q_{n+1, l}-q_{n-1, l}} \label{eq:tH3Qdef}
\end{gather}
equation \eqref{eq:tH3kp2} becomes
\begin{gather}
 Q_{n+1, l} +\frac{1}{\xi_{n}Q_{n, l}}=1, \label{eq:tH3kp3}
\end{gather}
which is the same discrete Riccati equation as in \eqref{eq:tH2kp3}. This means that the solution of \eqref{eq:tH3kp3} is given by \eqref{eq:tH2Qsol} with the appropriate def\/initions (\ref{eq:tH2lam}), (\ref{eq:tH2adef}), (\ref{eq:tH2bddef}). We can substitute into~\eqref{eq:tH3Qdef} the solution~\eqref{eq:tH2Qsol}
\begin{gather*}
 \frac{q_{n+1, l}-q_{n-1, l}}{ \alpha_2 q_{n, l}-q_{n-1, l}} = \frac{(c_{n}+\zeta_{l}) (c_{n+1}-c_{n-1})}{ (c_{n+1}+\zeta_{l})(c_{n}-c_{n-1})}% \label{eq:tH3kp4}
\end{gather*}
and we obtain an equation for $q_{n,l}$. Introducing
\begin{gather}
 P_{n,l} = ( c_{n}+\zeta_{l})q_{n,l} \label{eq:tH3Pdef}
\end{gather}
we obtain that $P_{n,l}$ solves the equation
\begin{gather}
 P_{n+1,l} -\alpha_{2}\frac{c_{n+1}-c_{n-1}}{c_{n}-c_{n-1}} P_{n,l} +\frac{c_{n+1}-c_{n}}{c_{n}-c_{n-1}}P_{n-1,l} =0. \label{eq:tH3Peq}
\end{gather}
Using the transformation
\begin{gather}
 P_{n,l} = \frac{R_{n,l}}{R_{n-1,l}}, \label{eq:tH3Rdef}
\end{gather}
we can cast equation \eqref{eq:tH3Peq} in discrete Riccati equation form
\begin{gather}
 R_{n+1, l}+ \frac{c_{n+1}-c_{n}}{c_{n}-c_{n-1}}\frac{1}{R_{n, l}} = \alpha_2\frac{c_{n+1}-c_{n-1}}{c_{n}-c_{n-1}}. \label{eq:tH3Req}
\end{gather}
Let $d_{n}$ be a particular solution of equation \eqref{eq:tH3Req}
\begin{gather}
 d_{n+1}+ \frac{c_{n+1}-c_{n}}{c_{n}-c_{n-1}}\frac{1}{d_{n}} = \alpha_2\frac{c_{n+1}-c_{n-1}}{c_{n}-c_{n-1}}. \label{eq:tH3deq}
\end{gather}
Assuming $d_{n}$ as the new arbitrary function we can express $c_{n}$ as the result of \emph{two} discrete integrations. Indeed introducing $z_{n}=c_{n}-c_{n-1}$ in~\eqref{eq:tH3deq} we have
\begin{gather}
 \frac{z_{n+1}}{z_{n}} = \frac{(d_{n+1}-\alpha_2)d_{n}}{\alpha_2 d_{n}-1}. \label{eq:tH3zeq}
\end{gather}
Equation \eqref{eq:tH3zeq} represents the f\/irst discrete integration, whereas the second one is given by the def\/inition
\begin{gather}
 c_{n}-c_{n-1}=z_{n}. \label{eq:tH3cneq}
\end{gather}
Now we can linearize the discrete Riccati equation \eqref{eq:tH3Req} by the transformation
\begin{gather}
 R_{n,l}= d_{n} + \frac{1}{S_{n,l}} \label{eq:tH3Sdef}
\end{gather}
and we get the following \emph{linear} equation for $S_{n,l}$
\begin{gather}
 S_{n+1, l} -\frac{d_{n}^2(c_{n}-c_{n-1})}{c_{n+1}-c_{n}}S_{n, l}= \frac{d_{n}(c_{n}-c_{n-1})}{c_{n+1}-c_{n}}. \label{eq:tH3Seq}
\end{gather}
Def\/ining
\begin{subequations} \label{eq:tH3dfdef}
 \begin{gather}
 d_{n} =\frac{e_{n}}{e_{n-1}}, \label{eq:tH3ddef} \\
 f_{n} - f_{n-1} =\frac{c_{n}-c_{n-1}}{e_{n}e_{n-1}}, \label{eq:tH3fdef}
 \end{gather}
\end{subequations}
the solution of \eqref{eq:tH3Seq} is
\begin{gather}
 S_{n, l} = \frac{(f_{n-1}+\beta_{l}) e_{n-1}^2}{c_{n}-c_{n-1}}. \label{eq:tH3Ssol}
\end{gather}
Here we have the f\/inal form of the constraint $f_{n}$, which is the same as the one given in \eqref{eq:tH3feq0}. Inserting \eqref{eq:tH3Ssol} and \eqref{eq:tH3dfdef} into \eqref{eq:tH3Sdef} we obtain
\begin{gather}
 R_{n, l} = \frac{e_{n} (f_{n}+\beta_{l})}{e_{n-1}(f_{n-1}+\beta_{l})}. \label{eq:tH3Rsol}
\end{gather}
Inserting the def\/inition of $R_{n,l}$ \eqref{eq:tH3Rdef} into \eqref{eq:tH3Rsol} we obtain
\begin{gather*}
 \frac{P_{n,l}}{e_{n} (f_{n}+\beta_{l})} = \frac{P_{n-1,l}}{e_{n-1}(f_{n-1}+\beta_{l})}, %\label{eq:tH3Peq2}
\end{gather*}
i.e.,
\begin{gather}
 P_{n,l} = \gamma_{l}e_{n} (f_{n}+\beta_{l}). \label{eq:tH3Psol}
\end{gather}
Introducing \eqref{eq:tH3Psol} into \eqref{eq:tH3Pdef} we obtain
\begin{gather}
 q_{n,l}= \frac{\gamma_{l}e_{n} (f_{n}+\beta_{l})}{c_{n}+\zeta_{l}}, \label{eq:tH3qsol}
\end{gather}
where $f_{n}$ is def\/ined by \eqref{eq:tH3fdef}, and $c_{n}$ is given by \eqref{eq:tH3zeq} and~\eqref{eq:tH3cneq}, i.e., $c_{n}$ is the solution of the equation
\begin{gather}
 \frac {c_{{n+1}}-c_{{n}}}{c_{{n}}-c_{{n-1}}} = {\frac {e_{{n+1}}-{\alpha_2}e_{{n}}}{{\alpha_2}e_{{n}}-e_{{n-1}}}}, \label{eq:tH3cneq2}
\end{gather}
and $e_{n}$ is an arbitrary function, i.e., we have that $c_{n}$ must solve equation \eqref{eq:tH3eeq0}.

Formally equation \eqref{eq:tH3qsol} has the form of the solution presented in formula \eqref{eq:tH3qnlsoldefgen}, but it depends on three arbitrary functions in the $l$ direction, namely $\zeta_{l}$, $\beta_{l}$ and $\gamma_{l}$. Therefore there must be a constraint between these functions. This constraint can be obtained by plugging~\eqref{eq:tH3qsol} into~\eqref{eq:tH3qnleq}, but here we have another bifurcation depending on the value of parameter~$\delta$. Indeed it is easy to see that we must distinguish the cases when $\delta\neq0$ and when $\delta=0$.

{\bf Case $\boldsymbol{\delta\neq0}$.} Inserting \eqref{eq:tH3qsol} into \eqref{eq:tH3qnleq} if $\delta\neq0$ factorizing the $n$ dependent part away we are left with
\begin{gather*}
 \zeta_{l+1} - \zeta_{l}= \frac{\epsilon^2}{\alpha_2 \delta^2 \alpha_3^2} \gamma_{l+1}\gamma_{l} (\beta_{l+1}-\beta_{l}).% \label{eq:tH3aleq}
\end{gather*}
This equation tells us that the function $\zeta_{l}$ can be expressed after a discrete integration in terms of the two arbitrary functions $\beta_{l}$ and $\gamma_{l}$. This condition is just \eqref{eq:tH3gameq0}. This yields the general solution of the \tHeq{3} equation \eqref{eq:tH3e} when $\delta\neq0$ and the f\/ield $q_{n,l}$ do not satisfy equation \eqref{eq:tH3qsing}.

{\bf Case $\boldsymbol{\delta=0}$.} Inserting \eqref{eq:tH3qsol} into \eqref{eq:tH3qnleq} if $\delta\neq0$ factorizing the $n$ dependent part the compatibility condition is just $\beta_{l+1}-\beta_{l}=0$, i.e., $\beta_{l}=\beta_{0}=\text{const}$. It is easy to check that the obtained value of $q_{n,l}$ through formula \eqref{eq:tH3qsol} is consistent with the substitution of $\delta=0$ in~\eqref{eq:tH3sys} and therefore that in the case $\delta=0$ the value of $q_{n,l}$ is given by
\begin{gather*}
 q_{n,l}= \frac{\gamma_{l}e_{n} (f_{n}+\beta_{0})}{c_{n}+\zeta_{l}}, %\label{eq:tH3qsoldel0}
\end{gather*}
where the functions $c_{n}$ and $f_{n}$ are def\/ined implicitly and can be found by discrete integration from \eqref{eq:tH3fdef} and~\eqref{eq:tH3cneq2} respectively. This is just equation~\eqref{eq:tH3qnlsoldef0}. Since equation~\eqref{eq:tH3pnlsoldefgen} is not singular if $\delta=0$ we obtain that in this case the general solution for the f\/ield $p_{n,l}$ is given by substituting $\delta=0$ into~\eqref{eq:tH3pnlsoldefgen}, i.e., just by equation~\eqref{eq:tH3pnlsoldef0} where $q_{n,l}$ is simply given by~\eqref{eq:tH3qnlsoldef0}. This yields the general solution of the \tHeq{3}
equation~\eqref{eq:tH3e} in the case when $\delta=0$.

\subsection*{Acknowledgements}

We would like to thank Professor Decio Levi for the many interesting and fruitful discussion during the preparation of this paper. We are also grateful to the anonymous referees whose comments greatly helped us in improving the paper.

GG has been supported by INFN IS-CSN4 \emph{Mathematical Methods of Nonlinear Physics} and by the Australian Research Council through an Australian Laureate Fellowship grant FL120100094.

\pdfbookmark[1]{References}{ref}
\LastPageEnding


\begin{thebibliography}{99}
\footnotesize\itemsep=0pt

\bibitem{ABS2003}
Adler V.E., Bobenko A.I., Suris Yu.B., Classif\/ication of integrable equations on
 quad-graphs. {T}he consistency approach, \href{https://doi.org/10.1007/s00220-002-0762-8}{\textit{Comm. Math. Phys.}}
 \textbf{233} (2003), 513--543, \href{https://arxiv.org/abs/nlin.SI/0202024}{nlin.SI/0202024}.

\bibitem{ABS2009}
Adler V.E., Bobenko A.I., Suris Yu.B., Discrete nonlinear hyperbolic equations:
 classif\/ication of integrable cases, \href{https://doi.org/10.1007/s10688-009-0002-5}{\textit{Funct. Anal. Appl.}} \textbf{43}
 (2009), 3--17, \href{https://arxiv.org/abs/0705.1663}{arXiv:0705.1663}.

\bibitem{AdlerStartsev1999}
Adler V.E., Startsev S.Ya., On discrete analogues of the {L}iouville equation,
 \href{https://doi.org/10.1007/BF02557219}{\textit{Theoret. and Math. Phys.}} \textbf{121} (1999), 1484--1495,
 \href{https://arxiv.org/abs/solv-int/9902016}{solv-int/9902016}.

\bibitem{BellonViallet1999}
Bellon M.P., Viallet C.M., Algebraic entropy, \href{https://doi.org/10.1007/s002200050652}{\textit{Comm. Math. Phys.}}
 \textbf{204} (1999), 425--437, \href{https://arxiv.org/abs/chao-dyn/9805006}{chao-dyn/9805006}.

\bibitem{BobenkoSuris2002}
Bobenko A.I., Suris Yu.B., Integrable systems on quad-graphs, \href{https://doi.org/10.1155/S1073792802110075}{\textit{Int. Math.
 Res. Not.}} \textbf{2002} (2002), 573--611, \href{https://arxiv.org/abs/nlin.SI/0110004}{nlin.SI/0110004}.

\bibitem{Bobenko2008book}
Bobenko A.I., Suris Yu.B., Discrete dif\/ferential geometry. Integrable structure,
 \href{https://doi.org/10.1007/978-3-7643-8621-4}{\textit{Graduate Studies in Mathematics}}, Vol.~98, Amer. Math. Soc.,
 Providence, RI, 2008.

\bibitem{Boll2011}
Boll R., Classif\/ication of 3{D} consistent quad-equations, \href{https://doi.org/10.1142/S1402925111001647}{\textit{J.~Nonlinear
 Math. Phys.}} \textbf{18} (2011), 337--365, \href{https://arxiv.org/abs/1009.4007}{arXiv:1009.4007}.

\bibitem{Boll2012a}
Boll R., Corrigendum: {C}lassif\/ication of 3{D} consistent quad-equations,
 \textit{J.~Nonlinear Math. Phys.} \textbf{19} (2012), 1292001, 3~pages.

\bibitem{Boll2012b}
Boll R., Classif\/ication and {L}agrangian structure of {3D} consistent
 quad-equations, Ph.D.~Thesis, Technische Universit\"at Berlin, 2012.


\bibitem{Bridgman2013}
Bridgman T., Hereman W., Quispel G.R.W., van~der Kamp P.H., Symbolic
 computation of {L}ax pairs of partial dif\/ference equations using consistency
 around the cube, \href{https://doi.org/10.1007/s10208-012-9133-9}{\textit{Found. Comput. Math.}} \textbf{13} (2013), 517--544,
 \href{https://arxiv.org/abs/1308.5473}{arXiv:1308.5473}.

\bibitem{HayButler2013}
Butler S., Hay M., Simple identif\/ication of fake {L}ax pairs,
 \href{https://arxiv.org/abs/1311.2406}{arXiv:1311.2406}.

\bibitem{HayButler2015}
Butler S., Hay M., Two def\/initions of fake {L}ax pairs, \href{https://doi.org/10.1063/1.4912469}{\textit{AIP Conf.
 Proc.}} \textbf{1648} (2015), 180006, 5~pages.

\bibitem{CalogeroDeGasperisIST_I}
Calogero F., Degasperis A., Spectral transform and solitons. {V}ol.~{I}. Tools
 to solve and investigate nonlinear evolution equations, \textit{Studies in
 Mathematics and its Applications}, Vol.~13, North-Holland Publishing Co.,
 Amsterdam~-- New York, 1982.

\bibitem{CalogeroNucci1991}
Calogero F., Nucci M.C., Lax pairs galore, \href{https://doi.org/10.1063/1.529096}{\textit{J.~Math. Phys.}} \textbf{32}
 (1991), 72--74.

\bibitem{DoliwaSantini1997}
Doliwa A., Santini P.M., Multidimensional quadrilateral lattices are
 integrable, \href{https://doi.org/10.1016/S0375-9601(97)00456-8}{\textit{Phys. Lett.~A}} \textbf{233} (1997), 365--372,
 \href{https://arxiv.org/abs/solv-int/9612007}{solv-int/9612007}.

\bibitem{FalquiViallet1993}
Falqui G., Viallet C.M., Singularity, complexity, and quasi-integrability of
 rational mappings, \href{https://doi.org/10.1007/BF02096835}{\textit{Comm. Math. Phys.}} \textbf{154} (1993), 111--125,
 \href{https://arxiv.org/abs/hep-th/9212105}{hep-th/9212105}.

\bibitem{GarifullinYamilov2012}
Garifullin R.N., Yamilov R.I., Generalized symmetry classif\/ication of discrete
 equations of a class depending on twelve parameters, \href{https://doi.org/10.1088/1751-8113/45/34/345205}{\textit{J.~Phys.~A:
 Math. Theor.}} \textbf{45} (2012), 345205, 23~pages, \href{https://arxiv.org/abs/1203.4369}{arXiv:1203.4369}.

\bibitem{GarifullinYamilov2015}
Garifullin R.N., Yamilov R.I., Integrable discrete nonautonomous quad-equations
 as {B}\"acklund auto-transformations for known {V}olterra and {T}oda type
 semidiscrete equations, \href{https://doi.org/10.1088/1742-6596/621/1/012005}{\textit{J.~Phys. Conf. Ser.}} \textbf{621} (2015),
 012005, 18~pages, \href{https://arxiv.org/abs/1405.1835}{arXiv:1405.1835}.

\bibitem{Grammaticos2005}
Grammaticos B., Ramani A., Viallet C.M., Solvable chaos, \href{https://doi.org/10.1016/j.physleta.2005.01.026}{\textit{Phys. Lett.~A}}
 \textbf{336} (2005), 152--158, \href{https://arxiv.org/abs/math-ph/0409081}{math-ph/0409081}.

\bibitem{GubScimSIDE12}
Gubbiotti G., Scimiterna C., Reconstructing a lattice equation: a
 non-autonomous approach to the {H}ietarinta equation, \href{https://doi.org/10.3842/SIGMA.2018.004}{\textit{SIGMA}}
 \textbf{14} (2018), 004, 21~pages, \href{https://arxiv.org/abs/1705.00298}{arXiv:1705.00298}.

\bibitem{GSL_general}
Gubbiotti G., Scimiterna C., Levi D., Algebraic entropy, symmetries and
 linearization of quad equations consistent on the cube, \href{https://doi.org/10.1080/14029251.2016.1237200}{\textit{J.~Nonlinear
 Math. Phys.}} \textbf{23} (2016), 507--543, \href{https://arxiv.org/abs/1603.07930}{arXiv:1603.07930}.

\bibitem{GSL_Gallipoli15}
Gubbiotti G., Scimiterna C., Levi D., Linearizability and a fake Lax pair for a
 nonlinear nonautonomous quad-graph equation consistent around the cube,
 \href{https://doi.org/10.1134/S0040577916100068}{\textit{Theoret. and Math. Phys.}} \textbf{189} (2016), 1459--–1471.

\bibitem{GSL_Pavel}
Gubbiotti G., Scimiterna C., Levi D., On partial dif\/ferential and dif\/ference
 equations with symmetries depending on arbitrary functions, \href{https://doi.org/10.14311/AP.2016.56.0193}{\textit{Acta
 Polytechnica}} \textbf{56} (2016), 193--201, \href{https://arxiv.org/abs/1512.01967}{arXiv:1512.01967}.

\bibitem{GSL_symmetries}
Gubbiotti G., Scimiterna C., Levi D., The non-autonomous {Y}d{KN} equation and
 generalized symmetries of {B}oll equations, \href{https://doi.org/10.1063/1.4982747}{\textit{J.~Math. Phys.}}
 \textbf{58} (2017), 053507, 18~pages, \href{https://arxiv.org/abs/1510.07175}{arXiv:1510.07175}.

\bibitem{GSL_QV}
Gubbiotti G., Scimiterna C., Levi D., A two-periodic generalization of the
 {$Q_{\rm V}$} equation, \href{https://doi.org/10.1093/integr/xyx004}{\textit{J.~Integrable Syst.}} \textbf{2} (2017),
 xyx004, 13~pages.

\bibitem{GSY_DarbouxI}
Gubbiotti G., Yamilov R.I., Darboux integrability of trapezoidal {$H^4$} and
 {$H^4$} families of lattice equations~{I}: {F}irst integrals,
 \href{https://doi.org/10.1088/1751-8121/aa7fd9}{\textit{J.~Phys.~A: Math. Theor.}} \textbf{50} (2017), 345205, 26~pages,
 \href{https://arxiv.org/abs/1608.03506}{arXiv:1608.03506}.

\bibitem{Habibullin2005}
Habibullin I.T., Characteristic algebras of fully discrete hyperbolic type
 equations, \href{https://doi.org/10.3842/SIGMA.2005.023}{\textit{SIGMA}}
 \textbf{1} (2005), 023, 9~pages, \href{https://arxiv.org/abs/nlin.SI/0506027}{nlin.SI/0506027}.

\bibitem{Hay2009}
Hay M., A completeness study on discrete, {$2\times2$} {L}ax pairs,
 \href{https://doi.org/10.1063/1.3177197}{\textit{J.~Math. Phys.}} \textbf{50} (2009), 103516, 29~pages,
 \href{https://arxiv.org/abs/0806.3940}{arXiv:0806.3940}.

\bibitem{Hay2011}
Hay M., A completeness study on certain {$2\times 2$} {L}ax pairs including
 zero terms, \href{https://doi.org/10.3842/SIGMA.2011.089}{\textit{SIGMA}} \textbf{7} (2011), 089, 12~pages,
 \href{https://arxiv.org/abs/1104.0084}{arXiv:1104.0084}.

\bibitem{Hietarinta2004}
Hietarinta J., A new two-dimensional lattice model that is `consistent around a
 cube', \href{https://doi.org/10.1088/0305-4470/37/6/L01}{\textit{J.~Phys.~A: Math. Gen.}} \textbf{37} (2004), L67--L73,
 \href{https://arxiv.org/abs/nlin.SI/0311034}{nlin.SI/0311034}.

\bibitem{Hietarinta2005}
Hietarinta J., Searching for {CAC}-maps, \href{https://doi.org/10.2991/jnmp.2005.12.s2.16}{\textit{J.~Nonlinear Math. Phys.}}
 \textbf{12} (2005), suppl.~2, 223--230.

\bibitem{HietarintaBook}
Hietarinta J., Def\/initions and predictions of integrability for dif\/ference
 equations, in Symmetries and Integrability of Dif\/ference Equations (Beijing,
 2009), \href{https://doi.org/10.1017/CBO9780511997136.005}{\textit{London Math. Soc. Lecture Note Ser.}}, Vol.~381, Editors
 D.~Levi, P.~Olver, Z.~Thomova, P.~Winternitz, Cambridge University Press,
 Cambridge, 2011, 83--114.

\bibitem{HietarintaJoshiNijhoff2016}
Hietarinta J., Joshi N., Nijhof\/f F.W., Discrete systems and integrability,
 \href{https://doi.org/10.1017/CBO9781107337411}{\textit{Cambridge Texts in Applied Mathematics}}, Cambridge University Press,
 Cambridge, 2016.

\bibitem{HietarintaViallet2007}
Hietarinta J., Viallet C.M., Searching for integrable lattice maps using
 factorization, \href{https://doi.org/10.1088/1751-8113/40/42/S09}{\textit{J.~Phys.~A: Math. Theor.}} \textbf{40} (2007),
 12629--12643, \href{https://arxiv.org/abs/0705.1903}{arXiv:0705.1903}.

\bibitem{HietarintaViallet2012}
Hietarinta J., Viallet C.M., Weak {L}ax pairs for lattice equations,
 \href{https://doi.org/10.1088/0951-7715/25/7/1955}{\textit{Nonlinearity}} \textbf{25} (2012), 1955--1966, \href{https://arxiv.org/abs/1105.3329}{arXiv:1105.3329}.

\bibitem{Liouville1853}
Liouville J., Sur l'equation aux dif\/f\'erences partielles $\partial^2 \log
 \lambda /\partial u\partial v \pm \lambda /(aa^2)=0$, \textit{J.~Math. Pures
 Appl.} \textbf{18} (1853), 71--72.

\bibitem{Nijhoff2002}
Nijhof\/f F.W., Lax pair for the {A}dler (lattice {K}richever--{N}ovikov) system,
 \href{https://doi.org/10.1016/S0375-9601(02)00287-6}{\textit{Phys. Lett.~A}} \textbf{297} (2002), 49--58, \href{https://arxiv.org/abs/nlin.SI/0110027}{nlin.SI/0110027}.

\bibitem{Nijhoff2001}
Nijhof\/f F.W., Walker A.J., The discrete and continuous {P}ainlev\'e {VI}
 hierarchy and the {G}arnier systems, \href{https://doi.org/10.1017/S0017089501000106}{\textit{Glasg. Math.~J.}} \textbf{43A}
 (2001), 109--123, \href{https://arxiv.org/abs/nlin.SI/0001054}{nlin.SI/0001054}.

\bibitem{PapageorgiouNijhoff1990}
Papageorgiou V.G., Nijhof\/f F.W., Capel H.W., Integrable mappings and nonlinear
 integrable lattice equations, \href{https://doi.org/10.1016/0375-9601(90)90876-P}{\textit{Phys. Lett.~A}} \textbf{147} (1990),
 106--114.

\bibitem{QuispelCapel1991}
Quispel G.R.W., Capel H.W., Papageorgiou V.G., Nijhof\/f F.W., Integrable
 mappings derived from soliton equations, \href{https://doi.org/10.1016/0378-4371(91)90258-E}{\textit{Phys.~A}} \textbf{173}
 (1991), 243--266.

\bibitem{RobertsTran2017}
Roberts J.A.G., Tran D.T., Algebraic entropy of (integrable) lattice equations
 and their reductions, \href{https://arxiv.org/abs/1703.01069}{arXiv:1703.01069}.

\bibitem{Viallet2006}
Viallet C.M., Algebraic entropy for lattice equations,
 \href{https://arxiv.org/abs/math-ph/0609043}{math-ph/0609043}.

\bibitem{Viallet2009}
Viallet C.M., Integrable lattice maps: {$Q_{\rm V}$}, a rational version of
 {$Q_4$}, \href{https://doi.org/10.1017/S0017089508004874}{\textit{Glasg. Math.~J.}} \textbf{51} (2009), 157--163,
 \href{https://arxiv.org/abs/0802.0294}{arXiv:0802.0294}.

\bibitem{Xenitidis2009}
Xenitidis P.D., Papageorgiou V.G., Symmetries and integrability of discrete
 equations def\/ined on a black-white lattice, \href{https://doi.org/10.1088/1751-8113/42/45/454025}{\textit{J.~Phys.~A: Math. Theor.}}
 \textbf{42} (2009), 454025, 13~pages, \href{https://arxiv.org/abs/0903.3152}{arXiv:0903.3152}.

\bibitem{Yamilov2006}
Yamilov R., Symmetries as integrability criteria for dif\/ferential dif\/ference
 equations, \href{https://doi.org/10.1088/0305-4470/39/45/R01}{\textit{J.~Phys.~A: Math. Gen.}} \textbf{39} (2006), R541--R623.

\bibitem{ZhiberSokolov2011}
Zhiber A.V., Sokolov V.V., Exactly integrable hyperbolic equations of
 {L}iouville type, \href{https://doi.org/10.1070/rm2001v056n01ABEH000357}{\textit{Russian Math. Surveys}} \textbf{56} (2001), no.~1,
 61--101.

\end{thebibliography}
\end{document}